%% file: main.tex
\documentclass[11pt]{article}

\usepackage{0-preamble}

\begin{document}

\title{\textbf{Uncloneable Quantum Advice}}

\newlength{\authornamewidth}
\setlength{\authornamewidth}{0.25\textwidth}
\author{
	\makebox[\authornamewidth]{Anne Broadbent}\\
	\small University of Ottawa\\
	\footnotesize \texttt{abroadbe@uottawa.ca}
	\and
	\makebox[\authornamewidth]{Martti Karvonen}\\
	\small University of Ottawa\\
	\footnotesize \texttt{martti.karvonen@uottawa.ca}
	\and
	\makebox[\authornamewidth]{S\'ebastien Lord}\\
	\small University of Ottawa\\
	\footnotesize \texttt{sebastien.lord@uottawa.ca}
}

\date{ }

\maketitle

\input{1-abstract.tex}

\newpage

\setcounter{tocdepth}{3}
\tableofcontents

\newpage

\input{2-intro.tex}

\input{3-prelims.tex}

\input{4-ingenerable.tex}

\input{5-state.tex}

\input{6-advice.tex}

\appendix

\input{7-mlr.tex}

\input{8-proofs.tex}

\addcontentsline{toc}{section}{References}

\input{9-bibliography.bbl}

\end{document}

%% file: 1-abstract.tex
\begin{abstract}

The famous no-cloning principle has been shown recently to enable a
number of uncloneable functionalities. Here we address for the first
time \emph{unkeyed} quantum uncloneablity, via the study of a
complexity-theoretic tool that enables a computation, but that is
natively unkeyed: \emph{quantum advice}. Remarkably, this is an
application of the no-cloning principle in a context where the quantum
states of interest are \emph{not} chosen by a random process. We show
the unconditional existence of promise problems admitting
\emph{uncloneable quantum advice}, and the existence of languages with
uncloneable advice, assuming the feasibility of quantum copy-protecting
certain functions. Along the way, we note that state complexity classes,
introduced by Rosenthal and Yuen (ITCS 2022) --- which concern
the computational difficulty of synthesizing sequences of quantum
states --- can be naturally generalized to obtain state \emph{cloning}
complexity classes. We make initial observations on these classes,
notably obtaining a result analogous to the existence of undecidable
problems.

Our proof technique establishes the existence of \emph{ingenerable
sequences of finite bit strings} --- essentially meaning that they
cannot be generated by any uniform circuit family. We then prove a
generic result showing that the difficulty of
accomplishing a computational task on uniformly random inputs implies
its difficulty on any fixed, ingenerable sequence. We use this result to
derandomize quantum cryptographic games that
relate to cloning, and then incorporate a result of Kundu and Tan (arXiv
2022) to obtain uncloneable advice. Applying this two-step process to a
monogamy-of-entanglement game yields a promise problem with uncloneable
advice, and applying it to the quantum copy-protection of pseudorandom
functions with super-logarithmic output lengths yields a language with
uncloneable advice.

\end{abstract}

%% file: 2-intro.tex
\section{Introduction}
\label{sc:intro}

The no-cloning principle has been a key tenet of quantum information
theory from its very early days \cite{Die82,WZ82}.\footnote{%
	Diek, Wooters, and Zurek proved the no-cloning principle to argue
	that a proposed protocol for faster-than-light communication
	\cite{Her82} was unphysical. However, a proof by Parks \cite{Par70}
	of this principle already existed in the literature, albeit with
	no explicit connections to quantum information theory. See
	the work of Ortigoso \cite{Ort18} for more on the history of the
	no-cloning principle.}
This principle tells us that it is, in general, impossible to create a
perfect copy of an \emph{unknown} quantum state. In these works, the
state is unknown in the sense that it is selected uniformly at random
from two possibilities. Follow up results, \emph{e.g.}:
\cite{BH96,BDE+98,Wer98}, also give bounds on the ability to create
close-but-imperfect copies of states chosen from some set.

Quantum cryptography, which seeks to leverage quantum mechanical
phenomena to achieve cryptographic goals, has greatly benefited from
the no-cloning principle. Indeed, this principle establishes a stark
qualitative difference between classical information, which can be
perfectly copied, and quantum information, which cannot. In its most
generic cryptographic application, it implies that a malicious
eavesdropper cannot, in general, keep a perfect transcript of the
quantum communication between two honest parties. This idea can be
understood as being at the core of the security of many quantum
cryptographic schemes, such as prepare-and-measure quantum key
distribution protocols, \emph{e.g.}: \cite{BB84}, and quantum money
schemes, \emph{e.g.}: \cite{Wie83}.

This leads to a fruitful avenue of research in quantum cryptography,
occasionally known as \emph{uncloneable cryptography}, which can be
summarized by the following question: \emph{Which notions from classical
information processing can be made ``uncloneable'' by the addition of
quantum mechanics and its no-cloning principle?} One such example, as
initially set-out by Aaronson \cite{Aar09}, is the task of
copy-protecting a function or a program. Quantum copy-protection, for a
family of functions~$\mc{F}$, can be understood as the process of
encoding a member $f$ of $\mc{F}$ as a quantum state $\rho_f$ such that
the following are satisfied:
\begin{itemize}
	\item
		Correctness:
		There exists an honest procedure which, on input of $\rho_f$ and
		$x$, outputs $f(x)$.
	\item
		Security: It is infeasible to ``split'' $\rho_f$ into two
		quantum systems, both allowing the evaluation of $f$ using any,
		possibly malicious, procedures.
\end{itemize}
Clearly, this is impossible to achieve in a purely classical setting but
the no-cloning principle opens the door to achieving this in the quantum
setting. Since Aaronson's introduction of this idea, there has been a
flurry of works considering the question of quantum copy-protection,
\emph{e.g.}: \cite{ALL+21,CMP20arxiv}, and the closely related notion of
secure software leasing, \emph{e.g.}: \cite{ALP21,BJL+21}.

A key aspect of these existing constructions is that their
security guarantees only hold if the function which is encoded is chosen
at random from some larger set and not disclosed to the end-user. More
formally, the family $\mc{F}$ is often understood as being the set of
all maps~$f(k, \cdot)$ generated from a single keyed function $f :
\{0,1\}^\kappa \times \{0,1\}^d \to \{0,1\}^c$. The recipient of the
copy-protected program then receives the state $\rho_{f(k,\cdot)}$
for a random key $k$ which is unknown to them. More to the point,
current approaches to quantum copy-protection cannot be applied to
specific, unkeyed, functionalities. For example, it is unclear if, using
existing techniques, it would be possible to copy-protect a
\emph{specific} algorithm $A$ solving some fixed and known decision
problem $P$.

In this work, we initiate the study of copy-protecting unkeyed
functionalities, \emph{i.e.}:~copy-protecting fixed functions which are
not chosen at random from a larger set. We frame our main results as the
construction of \emph{uncloneable quantum advice} for certain promise
problems and languages. This framing occurs naturally since advice
can already be understood as a program helping a user to solve a
decision problem \cite{Wat09}.

One drawback of our work is that the resulting advice is either
uncomputable, or \emph{extremely} difficult to generate. We leave
addressing this issues to future work.

\paragraph{Organization of this work.}
We complete our introduction by giving a high-level overview of our
contributions in \cref{sc:contributions} and highlighting some open
questions in \cref{sc:questions}. We then proceed to review some
more technical preliminaries in \cref{sc:prelims} and formalize our
novel technical tool, ingenerable sequences, in \cref{sc:ingenerable}.
We then examine cloning complexity classes in
\cref{sc:cloning-complexity} before continuing with the main goal of
this work in \cref{sc:advice}. In this last section, we briefly review
quantum copy-protection, formally define uncloneable advice, and show
how to instantiate this notion for one class of promise problems and one
class of languages.

\subsection{Contributions}
\label{sc:contributions}

We review here the main contributions in this work. We note that, for
pedagogical reasons, this high-level review does not follow the same
order as our detailed presentation.

\subsubsection{Uncloneable Quantum Advice}
\label{sc:upoly-review}

Our first conceptual contribution is the formal definition of
\emph{uncloneable quantum advice}.\footnote{%
	There has been an ongoing implicit debate in the literature on the
	proper spelling of the word \textit{uncloneable}. Many authors,
	perhaps even a majority, prefer \textit{unclonable},
	\textit{i.e.}~they omit the \textit{e} from \textit{clone} before
	appending the \textit{-able} suffix. While the
	\textit{Oxford English Dictionary} recognizes both \textit{clonable}
	and \textit{cloneable} (notably, neither appears in this reference
	with the \textit{un-} prefix), it expresses a preference for
	\textit{clonable} \cite{OED-clonable}. However, to the best of our
	knowledge, the earliest work in the field of quantum cryptography
	which uses the word \textit{uncloneable} as an element of technical
	terminology \cite{Got03} keeps the \textit{e}. We will continue to
	follow this convention.

	Note that this debate is not limited to the field of quantum
	cryptography. See, for example, Maes' discussion on this subject in
	their textbook on physically uncloneable functions
	\cite[sec.~2.3.1]{Mae13}.} Defining this notion is the content of
	\cref{sc:advice-df}.

The prototypical complexity class with \emph{quantum} advice is denoted
$\mathbf{BQP}/\text{qpoly}$. This class contains precisely the decision
problems $P = (P_\text{yes}, P_\text{no})$ for which there exists a
fixed sequence of quantum states $(\rho_n)_{n \in \N}$ and an efficient (polynomial-time)
family of quantum circuits $(C_n)_{n \in \N}$ such that on input of a
problem instance $p \in P$ and the corresponding advice state
$\rho_\abs{p}$, the circuit~$C_\abs{p}$ accepts (respectively, rejects)
with probability at least $2/3$ if $p \in P_\text{yes}$ (respectively,
$p \in P_\text{no})$. The ``poly'' in ``qpoly'' simply signifies that
the advice states can contain at most a polynomial number of qubits. The
``q'', of course, stands for ``quantum''. Quantum advice was first
introduced in \cite{NY04} and this notion was further developed in
\cite{A05}.

An important aspect of advice is that it is a fundamentally
\emph{non-uniform} resource and that there is no restriction on the
computational complexity of generating it. In fact, the advice states
may even encode the solutions to uncomputable problems.

In a survey of quantum complexity theory \cite{Wat09}, Watrous
states that ``[q]uantum advice is a formal abstraction that addresses
this question: How powerful is quantum software?'' Indeed, the circuit
family $(C_n)_{n \in \N}$ from the previous paragraph can be understood
as simply outputting the single bit which results from ``running the
software'' encoded in the advice state $\rho_\abs{p}$ with the problem
instance $p$ as the input. More formally, the decision problem $P$
naturally defines a map~$P_\text{yes} \cup P_\text{no} \to \{0,1\}$
where a problem instance $p$ is mapped to $1$ if it is an element of
$P_\text{yes}$ and mapped to $0$ if it is an element of
$P_\text{no}$. Thus, $\rho_n$ can be seen as a quantum
program which allows a holder to correctly evaluate this map on all
inputs of length $n$.

Since we have established a close parallel between quantum advice and
quantum programs, it is natural that our definition of
\emph{uncloneable} quantum advice should follow closely the definition
of \emph{uncloneable} quantum programs, as formalized by quantum
copy-protection.

More precisely, we will say that the quantum advice states
$(\rho_n)_{n \in \N}$ for some decision problem~$P$ are \emph{uncloneable} if
no efficient adversary $A$ can split the advice state $\rho_n$ between
two other adversaries~$B$ and $C$ such that both of them can
simultaneously and correctly solve independently sampled problem
instances $p_B$ and $p_C$ of length $n$ with more than a negligible
advantage over $\frac{1}{2}$. This is precisely the security requirement
that a copy-protection scheme would aim to provide for the
map~$P_\text{yes} \cup P_\text{no} \to \{0,1\}$ described above. We
specify here that the problem instances $p_B$ and~$p_C$ are sampled
uniformly at random among all \emph{yes} instances with probability
$\frac{1}{2}$ and uniformly at random among all \emph{no} instances with
probability $\frac{1}{2}$. Formalizing this, as we have done in
\cref{df:neglqp/upoly} of \cref{sc:advice-df}, yields a complexity class
which we denote $\mathbf{neglQP}/\text{upoly}$, where the ``u'' denotes
the fact that the advice is uncloneable. The class $\mathbf{neglQP}$,
for its part, is precisely the class of problems which can be solved
with negligible error in quantum polynomial time. As we now sketch, we
define $\mathbf{neglQP}/\text{upoly}$ and not simply
$\mathbf{BQP}/\text{upoly}$ because it is possible that these two
classes are distinct.

By the standard amplification technique of repeating a
computation a polynomial number of times and taking a majority vote on
the result, we see that $\mathbf{neglQP} = \mathbf{BQP}$. In fact, this
same argument also yields
$\mathbf{neglQP}/\text{qpoly} = \mathbf{BQP}/\text{qpoly}$, provided
that polynomially many copies of the basic advice state are provided.
However, it is unclear if $\mathbf{neglQP}/\text{upoly} =
\mathbf{BQP}/\text{upoly}$, since the basic amplification technique previously
sketched fails to yield this equality. Indeed, if polynomially many
copies of an advice state are provided, half can be given to one party
and half to another. This breaks the uncloneability guarantee. It is
also unclear if the other common amplification technique, given by
Marriot and Watrous \cite{MW05}, yields this equality between complexity
classes. Indeed, their technique does not require sending many copies of
the advice state, but it does require changing the advice and there
is no \emph{a priori} guarantee that the uncloneability guarantee is
preserved under this transformation.

\subsubsection{Cloning Complexity Classes}
\label{sc:cloning-complexity-review}

As a second conceptual contribution, we define \emph{cloning
complexity classes.}
The uncloneability guarantee for advice states discussed in the previous
section is \emph{not} simply that it is impossible to efficiently
implement the transformation $\rho_n \mapsto \rho_n \tensor \rho_n$.
Indeed, if it was possible to efficiently transform the advice state
$\rho_n$ into $\sigma_n \tensor \sigma_n$, where $\sigma_n \not= \rho_n$
but where each $\sigma_n$ would still allow malicious users to correctly
solve the decision problem, then $\rho_n$ fails to be uncloneable
advice in the sense described in \cref{sc:upoly-review}. In fact, this
distinction between the (in)ability to copy a quantum state and
the (in)ability of copying its underlying operational capabilities
is at the source of many interesting aspects and challenges of
uncloneable cryptography.

Nonetheless, it is evident that the infeasibility of implementing the
transformation $\rho_n \mapsto \rho_n \tensor \rho_n$ is a necessary
condition for the sequence $(\rho_n)_{n \in \N}$ to be uncloneable
advice for some decision problem. This naturally leads us to consider
questions of the following form, with a particular interest in cases
where the answer is negative:

\begin{quote}
		\emph{%
			Given a sequence $\left(\rho_n\right)_{n \in \N}$
			of fixed quantum states, is there a sequence of circuits
			$(C_n)_{n \in \N}$, satisfying some given
			computational constraints, such that $
				C_n(\rho_n)
				\approx
				\rho_n \tensor \rho_n
			$ for all $n$?}
\end{quote}

Our insight here, covered in \cref{sc:cloning-complexity}, is that this
type of question generalizes those captured by the \emph{state
complexity classes} recently studied by Metger, Rosenthal, and Yuen
\cite{RY21arxiv,MY23arxiv}. To take an explicit example of such a
class presented in these works, a sequence of states $(\rho_n)_{n\in\N}$
is in the state complexity class $\mathbf{statePSPACE}$ if and only
if there exists a uniform polynomial-space circuit family $(C_n)_{n \in
\N}$ such that $C_n$ outputs a state overwhelmingly close to $\rho_n$ on
the empty input. The questions of the form above can be recast to form
\emph{cloning complexity classes}. These classes can then be understood
as the ``$1 \to 2$'' generalization of the ``$0 \to 1$'' state
complexity classes. We formalize this in \cref{df:state-complexity}.

We note in passing that existing works on the no-cloning principle do not
answer cloning complexity questions of the form presented above as
they all consider the difficulty of cloning a state sampled at random
from some larger set. Our cloning complexity classes explicitly
remove this randomness from the cloning task.

The crux of this work is not focused on cloning complexity classes, but
the tools and techniques we developed to instantiate uncloneable quantum
advice do yield two interesting results on these classes and how they
relate to state complexity classes:
\begin{itemize}
	\item
				There exists a sequence of fixed quantum states which cannot be
		\emph{generated} by any uniform circuit family, but which can be
		cloned by an efficient circuit family (\Cref{th:hard-generate-easy-clone}).
	\item
		There exists a sequence of fixed quantum states which cannot be
		\emph{cloned} by any uniform circuit family, even those
		implementing arbitrarily large computations (\Cref{th:hard-generate-hard-clone}).
\end{itemize}
These two results, with an additional generic result showing that
cloning cannot be more difficult than generating
(\cref{th:clone-gen-bound}) are visualized in \cref{fg:gen-vs-clone}
which locates sequences of states with respect to two axes: one denoting
the difficulty to generate them and one denoting the difficulty to clone
them. Moreover, \cref{th:hard-generate-hard-clone} can be interpreted as
establishing a distinction between cloneable and uncloneable sequences
of states, similar to how maps can be partitioned between those which
are computable and those which are not. We leave the study of
complexity-theoretic analogues (\emph{e.g.}: establishing distinctions
between efficiently cloneable and uncloneable sequences) for future~work.

\begin{figure}
	\begin{center}
	\input{fig-generating-cloning.tex}
	\end{center}
	\caption{\label{fg:gen-vs-clone}%
		An annotated and informal representation of the space of
		sequences of quantum states contrasting the difficulty to
		generate and to clone a given sequence. Ingenerable
		sequences, which are novel to this work, are discussed in more
		details in \cref{sc:ingen-review}.}
\end{figure}
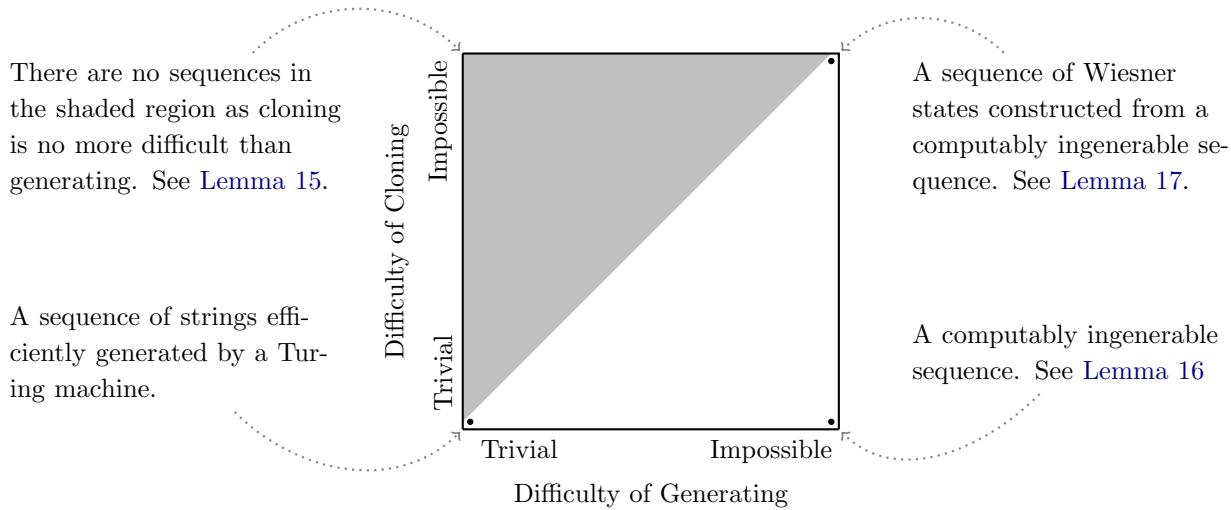

We emphasize that the development and study of cloning complexity
classes are not the main goal of this work; however, they do offer an
interesting way for us to connect some aspects of our work with existing
results. We also believe that the study of cloning complexity
classes could be of central interest and importance for further
developments in uncloneable cryptography.

\subsubsection{Ingenerable Sequences of Bit Strings}
\label{sc:ingen-review}

The main technical tool which we develop in this work are
\emph{ingenerable sequences} of bit strings. This is the content of
\cref{sc:ingenerable}.

In a nutshell, a sequence $(s_n)_{n \in \N}$ of bit strings
$s_n \in \{0,1\}^*$ is computably (respectively, exponential-time)
ingenerable if every uniform (respectively, exponential-time) circuit
family $(C_n)_{n \in \N}$ eventually always fails at generating the
elements of the sequence with sufficiently large probability. More
precisely, there exists a polynomial $p$ such that for every uniform
(respectively, exponential-time)
family of quantum circuits $(C_n)_{n\in\N}$ the quantity $\bra{s_n}
C_n(\varepsilon)\ket{s_n}$, where $\varepsilon$ denotes the empty input,
is eventually strictly smaller than $p(n) \cdot 2^{-\abs{s_n}}$. We
emphasize that the polynomial $p$ is universal, in the sense that it
does not depend on the circuit family $(C_n)_{n \in \N}$. However,
the value of $n$ at which point the inequality begins to be satisfied
can change from one circuit family to the next.

A simple counting argument demonstrates that such sequences exist
unconditionally for many desired lengths of the component strings.
While computably ingenerable sequence are by their nature uncomputable,
we also show that certain exponential-time ingenerable sequence can be
computed by classical deterministic Turing machines in triple
exponential time. Our counting argument is constructive, in the sense
that it uniquely describes an ingenerable sequence, even in the case of
computably ingenerable sequences where we cannot actually compute their
elements. Further, we give an explicit classical deterministic algorithm
running in triple exponential time which computes an exponential-time
ingenerable sequence. These results are stated explicitly in
\cref{th:ingenerable}.

Note that ingenerable sequences are purely classical information and so
they can trivially be copied, this immediately yields
\cref{th:hard-generate-easy-clone} which was discussed previously in
\cref{sc:cloning-complexity-review}.

Ingenerable sequences are useful for us in this work due to a generic
result, \cref{th:random-to-ingenerable}, relating the difficulty of a
computational task on uniformly random inputs to its difficulty on
inputs picked from an ingenerable sequence. Specifically, we show that
if a uniform (respectively, exponential-time) circuit family
$(C_n)_{n \in \N}$, which takes as input a string of at least
super-logarithmic and at most polynomial length, outputs the bit $1$
with negligible probability on uniformly random inputs, then it must
output $1$ with negligible probability on inputs fixed from a computably
(respectively, exponential-time) ingenerable sequence
$(s_n)_{n \in \N}$. Even more technically, we show that
\begin{equation}
	\E_{x} \bra{1} C_n(x) \ket{1} = \eta(n)
	\implies
	\bra{1} C_n(s_n) \ket{1} = \eta'(n)
\end{equation}
for two possibly distinct negligible functions $\eta$ and $\eta'$,
provided that $(s_n)_{n \in \N}$ is a suitable ingenerable sequence.
We emphasize here that this holds for
\emph{all} uniform (respectively, exponential-time) circuit families
$(C_n)_{n \in \N}$ while the ingenerable sequence $(s_n)_{n \in \N}$
remains fixed and independent of the circuit family.

This result allows us to derandomize certain cryptographic games related
to cloning. For example, in \cref{th:money-derandomized}, we demonstrate
the existence of a fixed sequence of Wiesner states\footnote{%
	A Wiesner state, introduced to quantum cryptography in \cite{Wie83},
	is any tensor product of the single-qubit computational basis
	states, $\ket{0}$ and $\ket{1}$, or of their Hadamard conjugates,
	$\ket{+} = \frac{\ket{0} + \ket{1}}{\sqrt{2}}$ and
	$\ket{-} = \frac{\ket{0} - \ket{1}}{\sqrt{2}}$. For two strings $x,
	\theta$ of the same length, we let $\ket{x^\theta} =
	H^{\theta_1}\ket{x_1} \tensor \cdots \tensor H^{\theta_n}
	\ket{x_n}$, where $H$ is the single-qubit Hadamard operator, denote
	a Wiesner state.}
which cannot be cloned by any uniform circuit family. Specifically, we
show that any uniform circuit family attempting to copy this fixed
sequence will generate states having negligible fidelity with perfect
copies of the original ones. This follows from applying our above result to the
theorem of Molina, Vidick, and Watrous showing that the ability to clone
a uniformly random Wiesner state is negligible in the length of this
state \cite{MVW13}. It then suffices to replace the uniformly random
Wiesner state with one which is determined by a computably ingenerable sequence.

\subsubsection{A Promise Problem with Uncloneable Advice from MoE games}
\label{sc:problem-review}

We now present our first main result: a promise problem $P =
(P_\text{yes}, P_\text{no})$
which admits uncloneable advice. This is the content of
\cref{sc:promise}. Recall that for a promise problem we do
not require~$P_\text{yes}$ and $P_\text{no}$ to partition $\{0,1\}^*$: some strings may
not be an element of either sets, but we are promised to only get
problem instances from $P_\text{yes}\cup P_\text{no}$.

First, we consider an exponential-time ingenerable sequence
$(x_n)_{n\in\N}$ and define the set $P_\text{yes}$ to be all strings $y$ where
the inner product $x_n \cdot y$ is $1$. Similarly, the set
$P_\text{no}$ will be all strings which have an inner product of $0$
with the elements of the ingenerable sequence. Evidently, this problem
cannot be efficiently solved without some advice on the sequence
$(x_n)_{n \in \N}$.
Indeed, if we were able to solve this problem without any advice, we
would be able to generate the ingenerable sequence. Now, the obvious
advice to directly give is $x_n$ as this would certainly allow an honest
user to solve all problem instances. However, this advice is trivially
cloneable as it is a classical piece of information. So, instead, we
will give as advice the Wiesner state $\ket*{x_n^{\theta_n}}$ for a
string $\theta_n$ which will be determined shortly. We know from prior
works on monogamy-of-entanglement (MoE) games \cite{TFKW13} that it is
infeasible for an adversary to split a random Wiesner state
$\ket{x^\theta}$ in such a way that two parties can recover the string
$x$ if they know $\theta$, provided that the initial splitting adversary
does not know~$\theta$. We recall and formalize this result as
\cref{th:tfkw} in our detailed exposition. In some sense, the state~$\ket{x^\theta}$ encodes information on the string $x$ in an uncloneable
way, which is what we need. However, three interrelated issues now arise.

The first problem is that if $\theta_n$ is fixed --- and also
$x_n$ for that matter --- how do we make sure that an adversary is
not able to create copies of the fixed state
$\ket*{x_n^{\theta_n}}$? \Cref{th:tfkw} considers
the case of uniformly random Wiesner states, not a fixed sequence of
such states. The solutio is to also take $\theta_n$
from an exponential-time ingenerable sequence. For technical reasons, which can be
roughly understood as making sure that $x_n$ and $\theta_n$
are not too correlated, we will consider a single exponential-time
ingenerable sequence $(s_n)_{n \in \N}$ and parse each string
$s_n$ as the concatenation $(x_n, \theta_n)$ of two
strings of length~$n$. We can then apply our generic derandomization
theorem for ingenerable sequences, \cref{th:random-to-ingenerable}
discussed previously in \cref{sc:ingen-review}, to the
MoE game described above to show that
no triplet of polynomial-time (or, even, exponential-time) adversaries
can have a non-negligible advantage at
winning this game when challenged with the fixed state
$\ket*{x_n^{\theta_n}}$.

Second, we have argued above that polynomial-time adversaries cannot win the
MoE  game of \cite{TFKW13} with a non-negligible
probability when challenged with the fixed states
$\ket*{x_n^{\theta_n}}$. However, this game requires both
guessing adversaries to correctly determine the complete
string~$x_n$ to win. Solving the problem $P$ only requires them to
determine the inner product $x_n\cdot y$ for some string $y$, a
task which may be easier. How can we bridge this gap? It suffices to use
the recent result of Kundu and Tan \cite{KT22arxiv1}, formalized in
\cref{th:kundu-tan}, which states that a negligible probability of both
guessing adversaries determining $x_n$ implies a that the probability
that both guessing adversaries determining~$x_n \cdot y$, for independent
and uniformly random values of $y$, is at most negligibly greater than
$\frac{1}{2}$.

Finally, how does the honest user of the advice know $\theta_n$?
We have sketched above our argument as to why the advice state
$\ket*{x_n^{\theta_n}}$ is uncloneable, but not yet as to
how it is useful. The naive way to use this state as advice for
this problem is to measure it in the $\theta_n$-Wiesner basis
to obtain the string~$x_n$. This requires the knowledge of
$\theta_n$, something which is ingenerable, and not known to the
honest user so far. To overcome this issue, we will modify our problem
to be a \emph{promise} problem. The promise is that every problem
instance $y$ is prefixed with the fixed string $\theta_\abs{y}$. In
other words, we are giving $\theta_n$ as part of the problem
instance. Because of this, $\theta_n$ is not accessible to an
adversary attempting to split the advice state before knowing a problem
instance.

In summary, we start with the monogamy-of-entanglement game of
\cite{TFKW13}, derandomize it via exponential-time ingenerable
sequences, and then apply Kundu and Tan's lemma.
We then argue that a promise problem with uncloneable advice can be
extracted from this construction.

We further note that this construction can be 
generalized by replacing the monogamy-of-entanglement game with any
sufficiently secure uncloneable encryption scheme \cite{BL20} where the
key $k_n$ used and the message $m_n$ to be be encoded play the role of
$\theta_n$ and $x_n$, respectively, and are also determined by an
ingenerable sequence.
We note that the resulting construction, up to the use of ingenerable
sequences to derandomize the scheme, is conceptually similar to Kundu
and Tan's construction of uncloneable encryption for a single-bit
message with \emph{variable keys} \cite{KT22arxiv1}.

As a last remark, we highlight that we use \emph{exponential-time}
ingenerable sequences in this construction. As previously mentioned, we
show in this work that there exist such sequences which can be computed
in triple exponential time. Thus, if $P$ is constructed from an
exponential-time ingenerable sequence which can be computed in triple
exponential time, then $P$ can be enumerated by a classical
deterministic Turing machine. Hence, our result is more than merely
existential: we show how to construct $P$ in such a way that the
\emph{yes} and \emph{no} instances can be enumerated and we explicitly
give an algorithm which does this enumeration.\footnote{%
	Strictly speaking, we give a deterministic classical algorithm
	in the proof of \cref{th:ingenerable} computing in
	triple exponential time the elements of an exponential-time
	ingenerable sequence. It is then  easy to see how to construct
	an algorithm enumerating $P_\text{yes}$ and $P_\text{no}$ when given
	an algorithm which computes $\theta_n$ and $x_n$ for any $n \in\N$.}

\subsubsection{A Language with Uncloneable Advice from Copy-Protected PRFs}
\label{sc:language-review}

Our second main result is the construction of a \emph{language} with
uncloneable advice. Here, we \emph{do} require $P_\text{yes}$ and
$P_\text{no}$ to partition $\{0,1\}^*$. This is the content of
\cref{sc:language} and \cref{sc:language-2}

To construct this language, we leverage even more explicitly the
connection between advice and programs: Our starting building block is the assumption that it is possible to copy-protect pseudorandom
functions with super-logarithmic outputs. This can be achieved under certain assumptions \cite{CLLZ21}.

At a high level, we can understand the problem instances of the
promise problem described in the previous section as pairs
$(\theta, y)$ where the string $\theta$ was used as an ``input'' for the
``program'' encoded by the Wiesner state $\ket{x^\theta}$. Evaluating a
Wiesner state in this paradigm then consists of measuring it in the
$\theta$-Wiesner basis and outputting the resulting $x$. The reason we
obtained a promise problem, and not a language, is that this ``program''
has a deterministic answer on only one input:~$\theta$.

Copy-protected programs do not suffer this limitation of only having one
``good input''. If $\rho_f$ is the copy-protected program of the
function $f$, then evaluating $\rho_f$ on any input $x$ will yield
$f(x)$ with high probability. Thus, we can reuse the same template as in
the previous promise problem: problem instances are pairs $(x, y)$
where $x$ is treated as an input to a program state $\rho_f$ to obtain a
string $f(x)$ with near certainty. We then take the inner product
$f(x) \cdot y$ to determine if $(x,y)$ is an element of the language.
However, since $\rho_f$ can be properly evaluated on \emph{any} input,
unlike the Wiesner state ``program'', we need not promise that the input
$x$ is some particular well-behaved string: it can be arbitrary. Thus,
we can obtain a language, and not merely a promise problem from this
type of construction.

Here, as for our promise problem, the result of Kundu and Tan
(\cref{th:kundu-tan}) is used to relate the difficulty of two
guessers independently and simultaneously determining $f(x) \cdot y$ to
their difficulty of determining $f(x)$. We also make use of
exponential-time ingenerable sequences to derandomize the function which is
copy-protected all the while maintaining the uncloneability of the
resulting program.

\subsection{Open Questions}
\label{sc:questions}

This work represents the first attempt to study uncloneable advice and
shows that this concept is, in principle, achievable. Beyond
identifying other languages and promise problems which admit uncloneable
advice, this work leaves many other open questions. We highlight two
such questions, mentioned in the previous discussion, which we believe
are of particular interest.

\paragraph{Can we say more on the complexity of cloning quantum states?}
As illustrated in \cref{fg:gen-vs-clone} and discussed in
\cref{sc:cloning-complexity-review}, our work establishes the
existence of sequences of states at the simultaneous extremes of cloning
and generating complexity.

Are there scenarios between the extremes sketched above? In other words,
how densely populated is the lower-right triangle of
\cref{fg:gen-vs-clone}? Our results imply the existence of a sequence of
states that are impossible to clone in polynomial time --- or indeed in
exponential time --- but which can be
perfectly generated in
triple exponential time. Can we close this gap? For example, can we find
a fixed sequence which cannot be cloned in polynomial time, but can be
generated in exponential time?

In a sense, our work initiates a theory of ``computability of cloning''
by showing a separation between sequences of states which can and cannot
be cloned. We believe that further and more fine-grained work in this
direction would help establish a theory on the complexity of cloning
quantum states. We hope that such a theory would then provide useful
tools for uncloneable cryptography and other aspects of quantum
information theory.

\paragraph{Can we ``boost'' the success probability of an honest
party using uncloneable advice while maintaining the uncloneability
property?} As discussed in~\cref{sc:upoly-review},
the standard error reduction technique of giving multiple copies of the
advice state evidently fails to maintain the required uncloneability
property of the initial state. Furthermore, the well-known amplification
technique of Marriott and Watrous \cite{MW05}, developed in the context
of quantum \emph{proofs} and the $\mathbf{QMA}$ complexity class, does
not appear to be directly applicable to the context of uncloneable
advice.
We highlight two obstacles in applying the result of Marriott and
Watrous to uncloneable advice.
First, the Marriott-Watrous technique requires changing the proof state
and it is unclear if this modification to the state would preserve the
required uncloneability property. Second, and perhaps more fatal, is
that directly applying the Marriott-Watrous technique to quantum
\emph{advice}, instead of quantum \emph{proofs}, does not in general
yield quantum \emph{advice} as the required change to the quantum state
depends on problem instance for which we wish to increase the
probability of success. This dependency is acceptable in the
context of \emph{proofs} as these may depend on the problem instance,
but is unacceptable in the context of \emph{advice} which must be
independent of the problem instance.

Addressing this question would elucidate the relation between the
complexity classes $\mathbf{neglQP}/\text{upoly}$ and
$\mathbf{BQP}/\text{upoly}$. Moreover, finding a generic boosting
technique for uncloneable advice could have wider repercussions in
uncloneable cryptography by effectively lowering the necessary threshold
for correctness.

\subsection{Acknowledgements}
This work was supported by the Air Force Office of Scientific Research under award number
FA9550-20-1-0375, Canada’s NSERC, and the University of Ottawa’s Research Chairs program.

%% file: fig-generating-cloning.tex
\begin{tikzpicture}

	\node (tt) at (0,0) {};
	\node (th) at (0,5) {};
	\node (ht) at (5,0) {};
	\node (hh) at (5,5) {};

	\fill[lightgray]
	($(tt) + (0,0.1)$) -- ($(th)$) -- ($(hh) - (0.1,0)$) -- cycle;

	\draw[thick]
	($(tt)$) -- ($(ht)$)
	node [pos=0.50, below=1.5em, sloped] (TextNode)
		{\small Difficulty of Generating}
	node [pos=0.15, below,       sloped] (TextNode) {\small Trivial}
	node [pos=0.82, below,       sloped] (TextNode) {\small Impossible};

	\draw[thick]
	($(tt)$) -- ($(th)$)
	node [pos=0.50, above=1.5em, sloped] (TextNode)
		{\small Difficulty of Cloning}
	node [pos=0.15, above,       sloped] (TextNode) {\small Trivial}
	node [pos=0.82, above,       sloped] (TextNode) {\small Impossible};

	\draw[thick] ($(ht)$) -- ($(hh)$) -- ($(th)$);
	
	\node[text width = 1.75in] (wie-txt) at ($(hh) + (3.2,-1)$)
	{\small A sequence of Wiesner states constructed from a computably ingenerable
	sequence. See \cref{th:hard-generate-hard-clone}.};
	\node (wie-seq) at ($(hh) - (0.1,0.1)$) {};
	\draw[fill=black]  ($(hh) - (0.1,0.1)$) circle (1pt);
	\draw[thick,->,gray,dotted] (wie-txt) to [out=135,in=045] (wie-seq);
	
	\node[text width = 1.75in] (hlt-txt) at ($(ht) + (3.2,1)$)
	{\small A computably ingenerable sequence. See
	\cref{th:hard-generate-easy-clone}};
	\node (hlt-seq) at ($(ht) - (0.1,-0.1)$) {};
	\draw[fill=black]  ($(hlt-seq)$) circle (1pt);
	\draw[thick,->,gray,dotted] (hlt-txt) to [out=225,in=315] (hlt-seq);

	\node[text width = 1.75in] (smp-txt) at ($(tt)+(-3.8,1)$)
	{\small A sequence of strings efficiently generated by a
	Turing machine.};
	\node (smp-seq) at ($(tt) + (0.1,0.1)$) {};
	\draw[fill=black]  ($(smp-seq)$) circle (1pt);
	\draw[thick,->,gray,dotted] (smp-txt) to [out=315,in=225] (smp-seq);

	\node[text width = 1.75in] (nan-txt) at ($(th) + (-3.8,-1)$)
	{\small There are no sequences in the shaded region as cloning is no
	more difficult than generating. See \cref{th:state-inclusion}.};
	\node (nan-seq) at ($(th) + (0.1,-0.1)$) {};
	\draw[thick,->,gray,dotted] (nan-txt) to [out=045,in=135] (nan-seq);

\end{tikzpicture}

%% file: 3-prelims.tex
\section{Preliminaries}
\label{sc:prelims}

\subsection{Notation and Terminology}

\paragraph{Sets, bit strings, miscellaneous.}
We let $\N$ denote the non-negative integers, $\R$ denote the real
numbers,~$\R^+$ denote the strictly positive real numbers, and
$\R_0^+ = \R^+ \cup\{0\}$, denote the non-negative real numbers. For any
$n \in \N$, we let $[n]$ denote the set $\{i\in\N\;:\; 1\leq i\leq n\}$.
For any set $\mc{X}$, we let~$\mc{P}(\mc{X})$ denote the power set of
$\mc{X}$.

For all $n \in \N$, let $\{0,1\}^n$ denote the set of bit strings of
length $n$ and let~$\{0,1\}^* = \cup_{n \in \N} \{0,1\}^n$ denote the
set of all bit strings of finite length. We let
$\abs{x}$ denote the length of the bit string $x$ and we denote the
empty bit string by $\varepsilon$, which is to say that
$\{0,1\}^0 = \{\varepsilon\}$. If $x$ and $y$ are two strings of the
same length, we let $x \cdot y$ denote their inner product modulo $2$.
The result of concatenating two bit strings $x \in \{0,1\}^n$ and $y \in
\{0,1\}^m$ is denoted by $(x, y) \in \{0,1\}^{n+m}$, using implicitly the
isomorphism $\{0,1\}^n \times \{0,1\}^m \cong \{0,1\}^{n+m}$. For all
$n \in \N$ we let $0^n$ and $1^n$ denote the all-zero and all-one bit
strings in $\{0,1\}^n$, respectively.

Logarithms are always taken in base $2$.

\paragraph{Probability theory.}
If $\mc{X}$ is a set, we use $\E_{x \gets \mc{X}}f(x)$ to denote the
expectation of $f$ when $x$ is sampled uniformly at random from
$\mc{X}$. If $\alpha$ is a random variable, we use $x \gets \alpha$ to
denote that $x$ is sampled according to $\alpha$.

\paragraph{Asymptotic analysis.}
We will use the $\omega$, $\Omega$, and $\mc{O}$ notation to
characterize the asymptotic behaviour of functions from $\N$ to $\R^+_0$
\cite{Knu76}.\footnote{%
	Note that we follow Brassard \cite{Bra85} by always treating
	$\omega(f)$, $\Omega(f)$, and $\mc{O}(g)$ as sets and avoiding
	``one-way equations''.}

Given two functions $f, g : \N \to \R$, we say that $f$ is
\emph{eventually smaller} than $g$ if there exists a
$n_0 \in \N$ such that $n \geq n_0 \implies f(n) \leq g(n)$.
A function $f : \N \to \R$ is said to be \emph{negligible} if it is
eventually smaller than $n \mapsto n^{-c}$ for all $c \in \N$. See, for
example, \cite{Bel02}. Recall that $2^{-f}$ is negligible if and only if
$f \in \omega(\log)$ and that if $\eta$ is negligible, $f$ is non-decreasing, and
$f \in \Omega(n^r)$ for some $r \in \R^+$, then $\eta \circ f$ is also
negligible.

\paragraph{Turing machines and their run times.}
We only consider deterministic Turing machines in this work and we refer
to standard references, such as \cite{AB09} or \cite{LV19}, for more
details. We do, however, establish some notation and terminology.

A Turing machine $\ttt{T}$ takes as input a string $x \in \{0,1\}^*$ and
either (i) halts and produces as output a string $y \in \{0,1\}^*$, or
(ii) never halts, in which case it does not produce an output. We write
$\ttt{T}(x) = y$ to denote the fact that the Turing machine $\ttt{T}$
halts and produces the output $y$ on input $x$. Turing machines perform
their calculations in a series of discrete steps. We say that a Turing
machine $\ttt{T}$ runs in
polynomial-time if and only if there exists a polynomial $p : \N \to \N$
such that on input of any string $x$, the machine $\ttt{T}$ halts after
at most $p(\abs{x})$ steps. For this paper, we identify efficient
(quantum) computations with polynomial-time (quantum) computations. We
also say that a Turing machine runs in exponential, double exponential,
or triple exponential-time if there exists a polynomial $p : \N \to \N$
such that on input of any string $x$, the machine halts after at most
$2^{p(\abs{x})}$, $2^{2^{p(\abs{x}})}$, or $2^{2^{2^{p(\abs{x})}}}$ steps,
respectively.

We let $\mc{T}$ denote the set of all Turing machines. Note that
$\mc{T}$ is countable and so there exists a bijective map
$\tau : \N \to \mc{T}$.

Finally, for maps $f : \N \to \N$, we say that a Turing machine
$\ttt{T}$ computes $f$, perhaps within some time bound, if on input of
$1^n$, i.e.~the bit string composed of $n$ instances of $1$,
it outputs $1^{f(n)}$. In particular, if $\ttt{T}$ can computes
$f : \N \to \N$ and $\ttt{T}(1^n)$ halts within $T(n)$ steps, then $f(n)
\leq T(n)$ as $\ttt{T}$ can write at most one symbol per step.

\paragraph{Quantum mechanics, Hilbert spaces, operators, and channels.}
Quantum mechanics, to the extent needed for this work, is a theory
concerning linear operators on finite-dimensional complex Hilbert
spaces. We refer the reader to the standard introductory textbooks of
Watrous \cite{Wat18} and Nielsen and Chuang \cite{NC10} for more
details. The set of linear and unitary operators on a given Hilbert
space $\mc{H}$ are denoted $\mc{L}(\mc{H})$ and $\mc{U}(\mc{H})$,
respectively. We recall that a \emph{channel} is a completely-positive
trace preserving map $\Phi : \mc{L}(\mc{H}) \to \mc{L}(\mc{H}')$ and
that channels precisely coincide with the set of all possible
transformations that a quantum system may undergo over a finite time
period. We also recall that by a \emph{state}, we mean a density
operator.

In this work, we assume that all Hilbert spaces are finite tensor
products of $\C^2$, which is to say that we only consider spaces
representing finite numbers of qubits. We suppose that
$\left(\C^2\right)^n$ admits
$\{\ket{s}\}_{s \in \{0,1\}^n}$ as an orthonormal basis and, when
appropriate, we identify a string $s$ with its corresponding density
operator $\ketbra{s}$.

We recall that the trace distance between two states is given by
$\frac{1}{2}\norm{\rho - \sigma}_\text{Tr}$ where
$\norm{\cdot}_{\Tr}$ denotes the trace norm, which is also the
$p=1$ case of the more general Schatten $p$-norms. We will denote the
completely bounded trace norm on channels, also known as the diamond
norm, by $\norm{\cdot}_\diamond$.

\paragraph{Gate sets and quantum circuits.}
\label{sc:prelims-circuits}
A \emph{gate set} $\mc{G}$ is a finite set of channels. We assume that
all elements of a gate set are channels of the form
$L \mapsto U L U^\dag$ for a unitary operator
$U \in \mc{U}({\C^2}^{\tensor n})$ acting on $n$ qubits whose whose
entries in the computational basis are algebraic. Note that $n$ need not
be identical across unitaries. There are two possible exceptions:
$\mc{G}$ may also include the state preparation map
$P : \C \to \mc{L}(\C^2)$, given by $c \mapsto c \ketbra{0}$, and the
single-qubit trace $\Tr : \mc{L}(\C^2) \to \C$.

For a gate set $\mc{G}$, we let $\langle \mc{G} \rangle$ denote the set
of channels generated by $\mc{G}$. More formally,
$\langle \mc{G} \rangle$ is precisely the set of all channels $\Psi$
which can be expressed as
\begin{equation}
\label{eq:circuit}
	\Psi
	=
	\Circ_{i = 1}^{d}
		\Tensor_{j = 1}^{w_i}
		\Psi_{i,j}
\end{equation}
for some non-negative integers $d, w_1, \ldots, w_d \in \N$ and channels
$\Psi_{i,j} \in \mc{G}$. We call an expression $C$ of the form presented
on the right-hand side of \cref{eq:circuit} a \emph{circuit} built from
$\mc{G}$ implementing the map $\Psi$ and we denote the set of all
circuits built from $\mc{G}$ by $\langle \mc{G} \rangle_\text{c}$.
Formally, we distinguish between different circuits even if they
implement in the same channel but, when convenient, we identify a
circuit with the channel it implements. A \emph{circuit family}
$C = (C_n)_{n \in \N}$ is a sequence of circuits built from the same
gate set. We occasionally write ``a circuit family $C$'' without the
explicit sequence.

An \emph{encoding} for a gate set $\mc{G}$ is a surjective map
$e : \{0,1\}^* \to \langle \mc{G} \rangle_\text{c}$. If $e(s) = C$, we
call the string $s$ a description of $C$. Encodings must satisfy
additional conditions \cite{Wat09} which we do not exhaustively
list. We require an encoding $e$ to be such that every $\Psi_{i,j}$ term
(or lack thereof) in the circuit $e(s)$ must be efficiently
computable from $i$, $j$, and~$s$. Every gate set admits an encoding.

We say that $C$ is an $(a,b)$-circuit if it implements a channel
$C : \mc{L}({\C^2}^{\tensor a}) \to \mc{L}({\C^2}^{\tensor b})$ for $a,b
\in \N$. We say that $C$ is a $(a,b)$-circuit family
if each $C_n$ is an $(a(n), b(n))$-circuit for $a,b : \N \to \N$. We say
that a circuit family $C$ is \emph{uniform} if there exists a Turing
machine $\ttt{T}$ which, on input of $1^n$, outputs a description of
$C_n$ in some prescribed encoding. We say that $C$ is
\emph{polynomial-time}, \emph{exponential-time}, or \emph{triple
exponential-time} if such a $\ttt{T}$
exists and runs in polynomial-time, exponential-time, or triple
exponential-time, respectively.

A gate set $\mc{G}$ is said to be \emph{universal} if for any
$\epsilon > 0$ and any channel $\Phi$ there exists a
$\Psi \in \langle \mc{G} \rangle$ such that
$\norm{\Phi - \Psi}_\diamond \leq \epsilon$. As a consequence of the
Solovay-Kitaev theorem \cite{Kit97} and its algorithmic implementations
\cite{DN06,BGT21arxiv1}, for any computable map $t : \N \to \N$, any
universal gate set $\mc{G}$, and any uniform circuit family $C'$ built
from any (possibly non-universal) gate set $\mc{G}'$, there exists a
uniform family $C$ built from $\mc{G}$ such that
$\norm{C_n - C'_n}_\diamond \leq 2^{-t(n)}$. Moreover, if $t$ is
efficiently computable, we can replace ``uniform'' with
``polynomial-time'' or ``exponential-time''. Unless otherwise specified,
we work with the
universal gate set composed of the state preparation map, the
single-qubit trace, as well as conjugation by the Toffoli, Hadamard, and
phase-shift unitaries \cite{Wat09}.

\subsection{A Simultaneous Quantum Goldreich-Levin Theorem}

We recall a recent result of Kundu and Tan \cite[Lemma 23]{KT22arxiv1}
and recast it slightly to consider uniform and polynomial-time circuit
families. In some sense, this result extends the quantum Goldreich-Levin
theorem \cite{AC02} to a setting with two adversaries.

This lemma concerns a scenario where two non-communicating parties,
Bob and Charlie, share a quantum state $\rho_{x_B,x_C}$. This state
represents the information each of these parties has on the value of two
bit strings $(x_B,x_C) \in \{0,1\}^n \times \{0,1\}^n$ sampled according
to a random variable~$\xi$. The proof given in \cite{KT22arxiv1} for
this lemma shows that if Bob and Charlie can simultaneously,
independently, and respectively determine the values of the inner
products $x_B \cdot y_B$ and $x_C \cdot y_C$ with probability at least
$\frac{1}{2} + \delta$, where the probability is taken over the uniform
random sampling of $y_B$ and $y_C$ (which are then given to Bob and
Charlie, respectively) and the sampling of $(x_B,x_C)$ from $\xi$, then
they could \emph{instead} simultaneously, independently, and
respectively determine the values of $x_B$ and $x_C$ with probability at
least $\delta^3/2$.

More technically, the proof shows how to transform a circuit
guessing the inner product $x_L \cdot y_L$ for either $L \in \{B, C\}$
into a circuit guessing $x_L$. The idea is to adapt the
Bernstein-Vazirani algorithm~\cite{BV97} where the oracle for the inner
product map $y \mapsto x \cdot y$ is replaced with a purified version of
the circuit guessing the inner product. An important property of this
transformation, not explicitly mentioned in \cite{KT22arxiv1} but which
we highlight here, is that this transformation is uniform (respectively,
polynomial-time or exponential-time) in the sense that applying it
circuit-by-circuit to a uniform (respectively, polynomial-time or
exponential-time) circuit family yields
another uniform (respectively, polynomial-time or exponential-time)
circuit family.

In practice, as in \cite{KT22arxiv1}, we take the contrapositive of the
above. We state that if Bob and Charlie have a negligible probability of
computing $x_B$ and $x_C$, then they have at most a negligible advantage
in computing $x_B \cdot y_B$ and $x_C \cdot y_C$.

\begin{lemma}
\label{th:kundu-tan}
	Let $\ell, b, c : \N \to \N$ be maps and, for all $n \in \N$,
	let $\xi_n$ be a random variable on the
	set~$\{0,1\}^{\ell(n)} \times \{0,1\}^{\ell(n)}$. For
	all $n \in \N$ and pairs $(x_B, x_C)$ in the support of
	$\xi_n$, let $\rho_{n, x_B, x_C}$ be a density operator
	on $b(n) + c(n)$ qubits.

	If for every pair $(B', C')$ of uniform $(b, \ell)$- and
	$(c, \ell)$-circuit families, respectively, the map
	\begin{equation}
		n
		\mapsto
		\E_{\substack{
			(x_B,x_C) \gets \xi_n
		}}
		\bra{x_B, x_C}
			\left(B'_n \tensor C'_n\right)
			\left(
				\rho_{n,x_B,x_C}
			\right)
		\ket{x_B, x_C}
	\end{equation}
	is negligible, then for every pair $(B, C)$ of uniform
	$(\ell + b, 1)$- and $(c + \ell, 1)$-circuit families, respectively,
	the map
	\begin{equation}
		n
		\mapsto
		\E_{\substack{
			y_B \gets \{0,1\}^{\ell(n)}  \\
			y_C \gets \{0,1\}^{\ell(n)} \\
			(x_B,x_C) \gets \xi_n
		}}
		\bra{x_B \cdot y_B, x_C \cdot y_C}
			\left(B_n \tensor C_n\right)
			\left(
				y_B \tensor \rho_{n,x_B,x_C} \tensor y_C
			\right)
		\ket{x_B \cdot y_B, x_C \cdot y_C}
		-
		\frac{1}{2}
	\end{equation}
	is negligible.

	Moreover, the above holds if we replace all instances of ``uniform''
	with ``polynomial-time'' or ``exponential-time''.
\end{lemma}

We note that an alternative proof of this lemma can also be found in
\cite{AKL23arxiv}.

%% file: 4-ingenerable.tex
\section{Ingenerable Sequences of Bit Strings}
\label{sc:ingenerable}

In this section, we define the notion of ingenerable sequences,
demonstrate that such sequences exist, and illustrate a few of their
simple properties.

Before proceeding, we emphasize that while our definitions and
results are formulated in the language of quantum circuits, \emph{there
is nothing inherently quantum at play in this section}, with the
exception of \cref{sc:derandomize-cloning-example}.
In particular, the definitions can easily be rephrased to fit a wide
variety of computational models, including classical circuits or
classical Turing machines.
However, we believe the resulting notions are only interesting when the
computational model is \emph{probabilistic} and \emph{uniform}.
Moreover, the proofs of our main existence theorems essentially only
rely on the assumption that there are countably many instances in the
computational model.

Thus, essentially all definitions and results of this section, excluding
\cref{sc:derandomize-cloning-example}, still hold if we were to
replace the word ``circuit'', where the ``quantum'' is implicit, with
the words ``probabilistic classical Turing machine'' or ``uniform
classical circuits supplied with additional random bits''.
Our definitions would also hold for ``non-uniform circuits'', classical
or quantum, but our existence theorems would not.

Of course, sequences which are ingenerable with respect to one model of
computation need not be ingenerable with respect to another.

\subsection{Definitions and Existence}

We first introduce terminology concerning sequences of bit
strings.

\begin{definition}
	Let $\ell : \N \to \N$ be a map. An $\ell$-sequence is a sequence of
	bit strings $(s_n)_{n \in \N}$ such that
	$s_n \in \{0,1\}^{\ell(n)}$ for all $n \in \N$.
\end{definition}

In practice, we will only concern ourselves with sequences of linear
length, such as $2n$-sequences, or of polynomial length. However, the
theory developed in this section is applicable to
arbitrary~$\ell$-sequences.

Now, we formally define the notion of $q$-ingenerability for some map
$q : \N \to \R$. In short, a sequence is $q$-ingenerable if every
circuit family under consideration \emph{eventually always fails} to
produce the elements of the sequence with probability at least $q$. We
also define a weaker notion of this idea, which we call \emph{weak}
ingenerability. A sequence is weakly ingenerable if every circuit family
\emph{fails infinitely often} to produce the elements of
the sequence with probability at least $q$.

We immediately see that every $q$-ingenerable sequence is weakly
$q$-ingenerable. Essentially all of our constructions will use
ingenerable sequences, but weak ingenerability will be easier to compare
and contrast with other existing definitions and notions.

Finally, we consider ingenerability with respect to two different
classes of circuit families. A computably ingenerable sequence is one
which cannot be generated by any uniform circuit family. An
exponential-time ingenerable sequence is one which cannot be generated
by any exponential-time circuit family. Our definition can be adapted
in a straightforward way to other classes of circuit families but
these two will be sufficient for this work. The main conceptual
advantage of exponential-time ingenerable sequences over computably
ingenerable sequences is that certain such sequences can actually be
computed by Turing machines, unlike computably ingenerable sequences
which are, by their natural, uncomputable. See \cref{th:ingenerable}.

Recall from the end of \cref{sc:prelims-circuits} that, for two maps
$a, b : \N \to \N$, an $(a,b)$-circuit family $C$ is a sequence of
circuits $(C_n)_{n \in \N}$ where $C_n$ takes $a(n)$ qubits as input and
produces $b(n)$ qubits as output.\footnote{In this context, we interpret
any $k \in \N$ as the constant map $n \mapsto k$.}

\begin{definition}
	Let $\mc{G}$ be a gate set and $q : \N \to \R$ be a map. An
	$\ell$-sequence~$(s_n)_{n \in \N}$ is computably (respectively,
	exponential-time) $q$-ingenerable with respect
	to $\mc{G}$ if for every uniform (respectively, exponential-time)
	$(0,\ell)$-circuit
	family $C$ built from $\mc{G}$ the inequality
	\begin{equation}
		\bra{s_n} C_n(\varepsilon) \ket{s_n}
		<
		q(n)
	\end{equation}
	holds for all but finitely many values of $n \in \N$.

	We say that $(s_n)_{n \in \N}$ is \emph{weakly} computably
	(respectively, exponential-time) $q$-ingenerable with
	respect to $\mc{G}$ if we only require the above inequality to hold
	for infinitely many values of $n \in \N$.
\end{definition}

It is trivial to show that the concept of ingenerability depends on the
underlying gate set $\mc{G}$, which is why we make this set explicit in
the definition. For example, consider the gate set composed of only
the $\ket{0}$ state preparation map, $\mc{G}=\{c\mapsto c\ketbra{0}\}$,
and the gate set $\mc{G}' = \mc{G} \cup \{L \mapsto X L X\}$ where we
add conjugation by the $X = \ketbra{0}{1} + \ketbra{1}{0}$ unitary to
$\mc{G}$. Clearly, the $n$-sequence $(1^n)_{n \in \N}$
is computably $c$-ingenerable with respect to $\mc{G}$ for any constant
$c > 0$, but not with respect to $\mc{G}'$.

The above example is a bit contrived since $\mc{G}$ and $\mc{G}'$ are
not very interesting gate sets. In particular, neither are universal.
However, it appears difficult to show that the notion of
$r$-ingenerability coincides even for two distinct but universal gate
sets~$\mc{G}$ and~$\mc{G}'$. This is due to the fact that universality
only guarantees that a circuit built from~$\mc{G}'$ can
\emph{approximate} a circuit built from~$\mc{G}$ to arbitrarily small
\emph{but non-zero} error and that $q$-ingenerability is defined via a
simple inequality. For example, it is \emph{a priori} possible for
there to exists an $\ell$-sequence $(s_n)_{n \in \N}$ and a
uniform $(0,\ell)$-circuit family $(C_n)_{n \in \N}$ built from $\mc{G}$
such that $
	\bra{s_n} C_n(\varepsilon) \ket{s_n}
	=
	1
$ for all $n \in \N$, but that for every uniform~$(0, \ell)$-circuit
family $(C'_n)_{n \in \N}$ built from $\mc{G}'$ we would
have that $
	\bra{s_n} C'_n(\varepsilon)\ket{s_n}
	<
	1
$ for all~$n \in \N$. Such a sequence would be $1$-ingenerable with
respect to $\mc{G}'$, but not with respect to $\mc{G}$.

We will later give a slightly modified definition of ingenerability,
\cref{df:ingenerable}, which will be sufficient for our needs and which
will avoid the issue of any possible dependence on the choice of the
gate set. This definition will also have the added benefit of not
requiring an explicit choice of an $q$ map. First, however, we
demonstrate the existence of $q$-ingenerable $\ell$-sequences with
respect to any given gate set $\mc{G}$ and for any suitable choices of
$q$ and $\ell$.

\begin{theorem}
\label{th:expo-ingenerable-gate-set}
	Let $\ell : \N \to \N$ and $q : \N \to \R$ be maps such that
	$q(n)\cdot 2^{\ell(n)} > n + 1 $ for all $n \in \N$. Then, for any
	gate set $\mc{G}$ there exists an $\ell$-sequence which is
	exponential-time $q$-ingenerable with respect to $\mc{G}$.
\end{theorem}

\begin{proof}
	We first construct an
	$\ell$-sequence $(s_n)_{n \in \N}$ of bit strings and then we show
	that it is exponential-time $q$-ingenerable with respect to
	$\mc{G}$. We interpret the output of all Turing machines in this
	proof as circuits constructed from $\mc{G}$.

	Let $\mc{R} = \{r_{k}(n) = 2^{(n + 2)^k}\}_{k \in \N}$ be a set of
	maps.
	In the context of this proof, these maps will be taken as
	bounds on the run times of various Turing machines.
	We note two useful properties of the set $\mc{R}$ before proceeding.
	First, $r_i \leq r_{i + 1}$ for all $i \in \N$. Second,
	if $g \in \mc{O}(2^{n^{k_0}})$ for some $k_0\in\N$, then there
	exists an $r_{k'} \in \mc{R}$ such that $g \leq r_{k'}$.

	Now, let $\mc{T}$ be the set of all Turing machines and let
	$\Phi : \N \to \mc{T} \times \mc{R}$ be a map such that for all
	$\ttt{T} \in \mc{T}$ there are infinitely many $r \in \mc{R}$ such
	that the pair $(\ttt{T}, r)$ is in the image of $\Phi$. By the
	previously discussed properties of $\mc{R}$, this implies that if
	$\ttt{T}$ runs in exponential-time, then there exists an
	$n_\ttt{T}$ such that $\Phi(n_\ttt{T}) = (\ttt{T}, r)$ and
	$\ttt{T}(1^n)$ halts in at most $r(n)$ steps for all $n \in \N$.
	Such a map $\Phi$ exists as $\mc{T} \times \mc{R}$ is countable
	and hence there is a surjection from $\N$ to $\mc{T} \times
	\mc{R}$. Note that $\Phi$ being a surjection is a sufficient,
	\emph{but not necessary}, condition.
	To ease later notation, let $\tau : \N \to \mc{T}$ and $\rho : \N \to
	\mc{R}$ be the unique maps satisfying $\Phi(n) = (\tau(n), \rho(n))$ for all $n
	\in \N$.

	Now, for every $t,n \in \N$, define the set\footnote{%
		Recall that $\tau(t)(1^n)(\varepsilon)$ is a quantum
		state, when this expression is defined. Indeed, we consider the
		$t$-th Turing machine, $\tau(t)$, and run it on input
		$1^n$. This yields the bit string $\tau(t)(1^n)$, assuming the
		machine halts. Implicitly, we interpret this bit string as a
		quantum circuit, which itself implicitly defines a quantum
		channel. Assuming this quantum channel takes as input the empty
		system, we apply it to the empty quantum state $\varepsilon$ to
		obtain as output the state $\tau(t)(1^n)(\varepsilon)$.
	}
	\begin{equation}
		\mc{S}_{t,n}
		=
		\begin{cases}
			\left\{
				x \in \{0,1\}^{\ell(n)}
				:
				\bra{x} \tau(t)(1^n)(\varepsilon)\ket{x}
				\geq
				q(n)
			\right\}
			&
			\parbox[t]{0.32\textwidth}{if, on input of $1^n$, $\tau(t)$
			halts in at most $\rho(t)(n)$ steps and yields a $(0,
			\ell(n))$-circuit.}
			\\
			\varnothing &\text{else.}
		\end{cases}
	\end{equation}
	In other words, $\mc{S}_{t,n}$ is the set of strings generated with
	probability at least $q(n)$ by the circuit described by the $t$-th
	Turing machine on input $1^n$, provided that the machine indeed
	halts in at most $\rho(t)(n)$ steps and outputs the description of a
	$(0, \ell(n))$-circuit. Under these assumptions, we
	have that
	$
		\sum_{x \in \{0,1\}^{\ell(n)}}
			\bra{x}\tau(t)(1^n)(\varepsilon)\ket{x}
		=
		1
	$ which implies that
	\begin{equation}
		\abs{\mc{S}_{t,n}} \leq \frac{1}{q(n)}.
	\end{equation}
	Note that this inequality also holds if $\tau(t)(1^n)$ does not halt
	within the specified number of steps or does not yield a
	$(0,\ell(n))$-circuit as $\abs{\mc{S}_{t,n}} = 0$ in
	this case and our assumptions imply that $q$ must be strictly
	positive. Next, define
	\begin{equation}
		\mc{S}_n
		=
		\bigcup_{t \in \{0,\ldots,n\}}
		\mc{S}_{t,n}.
	\end{equation}
	By our assumption on the maps $\ell$ and $q$, we have that
	\begin{equation}
		\abs{\mc{S}_n}
		\leq
		\left(n + 1\right)
		\cdot
		\frac{1}{q(n)}
		<
		2^{\ell(n)},
	\end{equation}
	which implies that $\{0,1\}^{\ell(n)} \setminus \mc{S}_n$ is
	non-empty. For all $n \in \N$, we take $s_n$ to be the first
	element, in lexicographic order, of
	$\{0,1\}^{\ell(n)} \setminus \mc{S}_n$. This yields an
	$\ell$-sequence $(s_n)_{n \in \N}$.

	Now, we show that $(s_n)_{n \in \N}$ is exponential-time
	$q$-ingenerable with respect to $\mc{G}$. Consider an
	exponential-time $(0,\ell)$-circuit family
	$(C_n)_{n \in \N}$. If no such circuit family exists, then
	$(s_n)_{n \in \N}$ is vacuously exponential-time $q$-ingenerable and
	we are done.

	Let $\ttt{T}$ be an exponential-time Turing machine which generates
	this circuit family. By our previous discussion on
	$\mc{R}$ and $\Phi$, there exists a $t_\ttt{T} \in \N$ such that
	$\Phi(t_\ttt{T}) = (\ttt{T}, r)$ and where $\ttt{T}(1^n)$ halts in
	at most $r(n)$ steps. It is now sufficient to show that
	\begin{equation}
		n \geq t_\ttt{T}
		\implies
		\bra{s_n}
		C_n(\varepsilon)
		\ket{s_n}
		<
		q(n)
	.
	\end{equation}
	Assume this was not the case and that there exists an
	$n' \geq t_\ttt{T}$ such that $
		\bra{s_{n'}}
			C_{n'}(\varepsilon)
		\ket{s_{n'}}
		\geq q(n')
	$. Then, by definition, $
		s_{n'}
		\in
		\mc{S}_{t_\ttt{T},n'}
		\subseteq
		\mc{S}_{n'}
	$. Contradiction, since $
		s_{n'} \in \{0,1\}^{\ell(n')}
		\setminus
		\mc{S}_{n'}
	$.
\end{proof}

Note that the above proof is constructive
in the sense that once $q$, $\ell$, $\mc{G}$ and $\Phi$ are fixed, then
the sequence $(s_n)_{n \in \N}$ is uniquely determined.

We also note that by modifying the set $\mc{R}$ used in the proof, we
can obtain sequences which are $q$-ingenerable with respect to circuit
families generated by Turing machines operating under various time
constraints. An important example for this work is the case where
$\mc{R} = \{r_k(n) = \infty\}_{k \in \N}$ and where we say that the
Turing machine $\ttt{T}$ halts in at most $\infty$ steps on a given
input if and only if it does halt. For completeness, we formalize this
below.

\begin{theorem}
\label{th:comput-ingenerable-gate-set}
	Let $\ell : \N \to \N$ and $q : \N \to \R$ be maps such that
	$q(n)\cdot 2^{\ell(n)} > n + 1 $ for all $n \in \N$. Then, for any
	gate set $\mc{G}$ there exists an $\ell$-sequence which is
	computably $q$-ingenerable with respect to $\mc{G}$.
\end{theorem}
\begin{proof}[Proof sketch.]
	Follow the proof given for \cref{th:expo-ingenerable-gate-set}, but
	with $\mc{R} = \{r_k(n) = \infty\}_{k \in \N}$.
\end{proof}

We now address the question of the possible dependence of our definition
of ingenerable sequences on the choice of the underlying gate set.
Looking ahead, this work will be mainly concerned with $\ell$-sequences
which are computably or exponential-time $(p\cdot2^{-\ell})$-ingenerable
for a polynomial~$p$, but that the particular polynomial is not
important. With this in mind, we formulate the following definition of
ingenerability which will be sufficient for our needs as well as
independent of the choice of the gate set. We also note that this
definition naturally extends to the notion of weak ingenerability.

\begin{definition}
\label{df:ingenerable}
	Let $\ell : \N \to \N$ be a map. An $\ell$-sequence is said to be
	computably (respectively, exponential-time) ingenerable if for every
	gate set $\mc{G}$ there exists a polynomial $p_\mc{G} : \N \to \R$
	such that the sequence is computably (respectively,
	exponential-time) $(p_\mc{G} \cdot 2^{-\ell})$-ingenerable with
	respect to $\mc{G}$.
\end{definition}

By the Solovay-Kitaev theorem, we will see that
$(p \cdot 2^{-\ell})$-ingenerability with respect to a universal gate
set for any polynomial $p : \N \to \R$ is a sufficient condition to be
ingenerable. By definition, it is also a necessary condition. However, a
technical condition on the computability of $\ell$ is needed to formally
prove this. The following lemma, which establishes conditions under
which ingenerability is a trivial property, will help us account for
this technicality.

\begin{lemma}
\label{th:ingenerable-trivial}
	Let $\ell : \N \to \N$ be a map. The following two statements are
	equivalent:
	\begin{enumerate}
		\item
			Every $\ell$-sequence is computably (respectively,
			exponential-time) ingenerable.
		\item
			Either $\ell$ is uncomputable (respectively, can't be
			computed in exponential-time), or $\ell \in \mc{O}(\log)$.
	\end{enumerate}
\end{lemma}
\begin{proof}
	We first show that the second statement implies the first.

	If $\ell$ is uncomputable (respectively, cannot be computed in
	exponential-time), then there are no uniform (respectively,
	exponential-time)~$(0,\ell)$-circuit families, as a Turing machine
	producing circuits with the appropriate number of output
	qubits could be used to compute $\ell$. As there are no suitable
	$(0,\ell)$-circuit family, it is vacuously true that every
	$\ell$-sequence is computably (respectively, exponential-time)
	ingenerable.

	If $\ell \in \mc{O}(\log)$, then there exists $n', c \in \N$
	such that $
		n \geq n'
		\implies
		\ell(n) \leq c \log(n)
	$, which implies that $2^{-\ell(n)} \geq n^{-c}$ for all
	sufficiently large values of $n$. Now, take
	$p(n)=n^{c+1}$ and $\tilde{n}=\max\{n',2\}$
	so that $
		n \geq \tilde{n}
		\implies
		p(n) \cdot 2^{-\ell(n)}
		\geq
		n
		\geq
		2
	$. Now, consider any $\ell$-sequence $(s_n)_{n \in \N}$
	and any uniform~$(0,\ell)$-circuit family $(C_n)_{n \in \N}$ built
	from any gate set $\mc{G}$. For all $n \geq \tilde{n}$, we have that
	\begin{equation}
		\bra{s_n} C_n(\varepsilon) \ket{s_n}
		\leq
		1
		<
		2
		\leq
		p(n) \cdot 2^{-\ell(n)}.
	\end{equation}
	Thus, $(s_n)_{n \in \N}$ is computably ingenerable, which implies
	that it is also exponential-time ingenerable.

	We now show that the first statement implies the second. Assume that
	every $\ell$-sequence is computably (respectively, exponential-time)
	ingenerable. Further, assume that $\ell$ is computable
	(respectively, computable in exponential-time), as otherwise we are
	done. Consider the gate set~$\mc{G}$ composed of precisely the
	$\ket{0}$ state preparation map, the
	$\ell$-sequence $(0^{\ell(n)})_{n \in \N}$, and the uniform
	(respectively, exponential-time)
	$(0,\ell)$-circuit family $(C_n)_{n \in \N}$ constructed from
	$\mc{G}$ where $C_n$ is simply the parallel composition of $\ell(n)$
	instances of the $\ket{0}$ state preparation map. Clearly,
	$
		\bra{0^{\ell(n)}}
		C_n(\varepsilon)
		\ket{0^{\ell(n)}}
		=
		1
	$. But, $(0^{\ell(n)})_{n \in \N}$ is assumed to be computably
	(respectively, exponential-time) ingenerable.
	Thus, there exists a polynomial $p_\mc{G} : \N \to \R$ such that
	$1 < p_\mc{G}(n) \cdot 2^{-\ell(n)}$ for all sufficiently large
	values of $n$. This implies that $\ell(n) \leq \log(p_\mc{G}(n))$
	for all sufficiently large values of $n$, which is sufficient to
	conclude that $\ell \in \mc{O}(\log)$.
\end{proof}

\begin{lemma}
\label{th:ingenerable-universal}
	Let $\mc{G}$ be a universal gate set and $\ell : \N \to \N$ be any
	map (respectively, any efficiently computable map). If an
	$\ell$-sequence is computably (respectively,
	exponential-time) $(p\cdot2^{-\ell})$-ingenerable with
	respect to $\mc{G}$ for a polynomial $p : \N \to \R$, then it is
	computably (respectively, exponential-time) ingenerable.
\end{lemma}
\begin{proof}
	Let $(s_n)_{n \in \N}$ be an $\ell$-sequence which is
	$(p \cdot 2^{-\ell}$)-ingenerable with respect to $\mc{G}$. Further,
	assume that $\ell$ is computable, as otherwise $(s_n)_{n \in \N}$ is
	already known to be ingenerable by \cref{th:ingenerable-trivial}.

	Let $\mc{G}'$ be any gate set and let $(C'_n)_{n \in \N}$ be a
	uniform (respectively, exponential-time) family of
	$(0,\ell)$-circuits built from this set. Now, let
	$(\tilde{C}_n)_{n \in \N}$ be a uniform family of
	$(0,\ell)$-circuits with gates from the universal gate set $\mc{G}$
	such that $
		\norm*{C_n - \tilde{C}_n}_\diamond
		\leq
		2^{-\ell(n)}
	$ for all $n \in \N$. Such a uniform (respectively,
	exponential-time) family exists by the Solovay-Kitaev theorem. Since
	$(s_n)_{n \in \N}$ is computably (respectively, exponential-time)
	$(p \cdot 2^{-\ell})$-ingenerable with respect to $\mc{G}$, we have
	that
	\begin{equation}
		\bra{s_n} C_n(\varepsilon) \ket{s_n}
		\leq
		\bra{s_n} \tilde{C}_n(\varepsilon) \ket{s_n}
		+
		2^{-\ell(n)}
		<
		(p(n) + 1) \cdot 2^{-\ell(n)}
	\end{equation}
	for all sufficiently large values of $n$. Hence, the
	$\ell$-sequence $(s_n)_{n \in \N}$ is computably (respectively,
	exponential-time)
	$(p+1)\cdot2^{-\ell}$-ingenerable with respect to $\mc{G}'$. This
	yields the desired result as $p+1$ is a polynomial.
\end{proof}

The culmination of the results in this section is the following theorem
which asserts the existence of computably ingenerable and
exponential-time $\ell$-sequences for any suitable choice of $\ell$.

\begin{theorem}
\label{th:ingenerable}
	There exists a computably ingenerable $\ell$-sequence for any map
	$\ell : \N \to \N$.

	Moreover, if $\ell$ is efficiently computable, then
	there exists an exponential-time ingenerable $\ell$-sequence which
	can be computed in triple exponential-time by a classical
	deterministic Turing machine.
\end{theorem}
\begin{proof}
	Let $\mc{G}$ be a universal gate set and let $p(n)= n + 2$. Then,
	by \cref{th:comput-ingenerable-gate-set}, there exists a
	computably $(p \cdot 2^{-\ell})$-ingenerable $\ell$-sequence with
	respect to $\mc{G}$. By \cref{th:ingenerable-universal}, this is a
	computably ingenerable $\ell$-sequence.

	Keeping the same $\mc{G}$ and $p$ and assuming that $\ell$ can be
	efficiently computed, we now show how to fix the map
	$\Phi$ in the proof of \cref{th:expo-ingenerable-gate-set} to
	obtain an $\ell$-sequence which is exponential-time ingenerable with
	respect to $\mc{G}$ but which can be
	computed in triple exponential-time by a classical deterministic
	Turing machine. By \cref{th:ingenerable-universal}, this is an
	exponential-time ingenerable sequence.
	
	Recall that
	$\mc{R} = \left\{r_k(n) = 2^{(n+2)^k}\right\}_{k \in \N}$.
	Let $\langle \cdot \rangle : \{0,1\}^* \to \mc{T}$ be a
	mapping from bit strings to Turing machines and let~$\ttt{U}$ be a
	Turing machine satisfying
	\begin{equation}
		\label{eq:ingen-turing-machine}
		\ttt{U}(s, x, 1^t)
		=
		\begin{cases}
			(1,\langle s \rangle(x)) &\text{if } \langle s \rangle \text{
				halts within } $t$ \text{ steps on input } x
			\\
			0 &\text{else}
		\end{cases}
	\end{equation}
	and which runs in polynomial time.\footnote{%
		Technically, $\ttt{U}$ should receive
		$(r(s),01,r(x),01,1^t)$ as input, where
		$r : \{0,1\}^* \to \{0,1\}^*$ is the map which repeats each
		element of the argument twice, i.e.~$r(010) = 001100$.
		This, with the help of the separators $01$, allows
		$\ttt{U}$ to unambiguously parse and separate $s$, $x$, and
		$1^t$.
		However, for simplicity, we omit these extra details in this
		proof.}
	In other words, $\ttt{U}$ is an efficient universal Turing machine
	with a time bound.\footnote{%
		Such a machine $\ttt{U}$ and encoding $\langle \cdot \rangle$
		can be found implicitly in \cite{NW06}.
		In particular, Theorem 7 of \cite{NW06} shows that the run time
		of their universal Turing machine $U_{3,11}$ is polynomial in
		the run time of the simulated machine $M$ and the number of
		states of $M$.
		Adding a mechanism to count the number of simulated steps can
		be done with at most a polynomial overhead.
		Moreover, they implicitly provide an injective encoding
		$e : \mc{T} \to \{0,1\}^*$ by, essentially, simply listing all
		transition rules of the Turing machine in a prescribed manner.
		From this encoding $e$, we can define the surjection
		$\langle \cdot \rangle : \mc{T} \to \{0,1\}$ by
		$\langle s \rangle = e^{-1}(s)$ if $s$ is in the image of $e$,
		and where $\langle s \rangle$ is otherwise the $2$ state Turing
		machine which immediately and always transitions from its
		initial state to its halting state in the first step.
		Note that this implies that the number of states of $\langle s
		\rangle$ is at most $\abs{s} + 2$.
		As a final note, the Turing machines of \cite{NW06} use a
		non-binary alphabet.
		However, this can be transformed to a Turing machine using the
		alphabet $\{0,1\}$ with a constant multiplicative overhead in
		time and space usage using standard techniques \cite{AB09}.
	}
	Further, let $b : \N \to \{0,1\}^*$ be the natural bijection between
	the non-negative integers and the list of all finite bit strings
	taken in lexicographic order, starting with
	the empty string, i.e.~$b(0) = \varepsilon$, $b(1) = 0$, $b(2) = 1$,
	and $b(3) = 00$.

	Let $(\alpha_n)_{n \in \N}$ and $(\beta_n)_{n \in \N}$ be two
	sequence of non-negative integers such that the following three
	conditions are satisfied. First, $\alpha_n, \beta_n \leq n$ for all
	$n \in \N$. Second, $n \mapsto (\alpha_n, \beta_n)$ can be computed
	efficiently in $n$. Third, for every $a \in \N$, there are
	infinitely many distinct $b \in \N$ such that $(a,b)$ is in the
	image of $n \mapsto (\alpha_n, \beta_n)$.\footnote{%
		As a concrete example, we can take $(\alpha_n)_{n \in \N}$ to be
		the sequence obtained by concatenating each finite sequence of
		the form $0, 1, 2, \ldots, n$ for all $n \in \N$ and take
		$(\beta_n)_{n \in \N}$ to be the sequence obtained by repeating,
		in increasing order, every element $n \in \N$ exactly $n + 1$
		times.}
	We define $\Phi$ by $n \mapsto (\langle b(\alpha_n) \rangle,
	r_{\beta_n})$.

	Note that by our assumptions on $(\alpha_n)_{n \in \N}$ and
	$(\beta_n)_{n \in \N}$, the map $\Phi$ we have constructed here
	satisfies the assumptions made of it in the proof of
	\cref{th:expo-ingenerable-gate-set}, namely that for every
	$\ttt{T} \in \mc{T}$ there are infinitely many $r \in \mc{R}$ such
	that $(\ttt{T}, r)$ is in the image of $\Phi$.

	Now that we have fixed $\Phi$, the $\ell$-sequence
	$(s_n)_{n \in \N}$ constructed in the proof of
	\cref{th:expo-ingenerable-gate-set} is fixed. It now suffices to
	show how it can be computed in triple exponential-time.

	Recall from the proof of \cref{th:expo-ingenerable-gate-set} that
	$s_n$ is the first element, in lexicographic order, of the non-empty
	set
	$\{0,1\}^{\ell(n)} \setminus \mc{S}_n$ where $\mc{S}_n = \cup_{t =
	0}^n \mc{S}_{t,n}$ and
	\begin{equation}
		\mc{S}_{t,n}
		=
		\begin{cases}
			\left\{
				x \in \{0,1\}^{\ell(n)}
				:
				\bra{x} \tau(t)(1^n)(\varepsilon)\ket{x}
				\geq
				q(n)
			\right\}
			&
			\parbox[t]{0.32\textwidth}{if, on input of $1^n$, $\tau(t)$
			halts in at most $\rho(t)(n)$ steps and yields a $(0,
			\ell(n))$-circuit.}
			\\
			\varnothing &\text{else.}
		\end{cases}
	\end{equation}
	where $\tau$ and $\rho$ are the projections on the first and second
	component of $\Phi$, respectively. Thus, to compute $s_n$, it
	suffices to be able to check the membership of a string in the sets
	$\mc{S}_{t,n}$ for all $t \leq n$.
	We describe a Turing machine $\ttt{T}'$ running in triple
	exponential-time which finds the first string, in lexicographic
	order, which is not in any of these sets.

	At a high-level, the triple exponential
	run time of $\ttt{T}'$ can be explained as follows.
	The first exponential is due to the need to simulated other Turing
	machine running in exponential-time.
	The second exponential is essentially due to the fact that
	$\ttt{T}'$ will be simulating more and more complex Turing machines
	as $n$ grows (specifically due to the fact that $n \mapsto r_n(n)$
	is not exponential in $n$, but is double exponential).
	The third exponential is due to the ``brute force'' numerical
	computation to analyse the quantum circuits produced by the
	simulated Turing machines.

	Formally, let $\ttt{T}'$ be a Turing machine which implements the
	following algorithm on input of $1^n$:
	\begin{enumerate}
		\item
			Compute $\ell(n)$.

			(By assumption, this can be done efficiently.)
		\item
			Initialize a string $s = 0^{\ell(n)}$ and a counter $t = 0$,
			both of which will be updated as the algorithm proceeds.

			(This can be done efficiently in $\ell(n)$.)
		\item
			If $t > n$, go to step 4. Else, do the following:
			\begin{enumerate}
				\item
					Compute $\alpha_t$ and $\beta_t$.

					(This can be done efficiently in $t$, and so
					efficiently in $n$ as $t \leq n$.)
				\item
					Run $\ttt{U}(b(\alpha_t), 1^n,
					1^{r_{\beta_t}(n)})$ and let $z$ denote the output.
				
					(By assumption, there exists a polynomial $p$, which
					we can assume to be non-decreasing, such that this
					takes at most
					$p(\abs{b(\alpha_t)} + n + r_{\beta_t}(n))$ steps.
					As $\alpha_t, \beta_t \leq t \leq n$,
					$\abs{b(n)} \leq n$, and $r_i \leq r_{i + 1}$, this
					then takes at most $p(2n + r_n(n))$ steps. Note that
					$r_n(n) = 2^{(n + 2)^n}$. Since $(n+2)^n \in
					\mc{O}(2^{n^2})$, we conclude that $r_n(n)$ and
					$p(2n + r_n(n))$ are at most double exponential in
					$n$.)
				\item
					If the first bit of $z$ is $0$, which is to say that
					the Turing machine being simulated did not
					terminate within the desired number of steps, do
					nothing. Else, let $z'$ be all but the first bit of
					$z$ and check that $z'$ encodes a $(0,\ell(n))$
					circuit. If it does not, do nothing. If it does,
					then do the following:

					(This check can be done efficiently in $\ell(n)$
					and the length of $z'$. Since $\ell(n)$ is
					at at most double exponential in $n$ and the
					length of $z'$ is at most double exponential
					in $n$, as it is the result of a simulation
					which ran for at most $r_n(n)$ simulated
					steps, then this check can be executed in
					time double exponential in $n$.)
					\begin{enumerate}
						\item
							Let $C$ be the $(0,\ell(n))$ circuit
							described by $z'$.
							Check if $\bra{s} C(\varepsilon) \ket{s} <
							p(n) \cdot 2^{-\ell(n)}$.
							If so, increment $t$ by $1$ and go back to
							step 3.
							If not, update $s$ to be the next string of
							length $\ell(n)$, in lexicographic order,
							and execute this step again.

							(The time complexity of this check is
							dominated by the complexity of computing the
							value $\bra{s} C(\varepsilon) \ket{s}$.
							Recall that $C$ is described by $z'$ which
							was produced after at most $r_n(n)$
							simulated steps. Hence,
							$\abs{z'} \leq r_n(n)$ and so
							$C$ can have at most $r_n(n)$ gates. Thus,
							we can compute $\bra{s} C(\varepsilon)
							\ket{s}$ in time exponential in $r_n(n)$,
							which is triple exponential in $n$ as
							$r_n(n)$ is double exponential in $n$.)
					\end{enumerate}
			\end{enumerate}
		\item
			Output $s$.
	\end{enumerate}
	The fact that $\ttt{T}'(1^n)$ indeed halts and outputs $s_n$ follows
	from the proof of \cref{th:expo-ingenerable-gate-set}. The
	time complexity is dominated by step 3.c.i, with each execution of
	that step taking time which is at most triple exponential in $n$.
	Furthermore, this step is executed at most $(n+1) 2^{\ell(n)}$
	times, which is itself at most triple exponential in $n$. We
	conclude that $\ttt{T}'$ runs in at most triple exponential time on
	input of $1^n$, which is the desired result.
\end{proof}

In \cref{sc:mlr}, we study some links between computably ingenerable
sequences and Martin-L\"of random sequences. Recall that a Martin-L\"of sequence is,
essentially, a fixed and infinite sequence of bits $w\in\{0,1\}^\infty$
which passes all possible computable tests of randomness. We show in
this appendix that every Martin-L\"of random sequence yields a weakly
computably ingenerable $\ell$-sequence following the application of a
natural bijection and that there exists computably ingenerable
$\ell$-sequences which are not Martin-L\"of random under the action of
the same bijection. We also show the existence of weakly computably
ingenerable sequences which are not computably ingenerable.

\subsection{Replacing Uniformly Random Strings with Ingenerable Strings}

The goal of this section is to prove \cref{th:random-to-ingenerable}
which is stated below. This theorem can be understood as stating that if
a given computation succeeds with sufficiently small probability over
uniformly random inputs, than applying that same computation on fixed
inputs determined by an ingenerable $\ell$-sequence also yields a
negligible success probability, provided that $\ell$ is large enough.

\begin{theorem}
\label{th:random-to-ingenerable}
	Let $\left(s_n\right)_{n \in \N}$ be a computably (respectively,
	exponential-time) ingenerable $\ell$-sequence for any $\ell : \N \to
	\N$ (respectively, $\ell \in \mc{O}(n^k)$ for some $k \in \N$). Then
	for any uniform (respectively, exponential-time) family
	$\left(C_n\right)_{n\in\N}$ of~$(\ell,1)$-circuits, there exists a
	polynomial $p$ such that for all sufficiently large values of $n$
	(where ``sufficiently large'' might depend on $(C_n)_{n\in\N}$) we
	have that
	\begin{equation}
		\bra{1} C_n(s_n) \ket{1}
		\leq
		p(n)
		\cdot
		\left(
			2^{-\ell(n)}
			+
			\E_{x \gets \{0,1\}^{\ell(n)}}
			\bra{1} C_n(x) \ket{1}
		\right).
	\end{equation}
	In particular, if $\ell \in \omega(\log)$ and $
		n
		\mapsto
		\E_{x \gets \{0,1\}^{\ell(n)}} \bra{1}C_n(x)\ket{1}
	$ is negligible, then so is
	$n \mapsto \bra{1}C_n(s_n) \ket{1}$.
\end{theorem}

The proof of \cref{th:random-to-ingenerable} is simply an application of
the following lemma which considers a naive and non-efficient way
to transform a circuit $C$ for a decision problem into a circuit
$\tilde{C}$ for a search problem. It then gives a lower bound on the
probability that $\tilde{C}$ outputs any given string.

\begin{lemma}
\label{th:decision-to-search}
	Let $n \in \N$ and let $C$ be an $(n,1)$-circuit. Now, let
	$\tilde{C}$ be a $(0,n)$-circuit which implements the following
	algorithm:
	\begin{enumerate}
		\item
			Initialize an empty set $\mc{S}$.
		\item
			Iterating over all strings $x \in \{0,1\}^n$, do the
			following: Run $C(x)$ and measure the output in the
			computational basis. If the result is $1$, add $x$ to
			$\mc{S}$. Else, do nothing.
		\item
			Sample uniformly at random a string $x$ from $\mc{S}$ and
			output it. If $\mc{S}$ is empty, sample $\tilde{x}$
			uniformly at random from $\{0,1\}^n$.
	\end{enumerate}
	Then, for any $s \in \{0,1\}^n$, we have that
	\begin{equation}
		\frac{
			\bra{1}C(s)\ket{1}
		}{
			1
			+
			\sum_{x \in \{0,1\}^n \setminus \{s\}}
			\bra{1}C(x)\ket{1}
		}
		\leq
		\bra{s}\tilde{C}(\varepsilon)\ket{s}
		.
	\end{equation}
\end{lemma}

We give the proof of the above lemma in
\cref{sc:decision-to-search}. The proof simply neglects the case
where $\mc{S}$ is empty by the time the algorithm executes step 3 and
notes that the left-hand side of the above inequality is essentially $
	\Pr\left[s \text{ is outputted.}\;|\; s \text{ is in } \mc{S}.\right]
	\cdot
	\Pr\left[s\text{ in } \mc{S}\right]
$, up to an application of Jensen's lemma.

We can now prove \cref{th:random-to-ingenerable}.

\begin{proof}[Proof of \cref{th:random-to-ingenerable}]
	Let $(\tilde{C}_n)_{n \in \N}$ be a $(0,\ell)$-circuit family such
	that $\tilde{C}_n$ implements the algorithm described in
	\cref{th:decision-to-search} with respect to $C_n$ for all $n\in\N$.

	Note that if $(C_n)_{n \in \N}$ is an exponential-time family and
	$\ell \in \mc{O}(n^k)$, then $(\tilde{C}_n)_{n \in \N}$ is an
	exponential-time circuit family. Otherwise, if $(C_n)_{n \in \N}$
	is a uniform family then so is $(\tilde{C}_n)_{n \in \N}$.

	We then have that
	\begin{equation}
		\bra{1}C_n(s_n)\ket{1}
		\leq
		\bra{s_n}\tilde{C}_n(\varepsilon)\ket{s_n}
		\left(
			1
			+
			\sum_{x \in \{0,1\}^{\ell(n)} \setminus \{s_n\}}
				\bra{1}C_n(x)\ket{1}
		\right)
	\end{equation}
	for all $n \in \N$. As $(s_n)_{n \in \N}$ is a computably
	(respectively, exponential-time) ingenerable $\ell$-sequence, there
	exists a polynomial $p$ such that, for all sufficiently large values
	of $n$, we have that $
		\bra{s_n} \tilde{C}_n(\varepsilon) \ket{s_n}
		\leq
		p(n) \cdot 2^{-\ell(n)}
	$. Thus,
	\begin{equation}
		\bra{1}C_n(s_n)\ket{1}
		\leq
		\frac{p(n)}{2^{\ell(n)}}
		\left(
			1
			+
			\sum_{x \in \{0,1\}^{\ell(n)} \setminus \{s_n\}}
				\bra{1}C_n(x)\ket{1}
		\right)
	\end{equation}
	for all sufficiently large values of $n$. Adding the
	$x = s_n$ term to the sum and distributing the
	$2^{-\ell(n)}$ factor yields the first desired result.

	Furthermore, if $\ell\in\omega(\log)$, then $2^{-\ell}$
	is negligible. If
	$\E_{x \gets \{0,1\}^{\ell(n)}} \bra{x} C(x) \ket{x}$ is
	also negligible, then the upper bound is a polynomial times the sum
	of two negligible functions and thus is itself negligible.
\end{proof}

\subsection{Derandomizing a Cloning Task via Ingenerable Sequences}
\label{sc:derandomize-cloning-example}

In this section, we prove that there exists a sequence of fixed states
$(\ketbra{\psi_n})_{n \in \N}$ which cannot be cloned by any
uniform family of quantum circuits~$(C_n)_{n \in \N}$. This
will be a simple demonstration of the ideas we will be using in our
constructions of uncloneable advice in \cref{sc:advice}.

The following lemma was first shown by Molina, Vidick, and Watrous
\cite{MVW13} to formalize the proof of security of Wiesner's quantum
money scheme \cite{Wie83}. It can be interpreted as a bound on how well
a party can clone a single Wiesner state $\ket*{x^\theta}$ selected
uniformly at random.

\begin{lemma}
\label{th:money}
	Let $n \in \N$ and let $C$ be an $(n,2n)$-circuit. Then,
	\begin{equation}
		\E_{\substack{x \gets \{0,1\}^n \\ \theta \gets \{0,1\}^n}}
				\bra*{x^\theta}\bra*{x^\theta}
					C\left(\ketbra*{x^\theta}\right)
				\ket*{x^\theta}\ket*{x^\theta}
		\leq
		\left(\frac{3}{4}\right)^n.
	\end{equation}
\end{lemma}

Note that the lemma holds for all quantum channels, but we restrict
ourselves to circuits to maintain consistency with the rest of this
work.

We can now ``derandomize'' \cref{th:money} by replacing the uniformly
random Wiesner state $\ket{x^\theta}$ with a fixed one chosen according
to an ingenerable sequence. It then suffices to invoke
\cref{th:random-to-ingenerable} to obtain the desired result.

\begin{theorem}
\label{th:money-derandomized}
	Let $(s_n)_{n \in \N}$ be a computably ingenerable
	$2n$-sequence and parse each $s_n$ as a pair of bit
	strings of length $n$, \textit{i.e.}: $
		s_n
		=
		(x_n, \theta_n)
	$ for $x_n, \theta_n \in \{0,1\}^n$ for all
	$n \in \N$. Then, for any uniform family
	of~$(n, 2n)$-circuits $(C_n)_{n \in \N}$, we
	have that
	\begin{equation}
		n
		\mapsto
		\bra*{x_n^{\theta_n}}
		\bra*{x_n^{\theta_n}}
		C_n\left(
			\ketbra*{x_n^{\theta_n}}
		\right)
		\ket*{x_n^{\theta_n}}
		\ket*{x_n^{\theta_n}}
	\end{equation}
	is a negligible function.
\end{theorem}

\begin{proof}
	Conceptually, it suffices to apply
	\cref{th:random-to-ingenerable} to \cref{th:money}. However, a few
	technical details are needed to frame \cref{th:money} in a way where
	\cref{th:random-to-ingenerable} is applicable.

	Let $(S_n)_{n \in \N}$ be a uniform family of
	$(2n, 3n)$-circuits such that on input of
	$x \tensor \theta$ they output the state
	$x \tensor \theta \tensor \ketbra{x^n}$ for any
	$x,\theta \in \{0,1\}^n$. Let $(R_n)_{n \in \N}$
	be a uniform family of $(4n, 1)$-circuits which, on input of
	$x \tensor \theta \tensor \rho$ for any strings
	$x, \theta \in \{0,1\}^n$ and any state $\rho$ on $2n$
	qubits, applies the unitary $H^\theta \tensor H^\theta$ to $\rho$,
	measures the resulting qubits in the computational basis, and
	outputs $1$ if and only if two copies of $x$ are obtained. It
	outputs $0$ otherwise.

	Now, consider the uniform family of $(2n,1)$-circuits
	$(\tilde{C}_n)_{n \in \N}$ where $\tilde{C}_n$ is
	obtained by composing the $S_n$, $C_n$, and $R_n$
	circuits, in that order and where $C_n$ act on the last
	$n$ qubits produced by $S_n$. This is illustrated in
	\cref{fg:money-derandomized}. At the level of channels, we see that
	\begin{equation}
		\tilde{C}_n
		=
		R_n
		\circ
		\left(
			\Id_n \tensor \Id_n \tensor C_n
		\right)
		\circ
		S_n.
	\end{equation}
	In particular, note that
	\begin{equation}
	\label{eq:money-derandomized}
		\bra{1}\tilde{C}_n(x \tensor \theta) \ket{1}
		=
		\bra*{x^\theta}\bra*{x^\theta}
			C\left(\ketbra*{x^\theta}\right)
		\ket*{x^\theta}\ket*{x^\theta}
	\end{equation}
	for all $x, \theta \in \{0,1\}^n$. Thus, by \cref{th:money},
	\begin{equation}
		n
		\mapsto
		\E_{\substack{x\gets\{0,1\}^n\\\theta\gets\{0,1\}^n}}
			\bra{1}\tilde{C}_n(x \tensor \theta) \ket{1}
	\end{equation}
	is negligible. Hence, by \cref{th:random-to-ingenerable} and the
	fact that $(s_n)_{n \in \N}$ is computably ingenerable,
	\begin{equation}
		n
		\mapsto
		\bra{1}
			\tilde{C}_n(x_n \tensor \theta_n)
		\ket{1}
	\end{equation}
	is also negligible, which, by \cref{eq:money-derandomized}, is the
	desired result.
\end{proof}

\begin{figure}
	\begin{center}
		\input{fig-circuit-money-derandomized.tex}
	\end{center}
	\caption{\label{fg:money-derandomized}%
		A schematic representation of the $\tilde{C}_n$ circuits
		constructed in the proof of \cref{th:money-derandomized}. The
		wires are labelled, when possible, with the states they are
		expected to carry in the context of the proof. Every wire
		represents $n$ qubits, except the initial and final wires
		which represent $2n$ and $1$ qubits, respectively.}
\end{figure}
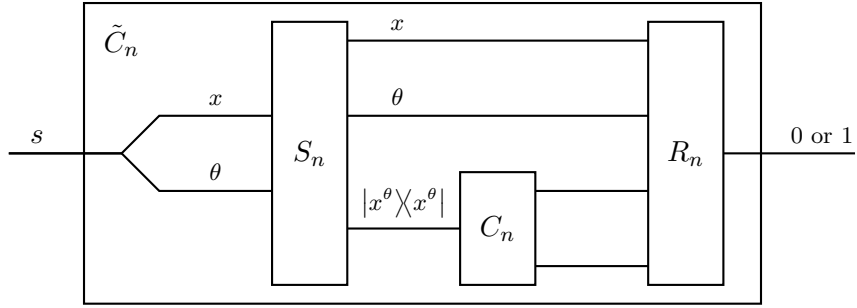

%% file: fig-circuit-money-derandomized.tex
	\begin{tikzpicture}[thick]

		\node (s) at (0,0) {$S_n$};
		\draw ($(s) + (0.5,1.75)$) rectangle ($(s) - (0.5,1.75)$);

		\node (c) at ($(s) + (2.5,-1)$) {$C_n$};
		\draw ($(c) + (0.5,0.75)$) rectangle ($(c) - (0.5,0.75)$);

		\node (r) at ($(s) + (5,0)$) {$R_n$};
		\draw ($(r) + (0.5,1.75)$) rectangle ($(r) - (0.5,1.75)$);

		\node (x) at ($(s) + (-2,0.5)$) {};
		\node (t) at ($(s) + (-2,-0.5)$) {};
		\node (i) at ($0.5*(x) + 0.5*(t) + (-2,0)$) {};

		\draw
			($(i)$)
			-- node [pos=1/4,above] {$s$}
			($(i) + (1.5,0)$)
			--
			($(x)$)
			-- node [midway, above] {\footnotesize$x$}
			($(s) + (-0.50,0.50)$);

		\draw
			($(i)$)
			--
			($(i) + (1.5,0)$)
			--
			($(t)$)
			-- node [midway, above] {\footnotesize$\theta$}
			($(s) + (-0.50,- 0.50)$);

		\draw
			($(s) + ( 0.50, 1.50)$)
			-- node [pos=1/6,above] {\footnotesize$x$}
			($(r) + (-0.50, 1.50)$);
		\draw
			($(s) + ( 0.50, 0.50)$)
			-- node [pos=1/6,above] {\footnotesize$\theta$}
			($(r) + (-0.50, 0.50)$);
		\draw
			($(s) + ( 0.50,-1.00)$)
			-- node [midway,above] {\footnotesize$\ketbra{x^\theta}$}
			($(c) + (-0.50, 0.00)$);
		\draw ($(c) + ( 0.50, 0.50)$) -- ($(r) + (-0.50,-0.50)$);
		\draw ($(c) + ( 0.50,-0.50)$) -- ($(r) + (-0.50,-1.50)$);

		\draw
			($(r) + (0.5,0)$)
			-- node [pos=0.75,above] {\footnotesize$0$ or $1$}
			($(r) + (2.25,0)$);

		\draw ($(s) + (-3,2)$) rectangle ($(r) + (1,-2)$);
		\node (d) at ($(s) + (-2.5,1.5)$) {$\tilde{C}_n$};

	\end{tikzpicture}

%% file: 5-state.tex
\section{State Complexity Classes and their Generalizations to Cloning}
\label{sc:cloning-complexity}

In this section, we define a notion of cloning complexity class
by generalizing the notion of a state complexity class, as
studied by Metger, Rosenthal, and Yuen \cite{RY21arxiv, MY23arxiv}.
Roughly speaking, these works define and study the complexity of
generating a given sequence $(\rho_n)_{n\in\N}$ of quantum
states. The main result of these works is showing that
$\mathbf{stateQIP} = \mathbf{statePSPACE}$, which is to say that the
sequences of states which can be efficiently generated with the help of
a possibly malicious prover are precisely those which can be generated
by polynomial-space quantum circuits. This is a state synthesis analogue
to the celebrated result of $\mathbf{QIP} = \mathbf{PSPACE}$
\cite{JJUW11}.

From our perspective, these prior works concerns the difficulty of creating one
copy of $\rho_n$ starting from no copies; we propose that it is
natural to also study the difficulty of creating two copies from one, or
more generally, $b$ copies starting from $a$ copies for
$a>b,a,b\in\N$.\footnote{%
	Recent work has also studied the quantum \emph{Kolmogorov}
	complexity of cloning a quantum state \cite{LFM+23arxiv}.
	The authors obtain results showing that, from this perspective,
	cloning a quantum state is almost always essentially as hard as
	creating it.}

\begin{definition}
\label{df:state-complexity}
	Let $\mc{C}$ be a class of circuit families, $\delta : \N \to \R$
	be a function, and $a,b \in \N$ be two non-negative integers. A
	sequence of states $(\rho_n)_{n\in\N}$ is in the state complexity
	class $\mathbf{state}^{a \to b}_\delta(\mc{C})$ if there exists a
	circuit family $C \in \mc{C}$ and an $n' \in \N$ such that
	\begin{equation}
		n \geq n'
		\implies
		\frac{1}{2}
		\norm{
			\rho_n^{\tensor b}
			-
			C_n(\rho_n^{\tensor a})
		}_\text{Tr}
		\leq
		\delta(n).
	\end{equation}
	We also define
	\begin{equation}
		\mathbf{state}^{a \to b}(\mc{C})
		=
		\bigcap_{k \in \N} \mathbf{state}_{n^{-k}}^{a \to b}(\mc{C})
		.
	\end{equation}
	If $a = 0$ and $b = 1$, we may omit the $a \to b$ superscript from
	all the above notation.
\end{definition}

Note that in full generality, we could also have $a$ and $b$ in the
definition be maps from $\N$ to $\N$ but this will not be necessary for
our observations.

Naturally occurring classes of circuit families which may be of interest
include the class of uniform circuit families, which we denote
$\mc{C}_\text{UNIF}$, and the class of polynomial-time circuit
families, which we denote $\mc{C}_\text{POLY}$. Letting
$\mc{C}_\text{PSPACE}$ denote the class of ``space-uniform
circuits'' as given in \cite[Definition 2.2]{MY23arxiv} and
$\mathbf{statePSPACE}$ be as given in \cite[Definition 2.4]{MY23arxiv},
we find, as expected, that
$\mathbf{state}(\mc{C}_\text{PSPACE}) = \mathbf{statePSPACE}$.%
\footnote{%
	Note that \cite{RY21arxiv, MY23arxiv} also require that a
	sequence $(\rho_n)_{n \in \N}$ in $\mathbf{statePSPACE}$ also
	satisfies the constraint that each
	$\rho_n$ is on precisely~$n$ qubits. However, they state that this
	is merely for convenience. We drop this requirement from our
	definition.}
We note that finding a class of circuit families, tentatively denoted
$\mc{C}_\text{QIP}$, such that $\mathbf{state}(\mc{C}_\text{QIP}) =
\mathbf{stateQIP}$, where the latter is as given in \cite[Definition
5.1]{MY23arxiv}, appears to be a bit more involved and non-trivial. In
short, it is unclear how to encode the required soundness property of
$\mathbf{stateQIP}$ into a class of circuit families. Nonetheless, we
believe \cref{df:state-complexity} to be a natural generalization of
state complexity classes in a wide range of interesting settings.

One interesting result of formalizing cloning complexity classes
is that it allows us to make more precise statements about the relations
between the complexity of generating and of cloning quantum states.
Specifically, we can now state and prove the three lemmas which were
sketched in \cref{fg:gen-vs-clone} from
\cref{sc:cloning-complexity-review}.

The first lemma formalizes the idea that, for essentially all classes of
circuit families, cloning a state cannot be more difficult than
generating it. For technical reasons, this lemmas does not apply to
certain pathological classes of circuit families (\emph{e.g.}: those
without the ability to apply the identity channel), or those representing
extremely limited computational power (\emph{e.g.}: those with a hard
bound on the number of qubits they can process). The idea is that if
one is able to generate a sequence of states, then one way to clone the
same sequence is to ignore the input and simply generate another copy.
\begin{lemma}
\label{th:state-inclusion}
\label{th:clone-gen-bound}
	Let $\mc{C}$ be a class of circuit families satisfying the following
	criteria:
	\begin{itemize}
		\item
			The class $\mc{C}$ can implement the identity on qubits it
			could generate. More precisely, if $\mc{C}$ includes a
			$(0,\ell)$-circuit family $C$ for some
			$\ell : \N \to \N$, then $\mc{C}$ also includes an
			$(\ell, \ell)$-circuit family $C^\text{id}$ where every
			circuit in this family implements the identity channel.
		\item
			The class $\mc{C}$ is closed under tensor products.
			More precisely, given two circuit families $(C_n)_{n \in
			\N}$ and $(C'_n)_{n \in \N}$ in $\mc{C}$, then the family
			$(C_n \tensor C'_n)_{n \in \N}$ is in $\mc{C}$.
	\end{itemize}
	Then, for all $\delta: \N \to \R$ and $a,b \in \N$ we have that
	\begin{equation}
		\mathbf{state}^{0 \to 1}_\delta(\mc{C})
		\subseteq
		\mathbf{state}^{1 \to 2}_\delta(\mc{C}).
	\end{equation}
\end{lemma}
\begin{proof}
	Let $
		(\rho_n)_{n \in \N}
		\in
		\mathbf{state}^{0 \to 1}_\delta(\mc{C})
	$. Let $C \in \mc{C}$ and $n' \in \N$ be such that
	\begin{equation}
		n \geq n'
		\implies
		\frac{1}{2}
		\norm{
			\rho_n
			-
			C_n(\varepsilon)
		}_\text{Tr}
		\leq
		\delta(n).
	\end{equation}
	Consider now the family $C' = (C^\text{id}_n\tensor C_n)_{n \in \N}$
	which, by our assumptions, is in $\mc{C}$. For all $n \geq n'$, we
	have that
	\begin{equation}
		\frac{1}{2}
		\norm{
			\rho_n^{\tensor 2}
			-
			C'_n(\rho_n)
		}_\text{Tr}
		=
		\frac{1}{2}
		\norm{
			\rho_n \tensor \rho_n
			-
			\rho_n \tensor C_n(\varepsilon)
		}_\text{Tr}
		\leq
		\frac{1}{2}
		\norm{
			\rho_n
			-
			C_n(\varepsilon)
		}_\text{Tr}
		\leq
		\delta(n)
	\end{equation}
	and so $(\rho_n)_{n\in\N}\in\mathbf{state}_\delta^{1\to 2}(\mc{C})$.
\end{proof}

The second lemma formalizes the idea that there exists sequences of
states which cannot be generated but which can be cloned, even with less
computational power. It suffices here to consider sequences of classical
information which cannot be consistently produced with high probability.
Weakly computably ingenerable sequences satisfy these properties.

\begin{lemma}
\label{th:hard-generate-easy-clone}
	Let $(s_n)_{n \in \N}$ be a weakly computably ingenerable
	$n$-sequence. Then, the sequence $(\ketbra{s_n})_{n \in \N}$ is in
	$\mathbf{state}^{1\to2}(\mc{C}_\text{POLY})$ but not in
	$\mathbf{state}(\mc{C}_\text{UNIF})$.
\end{lemma}
\begin{proof}
	By definition, $(\ketbra{s_n})_{n \in \N}$ is not in
	$\mathbf{state}(\mc{C}_\text{UNIF})$ since every uniform circuit
	family will produce the string $s_n$ with negligible probability for
	infinitely many values of $n$.

	On the other hand, let $(C_n)_{n \in \N}$ be an $(n, 2n)$-circuit
	family where $C_n$ simply
	appends $n$ qubits in the $\ket{0}$ state and then applies a
	CNOT gate on the $n + k$ qubit conditioned on the $k$
	qubit for all $k \in \{1, ..., n\}$. We see that $
		C_n(\ketbra{s_n})
		=
		\ketbra{s_n} \tensor \ketbra{s_n}
	$. Since $(C_n)_{n \in \N} \in \mc{C}_\text{POLY}$, we have that $
		(\ketbra{s_n})_{n \in \N}
		\in
		\mathbf{state}^{1\to2}(\mc{C}_\text{POLY})
	$.
\end{proof}

The third and final lemma states that there are sequences of states
which simply cannot be cloned by uniform circuit families. It is
obtained as a direct corollary from \cref{th:money-derandomized} and can
be interpreted as a cloning complexity analogue to the existence of
uncomputable functions.

\begin{lemma}
\label{th:hard-generate-hard-clone}
	Let $(s_n)_{n \in \N}$ be a computably ingenerable
	$2n$-sequence where we parse each
	$s_n$ as $(x_n, \theta_n)$ for two strings
	$x_n, \theta_n \in \{0,1\}^n$. Then,
	$
		(\ketbra*{x_n^{\theta_n}})_{n \in \N}
		\not\in
		\mathbf{state}^{1 \to 2}(\mc{C}_\text{UNIF})
	$. By \cref{th:clone-gen-bound}, this also implies that the sequence
	is not in $\mathbf{state}(\mc{C}_\text{UNIF})$.
\end{lemma}
\begin{proof}
	This is a direct corollary of \cref{th:money-derandomized}.
\end{proof}

Note that the corollary is a strictly weaker statement than the theorem.
Indeed, the theorem demonstrates that it is impossible to clone with
\emph{non-negligible} fidelity, via a uniform family of circuits, the
Wiesner states described by a computably ingenerable sequence. The
corollary only states that it is impossible to clone these states with
an \emph{overwhelming} fidelity. We emphasize that our results in
\cref{sc:advice} concerning uncloneable quantum advice will require the
stronger guarantee of the form given by the theorem.

%% file: 6-advice.tex
\section{Uncloneable Advice}
\label{sc:advice}

The main results of this section are to give a definition for the
complexity class of problems which can be solved by polynomial-time
quantum computations with negligible errors in the presence of advice
that is uncloneable and to give examples of problems in this class. We
denote this complexity class $\mathbf{neglQP}\text{/upoly}$.

As we briefly discussed in~\cref{sc:upoly-review}, we do not denote this
class $\mathbf{BQP}/\text{upoly}$; this is because we lack a generic
error reduction technique which can reduce bounded errors to negligible
errors, all the while maintaining the uncloneability property of the
advice states. In other words, the current state-of-the-art is that it
is possible that
$\mathbf{neglQP}/\text{upoly} \subseteq \mathbf{BQP}/\text{upoly}$ is a
strict containment.

\paragraph{Section overview.}
This section is organized as follows. We begin
by reviewing the basic notions of quantum copy-protection in
\cref{sc:copy-protection}. In \cref{sc:advice-df}, we formally define
the notion of uncloneable advice and give some basic remarks on this
definition. We also emphasize some parallels between quantum
copy-protection and uncloneable advice. In \cref{sc:promise}, we
describe a promise problem which unconditionally admits uncloneable
advice. In \cref{sc:language}, we describe a language with
uncloneable advice, assuming the feasibility of copy-protecting any one
sequence of maps satisfying fairly mild assumptions. In
\cref{sc:language-2}, we discuss one possible instantiation of the
construction presented in \cref{sc:language}.

\subsection{Quantum Copy-Protection}
\label{sc:copy-protection}

As discussed in \cref{sc:intro}, quantum copy-protection is broadly the
task of encoding a given function~$f$ as a quantum state~$\rho_f$ such
that the following two conditions are satisfied:
\begin{itemize}
	\item
		Correctness: A party with access to the quantum state $\rho_f$
		can correctly evaluate $f$ on any input by interacting with this
		state in a prescribed manner.
	\item
		Security: A party with access to the quantum state $\rho_f$
		cannot create and share a bipartite quantum state with two other
		separated parties such that both of these parties could
		correctly compute $f$ using any means available to them.
\end{itemize}

While correctness is required to hold for all inputs, security is
defined with respect to a sequence of random variables
$D = (D_n)_{n \in \N}$ jointly distributed on all functions considered
by the copy-protection scheme and pairs of inputs to these
functions. Security is attained if no efficient attacker can
split the program state for a function $f$, allowing both subsequent
parties to evaluate $f(x_B)$ and $f(x_C)$ with more than a negligible
advantage when $(f,x_B,x_C)$ is sampled from $D$.

We formalize the syntax and correctness guarantee of a copy-protection
scheme in \cref{df:cp-scheme} below. Security will then be formalized in
\cref{df:cp-secure}. Up to technical details of presentation, our
definition of the syntax and correctness of a copy-protection scheme
coincides with the definition given in \cite{CLLZ21}.\footnote{%
	The definitions given in \cite{CLLZ21} states that the map $f$
	should be pseudorandom function. However, they are applicable to
	any such family of maps, not only pseudorandom functions.}

\begin{definition}
\label{df:cp-scheme}
	Let
	\begin{equation}
	\label{eq:cp-maps}
		f = \left(
			f_n :
			\{0,1\}^{\kappa(n)} \times \{0,1\}^{d(n)}
			\to
			\{0,1\}^{c(n)}
		\right)_{n \in \N}
	\end{equation}
	be a sequence of maps whose domains and codomains are parameterized
	by maps $\kappa, d, c : \N \to \N$.

	A \emph{copy-protection scheme} for $f$ with program length
	$q : \N \to \N$ is a pair $(G, E)$ of efficient~$(\kappa,q)$- and
	$(q + d,c)$-circuit families, respectively.

	A copy-protection scheme is \emph{correct for $f$} if
	there exists a negligible function $\eta$ such that
	\begin{equation}
	\label{df:cp-scheme-eq}
		\bra{f_n(k,x)}
			E_n(G_n(k) \tensor x)
		\ket{f_n(k,x)}
		\geq
		1 - \eta(n)
	\end{equation}
	for all $n \in \N$, all $k \in \{0,1\}^{\kappa(n)}$, and
	all $x \in \{0,1\}^{c(n)}$.
\end{definition}

Note that the existence of a copy-protection scheme for a sequence of
maps $f$ as defined above implies that the maps $\kappa$, $d$, and $c$
are polynomially bounded and efficiently computable. If this was not the
case, there would not exist efficient $(\kappa, q)$- and
$(q + d, c)$-circuit families.

\begin{remark}
	The above definition can also be easily adapted to sequences of
	maps $(f_n)_{n \in \N}$ where, for all $n \in \N$, the domain of
	$f_n$ is a subset
	$S_n \subseteq \{0,1\}^{\kappa(n)} \times \{0,1\}^{d(n)}$.
	To do so, it suffices to only require the correctness condition,
	\cref{df:cp-scheme-eq}, to hold for pairs $(k,x) \in S_n$.
\end{remark}

\begin{remark}
	It is possible for a copy-protection scheme to be correct for two
	different sequences of maps.
	Indeed, assume $f = (f_n)_{n \in \N}$ and $f' = (f'_n)_{n \in \N}$
	are sequences of maps such that $f_n \not= f'_n$ for a non-zero but
	finite number of values of $n$. Then, it is straightforward to show
	that a copy-protection scheme $(G, E)$ is correct for $f$ if and
	only if it is correct for $f'$.
	In short, this is due to the fact that correctness is an asymptotic
	property and that $f_n = f'_n$ for all sufficiently large values of
	$n$.
	
	However, this is the maximal possible discrepancy between $f$
	and $f'$ in the sense that if $(G, E)$ is correct for $f$ and $f'$,
	then there must exist a $\tilde{n} \in \N$ such that $n \geq
	\tilde{n} \implies f_n = f'_n$.
	Indeed, by the assumed correctness of the scheme for $f$ and $f'$,
	the triangle inequality for the trace distance, the
	Fuchs-van de Graaf inequalities, and basic properties of negligible
	functions, there exists a negligible function $\eta$ such that
	\begin{equation}
	\begin{aligned}
		\frac{1}{2}\norm{f_n(k,x) - f'_n(k,x)}_1
		&\leq
		\sqrt{1 - \bra{f_n(k,x)} \rho_{n,k,x} \ket{f_n(k,x)}}
		+
		\sqrt{1 - \bra{f'_n(k,x)} \rho_{n,k,x} \ket{f'_n(k,x)}}
		\\&\leq
		\eta(n)
	\end{aligned}
	\end{equation}
	where $\rho_{n,k,x} = E_n(G_n(k) \tensor x)$ for all $n \in \N$, all
	$k \in \{0,1\}^{\kappa(n)}$, and all $x \in \{0,1\}^{d(n)}$.
	It then follows that $f_n = f'_n$ for all sufficiently large values
	of $n$ since $\eta(n) < 1 \implies f_n(k,x) = f'_n(k,x)$.
\end{remark}

An attack against a copy-protection scheme is a triplet of
efficient circuit families $(A, B, C)$.
We will sometimes refer to the first, $A$, as the \emph{splitting
adversary} and the later two, $B$ and $C$, as the \emph{guessing
adversaries.}
This reflects their respective tasks in the game, described below, used
to define and quantify the notion of security for a copy-protection
scheme.

\paragraph{The Copy-Protection Game}
Let $f = (f_n)_{n \in \N}$ be a sequence of maps as in
\cref{df:cp-scheme}, let~$(G, E)$ be a copy-protection scheme for $f$,
and let $D = (D_n)_{n \in \N}$ be a sequence of random variables where,
for all $n \in \N$, each $D_n$ is distributed on the set
$\{0,1\}^{\kappa(n)}\times\{0,1\}^{d(n)}\times\{0,1\}^{d(n)}$.
The copy-protection security game, for a given $n \in \N$ in addition to
the parameters described above, is played by a Referee against
collaborating Alice, Bob, and Charlie --- collectively known as the
adversaries --- as follows:
\begin{enumerate}
	\item
		The Referee samples a triplet $(k, x_B, x_C) \gets D_n$.
	\item
		The Referee prepares the state $\rho = E_n(k)$ and gives it to
		Alice.
	\item
		Alice prepares a bipartite quantum state $\rho$ and gives one
		part to Bob and the other to Charlie.
		This is the only communication between Alice, Bob, and Charlie
		that occurs during the game.
	\item
		The Referee gives $x_B$ to Bob and $x_C$ to Charlie.
		Note that Bob does not receive $x_C$ and Charlie does not
		receive $x_B$.
	\item
		Bob, with what they have received from the Referee and Alice,
		outputs a string $y_B$. Similarly, Charlie outputs a string
		$y_C$.
	\item
		Alice, Bob, and Charlie win if and only if $f_n(k,x_B) = y_B$
		\emph{and} $f_n(k,x_C) = y_C$.
\end{enumerate}

Formally, we will model Alice, Bob, and Charlie as efficient circuit
families $A = (A_n)_{n \in \N}$, $B = (B_n)_{n \in \N}$, and
$C = (C_n)_{n \in \N}$ respectively.
Note that, at this point, we do not impose any other type of
computational restraints on the Referee. In particular, we do not
require the random variables $D = (D_n)_{n \in \N}$ to admit an
efficient sampling procedure.
We formalize this game and the winning probability of the adversaries in
the next definition.

\begin{definition}
\label{df:cp-attack}
	Let $(G, E)$ be a copy-protection scheme with program lengths
	$q : \N \to \N$ for a sequence of maps $f = (f_n)_{n \in \N}$ as
	given in \cref{df:cp-scheme}.
	An attack against this scheme is a triplet $(A, B, C)$ of efficient
	$(q, q_B + q_C)$-, $(d + q_B, c)$-, and $(q_C + d, c)$- circuit
	families, respectively, for two maps~$q_B, q_C : \N \to \N$.
	For any $n \in \N$ and triplet $
		(k,x_B,x_C)
		\in
		\{0,1\}^{\kappa(n)} \times \{0,1\}^{d(n)} \times \{0,1\}^{d(n)}
	$, we define the state
	\begin{equation}
		\rho^{(A,B,C)}_{(G,E),n}(k,x_B,x_C)
		=
		\left(B_n \tensor C_n\right)
		\left(
			x_B
			\tensor
			\left(A_n \circ G_n\right)(k)
			\tensor
			x_C
		\right).
	\end{equation}

	Let $D = (D_n)_{n \in \N}$ be a sequence of random variables where,
	for each $n \in \N$, $D_n$ is distributed over the set
	$\{0,1\}^{\kappa(n)} \times \{0,1\}^{d(n)} \times \{0,1\}^{d(n)}$.
	For any attack $(A,B,C)$ against $(G,E)$, we define its
	winning probability with respect to $f$ and $D$, denoted
	$w^{(A,B,C)}_{(G,E),f,D} : \N \to \R$, as the function
	\begin{equation}
	\label{df:cp-attack-eq}
		n
		\mapsto
		\sum_{(k,x_B,x_C) \in S_n'}
		\Pr\left[D_n = (k,x_B,x_C)\right]
		\cdot
		\bra{f_n(k,x_B), f_n(k,x_C)}
			\rho_{(G,E),n}^{(A,B,C)}(k,x_B,x_C)
		\ket{f_n(k,x_B), f_n(k,x_C)}
	\end{equation}
	where $S_n' \subseteq \{0,1\}^{\kappa(n)} \times \{0,1\}^{d(n)}
	\times \{0,1\}^{d(n)}$ is precisely the set of all triplets
	$(k,x_B,x_C)$ where both $f_n(k,x_B)$ and $f_n(k,x_C)$ are defined.
\end{definition}

If every map $f_n$ in $f$ is defined on the entire
set $\{0,1\}^{\kappa(n)} \times \{0,1\}^{d(n)}$, then the summation over
$S'_n$ and the $\Pr\left[D_n = (k, x_B, x_C)\right]$ factor in
\cref{df:cp-attack-eq} is simply an alternative expression for the
expectation over sampling $(k, x_B, x_C)$ from $D_n$.
Such an expectation is typically how the winning probability of the
adversaries is presented in the literature.
However, our definition is also applicable to cases where some maps
$f_n$ are defined only on a strict subset of
$\{0,1\}^{\kappa(n)} \times \{0,1\}^{d(n)}$ and where the support of
$D_n$ is a strict superset of $S'_n$.
In these cases, the expectation would not be well defined.

It may seem unnatural to consider $w^{(A,B,C)}_{(G,E),f,D}$ for
sequences $f$ and $D$ where $D_n$ occasionally yields pairs $(k,x)$
outside of the domain of $f_n$, but our generalization
allows us to easily make sense of the winning probability of the
adversaries when~$f_n : \varnothing \to \{0,1\}^{c(n)}$ is the empty
map.
Note that it is impossible to even define a random variable distributed
over the domain of $f_n$ or, more precisely, over $S'_n = \varnothing$,
in this case.
Note that we can interpret appearances of the empty map in a sequence
$(f_n)_{n \in \N}$ as corresponding to cases where the map is
undefined.\footnote{%
	While it may be tempting to simply remove all instances of empty
	maps in a sequence $(f_n)$ and then re-index it, yielding a new
	sequence $(g_n)$, this can can have a material implications to the
	study of copy-protection for these maps as it impacts the scaling of
	the computational power permitted to schemes and adversaries.
	Indeed, suppose that the map $f_n$ is defined if and only if
	$n = 2^k$ for some $k \in \N$.
	Under the proposed re-indexing scheme, $g_n = f_{2^n}$ for all $n
	\in \N$.
	In particular, if we subject the circuit families of an efficient
	copy-protection scheme $(G, E)$ for $(f_n)_{n \in \N}$ to the same
	re-indexing, which would be the naive way to try to obtain a scheme
	for $(g_n = f_{2^n})_{n \in \N}$, they may become exponential-time
	circuit families.
	}
This will be useful when we later relate the notions of copy-protection
and uncloneable advice.

In general, we will only be interested in sequences $f$ and $D$ where
the support of $D_n$ is a subset of $S'_n$, unless $f_n$ is the empty
map.

\begin{remark}
	Let $n \in \N$ be such that $f_n : \varnothing \to \{0,1\}^{c(n)}$
	is the empty map, i.e.: $S'_n = \varnothing$.
	Then, adopting the convention that a sum over an empty set is $0$,
	we have that $\omega^{(A,B,C)}_{(G,E),f,D}(n) = 0$.
\end{remark}

The fact that $\omega^{(A,B,C)}_{(G,E),f,D}(n) = 0$ whenever $f_n$ is
the empty map, meaning that it is ``in practice undefined'', implies
that an adversary $(A, B, C)$ breaking the potential security of a
copy-protection scheme must do so for values of $n$ where $f_n$ is ``in
practice defined''.
We see this as a desirable property of our definition.

Following a standard cryptographic paradigm, we now wish to define the
security for a copy-protection scheme as the property that all efficient
adversaries have at most a negligible advantage over some trivially
attainable success probability.
The issue here is to identify what precisely is the trivial success
probability in the above described game.

There are a few attacks against copy-protection schemes which cannot
reasonably be prevented.
One such attacks, for all $n \in \N$, consists of having the splitting
adversary $A_n$ give the complete and unmodified program state $G_n(k)$
to the guessing party $B_n$ and giving nothing to $C_n$.
When the adversary $B_n$ receives their challenge $x_B$, they can then
honestly evaluate the program state $G_n(k)$ on $x_B$ to obtain, with
overwhelming probability assuming that the scheme $(G, E)$ is correct
for $f$, the proper answer.
The $C_n$ adversary, for their part, will simply output a uniformly
random string sampled from $\{0,1\}^{c(n)}$, the codomain of the maps
under consideration.
It is easy to see that the winning probability of this trivial attack is
$n \mapsto 2^{-c(n)} - \eta(n)$ for some negligible function $\eta$
accounting for the possibility of an incorrect evaluation by $B$.
Thus, we take $2^{-c(n)}$ as our trivial success probability.

\begin{definition}
\label{df:cp-secure}
	Let $f = (f_n)_{n \in \N}$ be a sequence of maps as given in
	\cref{df:cp-scheme}, $(G, E)$ be a copy-protection scheme for $f$,
	and $D = (D_n)_{n \in \N}$ be a sequence of random
	variables as given in \cref{df:cp-attack}.
	The copy-protection scheme $(G,E)$ is secure with respect to $f$ and
	$D$ if for every attack $(A, B, C)$ against it, the function
	\begin{equation}
		n \mapsto w^{(A,B,C)}_{(G, E), f, D}(n) - 2^{-c(n)}
	\end{equation}
	is negligible.
\end{definition}

Note that other works present definitions based on alternate notions of trivial success probabilities. 
For example, \cite{CMP20arxiv} permits the splitting adversary $A_n$ to
optimally choose which of $B_n$ or $C_n$ should received the unmodified
program state and further allows the non-receiving guessing adversary to
learn their challenge value~$x$ before making their guess for~$f_n(k, x)$. As this can only increase the value of the trivial winning
probability above $2^{-c(n)}$, the security notion we present here is
stronger.
As another example, \cite{CLLZ21} defines security with respect to a
trivial winning probability which is $0$ or negligible. While this is in
agreement with our definition if $n \mapsto 2^{-c(n)}$ is negligible, it
is an impractical definition to use when this is not the case, such as
if $c$ is constant function.

We leave to future work a more detailed comparison of existing security notions together with their notions of trivial success probabilities.

\subsection{Defining Uncloneable Advice}
\label{sc:advice-df}

In this section, we briefly review the necessary preliminaries on
complexity theory and then define the complexity classes
$\mathbf{neglQP}$ and $\mathbf{neglQP}/\text{upoly}$. We generally
follow the definitions and conventions given in Watrous' survey of
quantum complexity theory \cite{Wat09}.

First, we recall the definitions of a promise problem and of a language.

\begin{definition}
\label{df:problem}
	A \emph{promise problem} $P = (P_0, P_1)$ is a pair of disjoint
	subsets $P_0, P_1 \subseteq \{0,1\}^*$. The elements of
	$P_1$ are called the \emph{yes instances} and those of $P_0$ are the
	\emph{no instances}. Both \emph{yes} and \emph{no} instances are
	called \emph{problem instances}. If $P_0 \cup P_1=\{0,1\}^*$, we say
	that $P$ is a \emph{language} and the elements of $P_1$ are the
	\emph{words} of this language. We also establish some additional
	notation by overloading the symbol $P$ twice: first as a set, then
	as a map.

	First, we let $P = P_0 \cup P_1$ be the set of all problem
	instances.
	For all $n \in \N$, we also define the sets $P^n = P \cap \{0,1\}^n$,
	$P_0^n = P_0 \cap \{0,1\}^n$, and $P_1^n = P_1 \cap \{0,1\}^n$ to be
	all problem instances of length $n$, all \emph{no}
	instances of length $n$, and all \emph{yes} instances of length $n$.

	Second, we let $P : P_0 \cup P_1 \to \{0,1\}$ be the unique map
	satisfying $P(x) = 1 \iff x \in P_1$.
	We identify computing the
	map $P$ with the ability to \emph{solve} the problem $P$.
	For all $n \in \N$, we also define
	$P^n : P_0^n \cup P_1^n \to \{0,1\}$ to be the unique map
	satisfying $P^n(x) = 1 \iff x \in P_1^n$.
\end{definition}

Note that, with the notation of \cref{df:problem}, a problem $P = (P_0,
P_1)$ is completely characterized by the resulting sequence of maps
$(P^n : P_0^n \cup P_1^n \to \{0,1\})_{n \in \N}$.

Before proceeding to defining our novel complexity classes, we recall
the classes of problems which can be solved in quantum polynomial-time
with bounded errors, with or without quantum advice. As usual, we denote
these classes by~$\BQP$ and $\BQP/\text{qpoly}$.

\begin{definition}
	Let $a,b : \N \to \R$ be two maps. A promise problem $P$ is in the
	class $\BQP(a,b)$ if there exists a polynomial-time
	$(n,1)$-circuit family $(C_n)_{n \in \N}$ such
	that the following hold:
	\begin{enumerate}
		\item
			For all $x \in P_1$, we have that
			$\bra{1} C_\abs{x}(x) \ket{1} \geq a(\abs{x})$.
		\item
			For all $x \in P_0$, we have that
			$\bra{1} C_\abs{x}(x) \ket{1} \leq b(\abs{x})$.
	\end{enumerate}
\end{definition}

\begin{definition}
	Let $a,b : \N \to \R$ be a map. A promise problem $P$ is in the
	class $\BQP(a,b)\text{/qpoly}$ if there exists a
	sequence of states $(\rho_n)_{n \in \N}$, each on $q(n)$ qubits for
	a polynomially bounded $q : \N \to \N$, and a polynomial-time
	$(q + n,1)$-circuit family $(C_n)_{n \in \N}$ such that
	the following hold:
	\begin{enumerate}
		\item
			For all $x \in P_1$, we have that
			$\bra{1} C_\abs{x}(\rho_\abs{x} \tensor x) \ket{1}
			\geq a(\abs{x})$.
		\item
			For all $x \in P_0$, we have that
			$\bra{1} C_\abs{x}(\rho_\abs{x} \tensor x) \ket{1} \leq
			b(\abs{x})$.
	\end{enumerate}
\end{definition}

By a standard error-reduction-by-repetition argument, we can show that
there is a large flexibility in the choices of $a$ and $b$ without
changing the underlying class, so long as these are sufficiently
bounded. Thus, we define $\mathbf{BQP} = \mathbf{BQP}(2/3,1/3)$ and
$\mathbf{BQP}/\text{qpoly} = \mathbf{BQP}(2/3,1/3)/\text{qpoly}$.

We are now ready to define our novel complexity classes. We begin by
defining the class of problems which can be solved by polynomial-time
quantum computations with negligible errors. We denote this class
$\mathbf{neglQP}$ and immediately note that it is equivalent
to $\mathbf{BQP}$ due to standard error reduction techniques for this
latter class.

\begin{definition}
	A problem $P$ is in the complexity class $\mathbf{neglQP}$ if there
	exists a polynomial-time circuit family
	$(C_n)_{n \in \N}$ and a negligible function $\eta$ such
	that
	\begin{equation}
		x \in P \implies
		\bra{P(x)} C_\abs{x}(x) \ket{P(x)}
		\geq
		1 - \eta(\abs{x}).
	\end{equation}
\end{definition}

\begin{lemma}
	$\mathbf{neglQP} = \mathbf{BQP}$.
\end{lemma}
\begin{proof}
	The inclusion $\mathbf{BQP} \subseteq \mathbf{neglQP}$ follows
	directly by the standard error reduction technique for
	$\mathbf{BQP}$. It then suffices to show that the inclusion
	$\mathbf{neglQP} \subseteq \mathbf{BQP}$ also holds.
	
	Consider a problem $P \in \mathbf{neglPQ}$ which is solved by a
	polynomial-time circuit family $(C_n)_{n \in \N}$
	with negligible error $\eta$. In particular, there exists an
	$n_0 \in \N$ such that $n \geq n_0 \implies \eta(n)\leq\frac{1}{3}$.
	Thus, we can consider a new circuit family $(C'_n)_{n}$
	where each $C'_n$ for $n < n_0$ simply ``hard codes''
	each solution to each problem instance. Otherwise, if
	$n \geq n_0$, we simply take $C'_n = C_n$.
	It is easy to see that $(C'_n)_{n \in \N}$ is a polynomial-time
	circuit family which solves $P$ with error at most~$\frac{1}{3}$.
	Hence, we have that $P \in \mathbf{BQP}$ and so that
	$\mathbf{neglQP} \subseteq \mathbf{BQP}$.
\end{proof}

We can now state our definition for the complexity class of problems
which can be solved with negligible error by polynomial-time quantum
computation with the help of \emph{uncloneable} advice, which we denote
$\mathbf{neglQP}\text{/upoly}$.

As we have previously discussed in \cref{sc:upoly-review}, a problem $P$
is in $\mathbf{neglQP}\text{/upoly}$ if there exists a sequence of
advice states $(\rho_{n})_{n \in \N}$ satisfying the
following two criteria. First, an honest user must be able to solve
problem instances in $P$ with negligible errors when given a single copy
of the advice state. Second, a malicious user given a single copy of the
advice state cannot share this state between two other non-communicating
malicious users such that they could both solve problem instances in $P$
with more than a negligible advantage. We highlight once again that this
is analogous to the security criteria of a copy-protection scheme as set
out in \cref{df:cp-secure}.

The following definition formalizes the above.

\begin{definition}
\label{df:neglqp/upoly}
	A problem $P = (P_0, P_1)$ is in the complexity class
	$\mathbf{neglQP}/\text{upoly}$ if there exists a
	sequence of quantum states $(\rho_n)_{n \in \N}$,
	each on $q(n)$ qubits respectively for a polynomially bounded map
	$q : \N \to \N$, such that the following two conditions hold:
	\begin{enumerate}
		\item
			\emph{Correctness.}
			There exists a polynomial-time
			$(q + n,1)$-circuit family
			$(C_n)_{n \in \N}$
			and a negligible function $\eta$ such that
			\begin{equation}
				x \in P
				\implies
				\bra{P(x)}
					C_\abs{x}\left(
						\rho_\abs{x} \tensor x
					\right)
				\ket{P(x)}
				\geq
				1 - \eta(\abs{x})
				.
			\end{equation}
		\item
			\emph{Uncloneability.}
			For each $n \in \N$, let
			\begin{equation}
				k_n
				=
				\begin{cases}
					1 & \text{if either } P_0^n \text{ or } P_1^n \text{ is
					empty}
					\\
					\frac{1}{2} & \text{else.}
				\end{cases}
			\end{equation}
			and let $D_n$ be a random
			variable distributed on $\{0,1\}^n$ such that for
			both $b \in \{0,1\}$ we have that
			\begin{equation}
				x \in P_b^n
				\implies
				\Pr[D_n = x]
				=
				k_n
				\cdot
				\frac{1}{\abs{P_b^n}}.
			\end{equation}
			Then, for all triplets $
				\left(
					A,
					B,
					C
				\right)
			$ of polynomial-time
			$(q ,q_B + q_C)$-,
			$(n + q_B, 1)$-, and
			$(q_C + n, 1)$-circuit families,
			respectively, there exists a negligible function~$\eta'$
			such that for all $n \in \N$ satisfying
			$P^n \not= \varnothing$ we have that
			\begin{equation}
			\label{eq:neglqp/upoly}
				\E_{(x_B,x_C) \gets D_n \times D_n}
				\bra{P(x_B), P(x_C)}
					\left(B_n \tensor C_n\right)
					\left(
						x_b
						\tensor
						A_n(\rho_n)
						\tensor
						x_c
					\right)
				\ket{P(x_B), P(x_C)}
				\leq
				\frac{1}{2}
				+
				\eta'(n).
			\end{equation}
	\end{enumerate}
\end{definition}

Note that $\frac{1}{2}$ in the equation above plays the same role as
the $2^{-c(n)}$ term in the definition of security for
copy-protection: it captures the maximal winning probability of
strategies which give the advice unmodified to one guessing party and
has the other guessing party output an element uniformly at random from
the appropriate codomain.

We give a few other remarks on this definition.

\begin{remark}
	It would be possible to strengthen \cref{df:neglqp/upoly} by
	weakening the computational constraints on the adversaries
	considered in uncloneability criterion. While we do not formally
	define these variations here, we note that the advice we construct
	in the proof of \cref{th:promise} for the particular promise problem
	defined in that theorem would fulfill the uncloneability criterion
	against any uniform adversaries, not only polynomial-time
	adversaries.
\end{remark}

\begin{remark}
	For a given problem $P$, the sequence of distributions
	$(D_n)_{n \in \N}$ considered in point 2 of
	\cref{df:neglqp/upoly} is not uniquely defined. Indeed, for a
	particular $n \in \N$, the random variable $D_n$
	is uniquely defined if $P\cap\{0,1\}^n\not=\varnothing$ and is
	unconstrained otherwise.
	However, our definition is insensitive to $D_n$ for the values of
	$n$ where $P \cap \{0,1\}^n = \varnothing$ and so
	this ambiguity is not an issue.
\end{remark}

\begin{remark}
	Our definition of the sequence of distributions $D$ in point 2 of
	\cref{df:neglqp/upoly} could be changed, leading to possibly
	distinct complexity classes.

	Note that while we define random variables $D_n$ over
	$P^n$, in \cref{eq:neglqp/upoly} we actually use the
	random variables $D_n \times D_n$ over
	$\{0,1\}^n \times \{0,1\}^n$. In particular, any other
	definition of a sequence of distributions
	$(D'_n)_{n \in \N}$ on $\{0,1\}^n \times \{0,1\}^n$, provided that
	it satisfy
	\begin{equation}
	\label{eq:neglqp/upoly-consistency}
		\forall n \in \N
		\qq{}
		P^n \not= \varnothing
		\implies
		\Pr[
			D'_n
			\in
			P^n \times P^n
		]
		=
		1,
	\end{equation}
	could replace the sequence $(D_n \times D_n)_{n
	\in \N}$ in \cref{df:neglqp/upoly} and the definition would remain
	coherent. We note, however, that the choice of $k_n$ on the
	right-hand side of \cref{eq:neglqp/upoly} may no longer be the most
	appropriate for certain choices of distributions. Requiring the
	distributions to satisfy \cref{eq:neglqp/upoly-consistency} simply
	enforces the condition that the adversaries should be challenged on
	elements of $P$ with certainty.

	We emphasize here that the distributions $D_n$ depend on the
	problem $P$. If we wish to consider promise problems which are not
	languages, then it seems natural that the distributions depend on
	$P$. Indeed, if the distributions did not depend on $P$, it is
	unclear what we should ask of the adversaries when they are
	challenged on inputs which are not in $P$. While there are sensible
	answers to this question, such as considering any output to be
	valid or not consider these cases when computing the success
	probability of the adversaries, we leave further discussions along
	these lines to future work.

	We highlight three other natural ways in which we could have defined
	the distributions from which the problem instances $(x_0,x_1)$ are
	sampled. These emerge from the four possible ways to answer the
	following two questions:
	\begin{itemize}
		\item
			Should $x_0$ and $x_1$ be equal, or sampled independently?
		\item
			Should the challenges be sampled uniformly at random from
			$P^n$, or in such a way which makes $0$
			and $1$ equally likely to be the correct answers (when
			possible)?
	\end{itemize}
	Our choice for the distributions results from taking the latter
	answer to both questions.

	Finally, we note that this discussion on which distributions to use
	when challenging the adversaries for the uncloneability criterion
	echoes extremely similar discussions pertaining to secure software
	leasing \cite{ALP21,BJL+21} and copy-protected functions
	\cite{CMP20arxiv}.
\end{remark}

With \cref{df:neglqp/upoly} in hand, we can now formalize the relation
between copy-protection and uncloneable advice.
Let $(G, E)$ be a pair of circuit families forming a copy-protection
scheme as defined in \cref{df:cp-scheme}, except that $G$ need not be
efficient, or even uniform: $G$ can be an arbitrary circuit family.
Call such a scheme a \emph{copy-protection scheme with unconstrained
generation}. Correctness and security for copy-protection schemes with
unconstrained generation is exactly as defined for usual
copy-protection schemes.
The next theorem states that a problem $P$ is in
$\mathbf{neglQP}\text{/upoly}$ if and only if there exists a
copy-protection with unconstrained generation $(G, E)$ which is correct
and secure for the sequence of maps
$(P^n : P^n_0 \cup P^n_1 \to \{0,1\})_{n \in \N}$ with respect to random
variables $D = (D_n)_{n \in \N}$ as described in \cref{df:neglqp/upoly}.

\begin{theorem}
	Let $P$ be a problem and let $D = (D_n)_{n \in \N}$ be a sequence of
	random variables as described in \cref{df:neglqp/upoly}.
	For all $n \in \N$, let $
		\tilde{P}^n :
		\{0,1\}^0 \times \left(P_0^n \cup P_1^n\right) \to \{0,1\}
	$ be the map defined by $\tilde{P}^n(\varepsilon, x) = P^n(x)$ for
	all $x \in P_0^n \cup P_1^n$ and
	let $\tilde{D}_n$ be the random variable defined by
	\begin{equation}
		\Pr\left[\tilde{D}_n = (\varepsilon, x_B, x_C)\right]
		=
		\Pr\left[D_n \times D_n = (x_B, x_C)\right].
	\end{equation}
	Then, $P$ is in $\mathbf{neglQP}\text{/upoly}$ if and only if there
	exists a copy-protection scheme with unconstrained generation $(G,
	E)$ for
	the sequence of maps $\tilde{P} = (\tilde{P}^n)_{n \in \N}$ which is
	correct and secure with respect to the distributions
	$\tilde{D} = (\tilde{D}_n)_{n \in \N}$.
\end{theorem}
\begin{proof}[Prood sketch]
	Assume that $P$ is in $\mathbf{nelgQP}\text{/upoly}$ by using the
	sequence of advice states $(\rho_n)_{n \in \N}$.
	By the correctness condition of $\mathbf{neglQP}\text{/upoly}$,
	there exists an efficient circuit family $C = (C_n)_{n \in \N}$ and
	a negligible function $\eta$ such that $x \in P \implies \bra{P(x)}
	C_n(\rho_{\abs{x}} \tensor x) \ket{P(x)} \geq 1 - \eta(\abs{x})$.
	Moreover, there exists a circuit family $G = (G_n)_{n \in \N}$ such
	that $\frac{1}{2}\norm{\rho_n - G_n(\varepsilon)}_1 \leq
	2^{-n}$.\footnote{It may be that it is impossible to
	achieve $\rho_n = G_n(\varepsilon)$ for all $n$ due to the
	underlying gate set from which $G$ is constructed. This negligible
	error is not an issue as both correctness and security of
	uncloneable advice is defined up to negligible errors.}
	Consider now $(G,C)$ as a copy-protection scheme with unconstrained
	generation for $(\tilde{P}^n)_{n \in \N}$.
	We note two things:
	\begin{enumerate}
		\item
			The correctness of the advice directly implies the
			correctness of $(G,C)$ for $(\tilde{P}^n)_{n \in \N}$.
		\item
			The uncloneability of the advice implies that the scheme
			$(G, C)$ is
			secure for $(\tilde{P}^n)_{n \in \N}$ with respect to
			$(\tilde{D}_n)_{n \in \N}$.

			Indeed, adversaries against the scheme $(G, C)$ are
			precisely adversaries against the advice states
			$(\rho_n)_{n \in \N}$ for the uncloneability criterion of
			$\mathbf{neglQP}\text{/upoly}$ and the winning probability
			of such an adversary is identical in both scenarios for
			those values of $n$ satisfying $P^n \not= \varnothing$.
			As the advice is uncloneable, we have
			that for all adversaries $(A, B, C)$ there exists a
			negligible function $\eta$ such that
			$P^n \not= 0 \implies
			\omega^{(A,B,C)}_{(G,C),\tilde{P},\tilde{D}}(n) \leq
			\frac{1}{2} + \eta(n)$.
			Moreover, by definition, we have that
			$\omega^{(A,B,C)}_{(G,C),\tilde{P},\tilde{D}}(n) = 0$
			whenever $P^n = \varnothing$.
			It follows that $n \mapsto
			\omega^{(A,B,C)}_{(G,C),\tilde{P},\tilde{D}} - \frac{1}{2}$
			is negligible and so $(G,C)$ is secure for $\tilde{P}$ with
			respect to $\tilde{D}$.
	\end{enumerate}

	For the other direction, assume that $(G, E)$ is a copy-protection
	scheme with unconstrained generation for $(\tilde{P}^n)_{n \in \N}$
	which is correct and secure with respect to $(\tilde{D}_n)_{n \in
	\N}$. Now, for all $n \in \N$ define the states $\rho_n =
	G_n(\varepsilon)$. It follows that we can show that $P$ is in
	$\mathbf{neglQP}\text{/upoly}$ with the help of the advice states
	$(\rho_n)_{n \in \N}$ since correctness and security of the
	copy-protection scheme directly imply correctness and security of
	the uncloneable advice.
\end{proof}

As a final remark for this section, we note that the intersection of
$\BQP$ and $\mathbf{neglQP}/\text{upoly}$ can only contain essentially
trivial problems.

\begin{proposition}
	A problem $P$ is in
	$\mathbf{BQP} \cap \mathbf{neglQP}/\text{upoly}$ if and only if
	$\abs{P}$ is finite.
\end{proposition}

The proof of this proposition reduces to two ideas.
First, if $P \in \mathbf{BQP}$ then it can solved with the
empty state $1 \in \mc{D}(\C)$ as advice.
Second, the empty state is perfectly cloneable.
Or, more precisely, two copies of the empty state can be easily produced
from any other state.

\begin{proof}

	Assume that $P$ is a problem such that $\abs{P}$ is finite.
	Then, it is trivially true that $P \in \mathbf{BQP}$.
	This can be shown by considering a circuit family $(C_n)_{n \in \N}$
	where each $C_n$ has a ``hard-coded'' list of the elements of
	$P_0^n$ which can then be used to solve all problem instances of
	length $n$ by simply checking if the instance is in $P_0^n$ or not.
	As $\abs{P}$ is finite, this can be implemented efficiently.
	On the other hand, let $m \in \N$ be the maximal length of the
	elements in $P$, i.e.: $x \in P \implies \abs{x} \leq m$.
	Consider the negligible function $\eta : \N \to \R$ defined by
	$\eta(n) = 1$ if $n \leq m$ and $\eta(n) = 0$ if $n > m$ and the
	sequence of empty states $(\rho_n = 1)_{n \in \N}$.
	Correctness and uncloneability as given in \cref{df:neglqp/upoly}
	holds with respect to this sequence of states and this negligible
	function.
	Thus, $P \in \mathbf{neglQP}\text{/upoly}$.

	Assume that $P \in \mathbf{BQP}$ and $P \in
	\mathbf{neglQP}\text{/upoly}$ where the second inclusion is obtained
	with respect to the sequence of advice states $(\rho_n)_{n \in \N}$.
	Let $C = (C_n)_{n \in \N}$ be an efficient $(n,1)$-circuit family and
	$\eta$ be a negligible function such that
	$x \in P \implies \bra{P(x)} C_\abs{x}(x) \ket{P(x)} \geq 1 -
	\eta(\abs{x})$.
	Such a negligible function and circuit family exists by virtue of
	$P$ being in $\mathbf{BQP}$.
	Consider now the triplet of efficient circuit families $(A, C, C)$
	where each circuit $A_n$ simply discards its input.
	Finally, let $(D_n)_{n \in \N}$ be a sequence of random variables as
	defined in \cref{df:neglqp/upoly}.
	As the circuit family $C$ solves $P$ with negligible error on all
	inputs and $P$ is assumed to be in $\mathbf{neglQP}\text{/upoly}$,
	there exists another negligible function $\eta'$ such that for all
	$n \in \N$ satisfying $P^n \not= \varnothing$ we have
	\begin{equation}
		1 - 2\eta(n)
		\leq
		\E_{(x_B, x_C) \gets D_n \times D_n}
		\bra{P(x_B), P(x_C)}
		\left(C_n \tensor C_n\right)
		\left(x_B \tensor A_n(\rho_n) \tensor x_C\right)
		\ket{P(x_B), P(x_C)}
		\leq \frac{1}{2} + \eta'(n),
	\end{equation}
	implying that $\frac{1}{2} \leq \eta'(n) + 2\eta(n)$.
	Since $\eta'$ and $\eta$ are negligible, this inequality can hold
	for at most finitely many values of $n$.
	Thus, $P^n = \varnothing$ for all sufficiently large values of $n$,
	implying that $\abs{P}$ is finite.
\end{proof}

\subsection{A Promise Problem with Uncloneable Advice, Unconditionally}
\label{sc:promise}

In this section, we formally describe the promise problem with
uncloneable advice which was discussed in \cref{sc:problem-review}. The
goal is to prove the following theorem.

\begin{theorem}
\label{th:promise}
	Let $(s_\lambda)_{\lambda \in \N}$ be an exponential-time
	ingenerable $2\lambda$-sequence and parse each $s_\lambda$ as a pair of strings
	of length $\lambda$, \emph{i.e.:}
	$s_\lambda = (x_\lambda, \theta_\lambda)$ for $x_\lambda,
	\theta_\lambda \in \{0,1\}^\lambda$ for all $\lambda \in \N$.
	For each $b \in \{0,1\}$, let
	\begin{equation}
		P_b
		=
		\bigcup_{\lambda \in \N}
		\left\{
			(\theta_\lambda, y)
			\;:\;
			y \in \{0,1\}^\lambda
			\land
			x_\lambda \cdot y = b
		\right\}.
	\end{equation}
	Then, the promise problem $P = (P_0, P_1)$ is in
	$\mathbf{neglQP}\text{/upoly}$.
\end{theorem}

We prove this theorem by first collecting three lemmas, each
encapsulating one step of the argument we made above.

We begin by recalling a lemma concerning the monogamy-of-entanglement
game of \cite{TFKW13}.\footnote{%
	This particular form of the lemma was first stated, essentially, in
	\cite{BL20}. However, it is an immediate corollary of a more general
	result found in \cite{TFKW13} obtained by exploiting the operational
	relation between measuring half an EPR pair \cite{EPR35} and sending
	a random Wiesner state.}

\begin{lemma}[{\cite{TFKW13}}]
\label{th:tfkw}
	Let $n,a_B,a_C \in \N$ be three non-negative integers and let
	$(A, B, C)$ be a triplet of $(n,a_B+a_C)$-, $(a_B+n, n)$-, and
	$(a_C+n, n)$-circuits, respectively. Then, we have that
	\begin{equation}
		\E_{\substack{x \gets \{0,1\}^n \\ \theta \gets \{0,1\}^n}}
		\bra{x,x}
			\left(B \tensor C\right)
			\left(
				\theta
				\tensor
				A\left(\ketbra*{x^\theta} \right)
				\tensor
				\theta
			\right)
		\ket{x,x}
		\leq
		\left(\frac{1}{2} + \frac{1}{2\sqrt{2}}\right)^n
		<
		0.86^n.
	\end{equation}
\end{lemma}

As for \cref{th:money}, this lemma holds for any triplet of channels,
but we frame it in terms of exponential-time circuits to maintain
consistency with the rest of this work and, in particular, this section.
Now, we ``derandomize'' the above, analogously to how we previously
obtained \cref{th:money-derandomized}.

\begin{lemma}
\label{th:tfkw-derandomized}
	Let $(s_\lambda)_{\lambda \in \N}$ be an exponential-time
	ingenerable $2\lambda$-sequence and parse each $s_\lambda$ as a pair
	of strings of length $\lambda$, \emph{i.e.:} $s_\lambda = (x_\lambda, \theta_\lambda)$
	for $x_\lambda, \theta_\lambda \in \{0,1\}^\lambda$ for all $\lambda
	\in \N$.

	Then, for any triplet $
		(
			(A_\lambda)_{\lambda \in \N},
			(B_\lambda)_{\lambda \in \N},
			(C_\lambda)_{\lambda \in \N}
		)
	$ of exponential-time $(\lambda, b + c)$-,
	$(\lambda + b, \lambda)$-, and
	$(c + \lambda, \lambda)$-circuit families for maps $b,c : \N\to\N$,
	respectively, we have that
	\begin{equation}
		\lambda
		\mapsto
		\bra{x_\lambda, x_\lambda}
			\left(B_\lambda \tensor C_\lambda\right)
			\left(
				\theta_\lambda
				\tensor
				A_\lambda\left(\ketbra*{x_\lambda^{\theta_\lambda}}\right)
				\tensor
				\theta_\lambda
			\right)
		\ket{x_\lambda, x_\lambda}
	\end{equation}
	is negligible.
\end{lemma}
\begin{figure}
	\begin{center}
		\input{fig-circuit-tfkw-derandomized.tex}
	\end{center}
	\caption{\label{fg:tfkw-derandomized}%
		A schematic representation of the $\tilde{C}_\lambda$ circuit
		constructed in the proof of \cref{th:tfkw-derandomized}. The
		wires are labelled, when possible, with the states they are
		expected to carry in the context of the proof. Every wire
		represents $\lambda$ qubits, except the initial wire, final
		wire, and those between the $A_\lambda$, $B_\lambda$, and
		$C_\lambda$ circuits. These represent, respectively, $2\lambda$,
		$1$, $b_\lambda$, and $c_\lambda$ qubits.
	}
\end{figure}
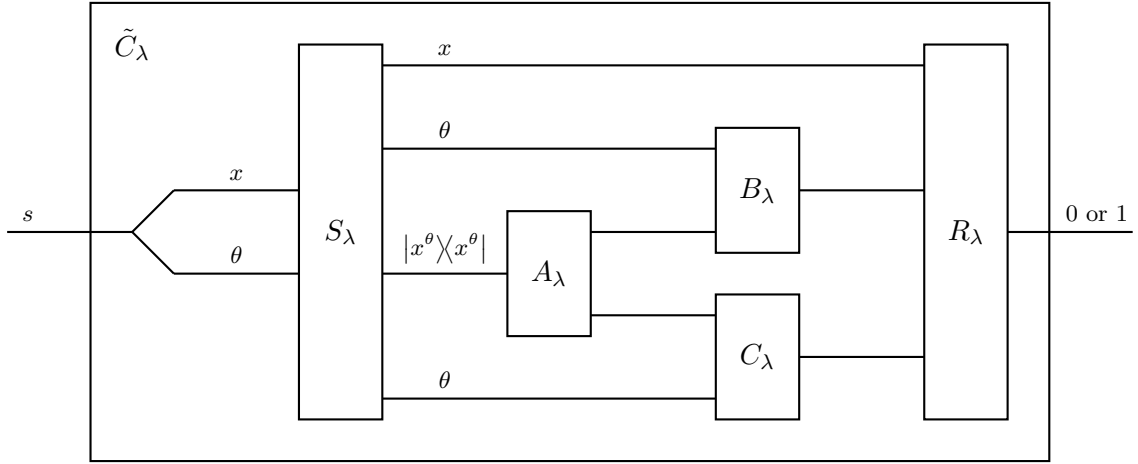
\begin{proof}
	It suffices to apply \cref{th:random-to-ingenerable} to
	\cref{th:tfkw}. We follow the same general ideas as in the proof of
	\cref{th:money-derandomized}.

	Let $(S_\lambda)_{\lambda \in \N}$ be a polynomial-time family of
	$(2\lambda,4\lambda)$-circuits such that each circuit, on input of
	$x \tensor \theta$, outputs the state
	$x \tensor \theta \tensor \ketbra{x^\theta} \tensor \theta$ for
	any $x,\theta \in \{0,1\}^\lambda$. Let $(R_\lambda)_{\lambda\in\N}$
	be a polynomial-time family of $(3\lambda,1)$-circuits where each
	circuit measures all of its input qubits in the computational bases,
	parses it as three strings of length $\lambda$, and outputs $1$ if
	all three are equal and $0$ otherwise.

	We now consider the exponential-time family of $(2\lambda,1)$-circuits
	$(\tilde{C}_\lambda)_{\lambda \in \N}$ such that, for all
	$\lambda \in \N$, $\tilde{C}_\lambda$ is the composition of the
	$A_\lambda$, $B_\lambda$, $C_\lambda$, $S_\lambda$, and $R_\lambda$
	circuits such that the resulting channel is given by
	\begin{equation}
		\tilde{C}_\lambda
		=
		R_\lambda
		\circ
		\left(
			\text{Id}_{\lambda}
			\tensor
			B_\lambda
			\tensor
			C_\lambda
		\right)
		\circ
		\left(
			\text{Id}_{\lambda}
			\tensor
			\text{Id}_{\lambda}
			\tensor
			A_\lambda
			\tensor
			\text{Id}_{\lambda}
		\right)
		\circ
		S_\lambda.
	\end{equation}
	This composition is illustrated in \cref{fg:tfkw-derandomized}.
	Trivial calculations shows that for all
	$x, \theta \in \{0,1\}^\lambda$ we have that
	\begin{equation}
	\label{eq:tfkw-derandomized}
		\bra{1} \tilde{C}_\lambda(x \tensor \theta)\ket{1}
		=
		\bra{x, x}
			\left(B_\lambda \tensor C_\lambda\right)
			\left(
				\theta
				\tensor
				A_\lambda(\ketbra*{x^\theta})
				\tensor
				\theta
			\right)
		\ket{x, x}
	\end{equation}
	In particular, by \cref{th:tfkw},
	\begin{equation}
		\E_{
			\substack{
				x \gets \{0,1\}^\lambda\\
				\theta \gets \{0,1\}^\lambda
			}
		}
		\bra{1}
			\tilde{C}_\lambda(x \tensor \theta)
		\ket{1}
		\leq
		\left(\frac{1}{2} + \frac{1}{2\sqrt{2}}\right)^\lambda
	\end{equation}
	is a negligible function. Thus, by
	\cref{th:random-to-ingenerable} and the fact that
	$(s_\lambda)_{\lambda \in \N}$ is an exponential-time ingenerable
	$2\lambda$-sequence, the function
	$
		\lambda
		\mapsto
		\bra{1}
			\tilde{C}_\lambda(x_\lambda \tensor \theta_\lambda)
		\ket{1}
	$ is also negligible which, by \cref{eq:tfkw-derandomized}, is the
	desired result.
\end{proof}

Finally, we can apply \cref{th:kundu-tan} to show that any
exponential-time adversaries will have a negligible advantage in determining the inner
product of $x_\lambda$ with uniformly random strings $y$.

\begin{lemma}
\label{th:tfkw-kundu-tan-derandomized}
	Let $(s_\lambda)_{\lambda \in \N}$ be an exponential-time
	ingenerable $2\lambda$-sequence and parse each $s_\lambda$ as a
	pair of strings of length
	$\lambda$, \emph{i.e.:} $s_\lambda = (x_\lambda,\theta_\lambda)$
	for two strings $x_\lambda, \theta_\lambda \in \{0,1\}^\lambda$ for all $\lambda
	\in \N$. Then, for any triplet~$
		(
			(A_\lambda)_{\lambda \in \N},
			(B_\lambda)_{\lambda \in \N},
			(C_\lambda)_{\lambda \in \N}
		)
	$ of exponential-time $(\lambda, b + c)$-, $(2\lambda + b, 1)$-,
	and $(c + 2\lambda, 1)$-circuit families, respectively and for maps
	$b, c : \N \to \N$, there exists a negligible function $\eta$ such
	that
	\begin{equation}
		\E_{\substack{y_B \gets \{0,1\}^\lambda\\y_C \gets
		\{0,1\}^\lambda}}
		\bra{x_\lambda \cdot y_B, x_\lambda \cdot y_C}
			\left(B_\lambda \tensor C_\lambda\right)
			\left(
				y_B
				\tensor
				\theta_\lambda
				\tensor
				\sigma_\lambda
				\tensor
				\theta_\lambda
				\tensor
				y_C
			\right)
		\ket{x_\lambda \cdot y_B, x_\lambda \cdot y_C}
		=
		\frac{1}{2} + \eta(\lambda)
	\end{equation}
	where $
		\sigma_\lambda
		=
		A_\lambda\left(\ketbra*{x_\lambda^{\theta_\lambda}}\right)
	$.
\end{lemma}
\begin{proof}
	In the notation of \cref{th:kundu-tan}, we let $\xi_\lambda$ be the
	random variable distributed on
	$\{0,1\}^\lambda \times \{0,1\}^\lambda$ such that
	$\Pr[\xi_\lambda = (x_\lambda, x_\lambda)] = 1$ and we let $
		\rho_{\lambda, x_\lambda, x_\lambda}
		=
		\theta_\lambda
		\tensor
		\sigma_{\lambda}
		\tensor
		\theta_\lambda
	$.

	\Cref{th:tfkw-derandomized} implies that for every pair
	$((B'_\lambda)_{\lambda \in \N}, (C'_\lambda)_{\lambda \in \N})$ of
	exponential-time $(\lambda + b, \lambda)$- and $(c + \lambda,
	\lambda)$-circuits, respectively, the function
	\begin{equation}
		\lambda
		\mapsto
		\bra{x_\lambda, x_\lambda}
			(B'_\lambda \tensor C'_\lambda)
			(\rho_{\lambda, x_\lambda, x_\lambda})
		\ket{x_\lambda, x_\lambda}
	\end{equation}
	is negligible. Since the support of $\xi_\lambda$ is
	$\{(x_\lambda,x_\lambda)\}$, the premise of \cref{th:kundu-tan} is
	satisfied. Thus,
	\begin{equation}
		\lambda
		\mapsto
		\E_{\substack{y_B \gets \{0,1\}^\lambda\\y_C \gets
		\{0,1\}^\lambda}}
		\bra{x_\lambda \cdot y_B, x_\lambda \cdot y_C}
			\left(B_\lambda \tensor C_\lambda\right)
			\left(
				y_B
				\tensor
				\rho_\lambda
				\tensor
				y_C
			\right)
		\ket{x_\lambda \cdot y_B,x_\lambda \cdot y_C}
		-
		\frac{1}{2}
	\end{equation}
	is also negligible. Taking $\eta$ to be this map and expanding the
	definition of $\rho_\lambda$ yields the result.
\end{proof}

We give a final technical lemma concerning exponential-time ingenerable
sequences before proving the theorem.

\begin{lemma}
\label{th:ingenerable-prefix}
	Let $(s_\lambda)_{\lambda \in \N}$ be an exponential-time
	ingenerable $2\lambda$-sequence and parse each $s_\lambda$ as
	$(x_\lambda, \theta_\lambda)$ for two
	strings~$x_\lambda, \theta_\lambda \in \{0,1\}^\lambda$. Then,
	there exists a $\tilde{\lambda}$ such that $\lambda \geq
	\tilde{\lambda}$ implies that
	$x_\lambda \not= 0^\lambda$.
\end{lemma}
\begin{proof}
	Assume this was not the case and $x_\lambda = 0^\lambda$ infinitely often.
	Consider now the polynomial-time circuit family $(C_\lambda)_{\lambda \in \N}$ such
	that $C_\lambda$ samples a uniformly random $y \in \{0,1\}^\lambda$
	and outputs $(0^\lambda, y)$. Then, $\bra{s_\lambda}
	C_\lambda(\varepsilon) \ket{s_\lambda} = 2^{-\lambda}$
	infinitely often. But, since $(s_\lambda)_{\lambda \in \N}$ is
	exponential-time ingenerable, there exits a polynomial $p$ such that $
		\bra{s_\lambda} C_\lambda(\varepsilon) \ket{s_\lambda}
		<
		p(n) 2^{-2\lambda}
	$ for all sufficiently large values of $\lambda$.
	This implies that $2^{\lambda} < p(\lambda)$ infinitely often, a
	contradiction. Hence, $x_\lambda = 0^\lambda$ only finitely often.
\end{proof}

We can now give the proof of \cref{th:promise}. Note that we will use
$\nu$ as our main indexing variable and set $\nu = 2\lambda$ whenever
$\nu$ is even. This will make the notation pertaining to $\lambda$ a bit
more consistent when invoking the previous lemma.

\begin{proof}[Proof of \cref{th:promise}.]
	Define the sequence of states $(\rho_{\nu})_{\nu \in \N}$ as
	\begin{equation}
		\rho_{\nu}
		=
		\begin{cases}
			\ketbra*{x_\lambda^{\theta_\lambda}}
			&
			\qq{if} \nu = 2 \lambda \text{ for } \lambda \in \N\\
			\varepsilon
			&
			\qq{else.}
		\end{cases}
	\end{equation}
	for all $\nu \in \N$. We will use these states as the advice for
	$P$. Let $a : \N \to \N$ be such that $a(\lambda)$ is the number of
	qubits in the advice state $\rho_\nu$. Note that $a(\nu) = \lambda$
	if $\nu = 2\lambda$ and $q(\nu) = 0$ otherwise.

	\emph{Correctness.}
	To show the correctness criterion, we give an explicit
	polynomial-time circuit family which solves $P$ with certainty when
	given these advice states. Let $(C_\nu)_{\nu \in \N}$ be the
	polynomial-time~$(a + \nu, 1)$-circuit family which is described
	below:
	\begin{itemize}
		\item
			If $\nu$ is odd, $C_\nu$ discards the input and
			outputs a single qubit in the $0$ state.
		\item
			If $\nu = 2\lambda$ is even, $C_\nu$ does the following:
			\begin{enumerate}
				\item
					Parse the $q(\nu) + \nu = 3 \lambda$ input
					qubits as three registers
					$\tsf{A} \tensor \tsf{B} \tensor \tsf{S}$ of
					$\lambda$ qubits each.
				\item
					For each $i \in [\lambda]$, apply a Hadamard gate
					$H$ on the $i$-th qubit of $\tsf{A}$ conditioned on
					the $i$-th qubit of $\tsf{B}$.
				\item
					Using an extra qubit initialized in the $0$ state,
					which we will call the $\tsf{R}$ register, compute
					the inner product of the $\tsf{A}$ and $\tsf{S}$
					registers and store the result in the $\tsf{R}$
					register. We do this by applying $\lambda$ Toffoli
					gates on $\tsf{R}$, each conditioned on the $i$-th
					qubits of $\tsf{A}$ and $\tsf{S}$, respectively, for
					each $i \in [\lambda]$.
				\item
					Discard the $\tsf{A}$, $\tsf{B}$, and $\tsf{S}$
					registers and output the $\tsf{R}$ register.
			\end{enumerate}
	\end{itemize}
	Examples of these circuits, for $\nu = 3$ and $\nu = 4$, are
	illustrated in \cref{fg:promise-circuits}. For all $\lambda \in \N$
	and all string $y \in \{0,1\}^\lambda$, a trivial calculation yields
	\begin{equation}
		C_{2\lambda}\left(
			\rho_{2\lambda} \tensor \theta_\lambda \tensor y
		\right)
		=
		C_{2\lambda}\left(
			\ketbra*{x_\lambda^{\theta_\lambda}}
			\tensor
			\theta_\lambda
			\tensor y
		\right)
		=
		x_\lambda \cdot y
		=
		P(\theta_\lambda, y)
	\end{equation}
	Thus, for any $z \in P$, which must be of the form
	$z = (\theta_\lambda, y)$ for some $y \in \{0,1\}^\lambda$, we
	have that $C_{\abs{z}}$ correctly computes $P(z)$ with certainty when
	also given the advice state $\rho_{\abs{z}}$. Hence, the correctness
	criterion is satisfied.

	\begin{figure}
		\begin{center}
			\begin{subfigure}{0.45\textwidth}
				\begin{center}
					\input{fig-circuit-promise-odd.tex}
				\end{center}
				\subcaption{The circuit $C_3$.}
			\end{subfigure}
			\begin{subfigure}{0.45\textwidth}
				\begin{center}
					\input{fig-circuit-promise-even.tex}
				\end{center}
				\subcaption{%
					The circuit $C_4$. The advice state is inputted in
					the $\tsf{A}$ register and the problem instance in
					the other registers.
				}
			\end{subfigure}
		\end{center}
		\caption{%
			\label{fg:promise-circuits}%
			The circuits $C_\nu$ for $\nu = 3,4$ used in
			the proof of \cref{th:promise} to demonstrate correctness.
		}
	\end{figure}

	\emph{Uncloneability.} Our first step in showing that these
	advice states satisfy the uncloneability criterion is to gain a
	better understanding of the sequence of distributions
	$(D_\nu \times D_\nu)_{\nu \in \N}$, as given in
	\cref{df:neglqp/upoly}, for this promise problem.

	Recall that by \cref{th:ingenerable-prefix}, there are only finitely
	many instances where $x_\lambda = 0^\lambda$. This implies that for each
	$b \in \{0,1\}$ and all sufficiently large values of $\lambda$, the
	set~$P_b^{2\lambda}$ is exactly half of
	$P^{2\lambda}$. Indeed, this follows from the fact that
	$\abs{\{y\in\{0,1\}^\lambda \;:\; x\cdot y=b\}}$ is $2^{\lambda-1}$
	precisely when $x \not= 0^\lambda$. This is to say that, for all
	sufficiently large $\lambda$, there are always as many yes
	instances as no instances of length $2\lambda$. In particular, for
	all sufficiently large values of $\lambda$, the random variable
	$D_{2\lambda}$ is uniformly random over $P^{2\lambda} =
	\{\theta_\lambda\} \times \{0,1\}^\lambda$.

	Now, consider a triplet $
		(
			(A_\nu)_{\nu \in \N},
			(B_\nu)_{\nu \in \N},
			(C_\nu)_{\nu \in \N}
		)
	$ of polynomial-time $(q, b + c)$-, $(2\lambda + b, 1)$-, and
	$(c + 2\lambda, 1)$-circuit families, respectively and for maps
	$b,c : \N \to \N$. Let $v : \N \to \R$ be the function given by
	\begin{equation}
		v(\lambda)
		=
		\E_{(z_B, z_C) \gets D_{2\lambda} \times D_{2\lambda}}
		\bra{P(z_B), P(z_C)}
		\left(B_{2\lambda} \tensor C_{2\lambda}\right)
		\left(
			z_B
			\tensor
			A_{2\lambda}\left(\rho_{2\lambda}\right)
			\tensor
			z_C
		\right)
		\ket{P(z_B), P(z_C)}.
	\end{equation}
	In other words, $v(\lambda)$ is the probability that the triplet of
	circuit families split and successfully using the advice state to
	solve problem instances in $P^{2\lambda}$ sampled according to
	$D_{2\lambda} \times D_{2\lambda}$. To show that the
	uncloneability criterion is satisfied, it suffices so show the
	existence of a negligible function $\eta$ such
	that~$v(\lambda) \leq \frac{1}{2} + \eta(\lambda)$ as, in the
	notation of \cref{df:neglqp/upoly}, it will then suffice to take
	$\eta'$ to be function such that $\eta'(\nu) = 0$ if $\nu$ is odd
	and $\eta'(\lambda)$ if $\nu = 2\lambda$ is even. Clearly, this
	$\eta'$ will be negligible if $v$ is negligible.

	For $L \in \{A, B, C\}$, let $L'_\lambda = L_{2\lambda}$ and note
	that $(L_\lambda)_{\lambda \in \N}$ is a polynomial-time family of
	circuits.
	Now, since the random variable $D_{2\lambda}$ is uniformly
	distributed on $\{(\theta_\lambda, y) : y \in \{0,1\}^\lambda\}$
	for all sufficiently large values of $\lambda$, we have that
	\begin{equation}
		v(\lambda)
		=
		\E_{\substack{y_B\gets\{0,1\}^\lambda\\y_C\gets\{0,1\}^\lambda}}
		\bra{x_\lambda \cdot y_B, x_\lambda \cdot y_C}
			(B'_\lambda \tensor C'_\lambda)
			\left(
				\theta_\lambda
				\tensor
				y_B
				\tensor
				A'_\lambda(\rho_{2\lambda})
				\tensor
				\theta_\lambda
				\tensor
				y_C
			\right)
		\ket{x_\lambda \cdot y_B, x_\lambda \cdot y_C}
	\end{equation}
	for all sufficiently large values of $\lambda$. Up to performing the
	swapping map
	$\theta_\lambda \tensor y_B \mapsto y_B \tensor \theta_\lambda$ on
	the first $2\lambda$ qubits given to $B'_\lambda$, an operation
	which can easily be incorporated into this circuit,
	\cref{th:tfkw-kundu-tan-derandomized} implies the existence of a
	negligible function $\tilde{\eta}$ such that the right-hand
	side of the above equation is precisely $\frac{1}{2} +
	\tilde{\eta}(\lambda)$. Since $v(\lambda)$ is equal to
	$\frac{1}{2} + \tilde{\eta}(\lambda)$ for all sufficiently large
	values of $\lambda$, there exists a function
	$\eta : \N \to \R$ which differs from $\tilde{\eta}$ on only
	finitely many inputs such that~$v = \frac{1}{2} + \eta$. Since
	$\eta$ differs only finitely often from a negligible function, it is
	itself negligible.
\end{proof}

\begin{remark}
	If the language in \cref{th:promise} is instantiated with an
	exponential-time ingenerable sequence $(s_n)_{n \in \N}$ which
	satisfies the stronger notion of being \emph{computably}
	ingenerable, then the advice states given in the proof also satisfy
	uncloneability against uniform adversaries.

	On the other hand, if the language is instantiated with an
	exponential time ingenerable sequences which can be computed in
	triple-exponential time, then we note that the advice states can be
	generated by quantum circuits described by a Turing machine running
	in triple-exponential time. By \cref{th:ingenerable}, such sequences
	exist.
\end{remark}

Note that the advice states which we give in the proof of
\cref{th:promise} are quantum but, in some sense, the advice is actually
the \emph{classical} strings $(x_n)_{n \in \N}$. We encode this
classical advice into a quantum state only to achieve uncloneability,
not additional computational power.
With this in mind, the following proposition formalizes the idea that
uncloneable advice need not be more powerful than classical advice.

\begin{proposition}
\label{th:neglqpupoly-ppoly}
	The intersection of $\textbf{neglQP}/\text{upoly}$ and
	$\textbf{P}/\text{poly}$, where the latter is the class of promise
	problems which can be solved in polynomial-time by a classical
	deterministic Turing machine with access to polynomially many
	classical bits of advice, is non-empty.
\end{proposition}
\begin{proof}[Proof sketch.]
	The promise problems constructed in \cref{th:promise}, which states
	that these problems are in $\mathbf{neglQP}/\text{upoly}$, can also
	be seen to be in $\textbf{P}/\text{poly}$. Indeed, if the problem is
	instantiated with the exponential-time ingenerable $2n$-sequence
	$(s_n = (\theta_n, x_n))_{n \in \N}$, then the classical advice
	strings~$(s'_n)_{n \in \N}$ defined by
	\begin{equation}
		s'_n
		=
		\begin{cases}
			x_n & \qq{if} n\text{ is even} \\
			\varepsilon & \qq{else}
		\end{cases}
	\end{equation}
	are sufficient to show that the problem constructed in
	\cref{th:promise} is also in $\textbf{P}/\text{poly}$.
\end{proof}

Similarly, it is clear that if a classical party has enough
computational power to generate the sequence $(s_n)_{n \in \N}$ or, more
precisely, the sequence $(x_n)_{n \in \N}$, then it can solve the
promise problem constructed in \cref{th:promise} from $(s_n)_{n\in\N}$.

\begin{proposition}
\label{th:neglqpupoly-eeexp}
	The intersection of $\textbf{neglQP}/\text{upoly}$ and
	$\textbf{EEEXP}$, where the latter is the class of promise problems
	which can be solved in triple-exponential time by a classical
	deterministic Turing machine, is non-empty.
\end{proposition}
\begin{proof}[Proof sketch.]
	Some instantiations of the promise problems constructed in
	\cref{th:promise}, which states that these problems are in
	$\mathbf{neglQP}/\text{upoly}$, can be seen to be in
	$\mathbf{EEEXP}$. Indeed, if the exponential-time ingenerable
	sequence $(s_n)_{n \in \N}$ used in the construction is computable
	in triple-exponential time by a classical deterministic Turing
	machine, then the resulting promise problem is in
	$\mathbf{EEEXP}$. By \cref{th:ingenerable}, such sequences exist.
\end{proof}

\Cref{th:neglqpupoly-eeexp,th:neglqpupoly-ppoly} immediately point to
an open question which we believe is of interest: Are there non-trivial
problems, i.e.: with infinitely many \emph{yes} and \emph{no} instances,
in the intersection of $\mathbf{neglQP}/\text{upoly}$ and
$\mathbf{EEXP}$, the class of problems which can be solved by classical
deterministic Turing machines in double exponential time? What about in
the intersection of $\mathbf{neglQP}/\text{upoly}$ and $\mathbf{EXP}$ or
any other classical deterministic super-polynomial time complexity
classes? These questions appear to reflect those considered in
\cref{sc:cloning-complexity} about gaps in the complexity of generating
and cloning sequences of states.

\subsection{A Language with Uncloneable Advice, Under Certain
Assumptions}
\label{sc:language}

In this section, we describe a class of languages which admit
uncloneable advice under the assumption that it is possible to
copy-protect certain families of maps. We proceed
in two steps. First, we give a general template in \cref{df:language}
for the class of languages we consider. Each language from this template
is parameterized by a sequence of maps and is illustrated in
\cref{fg:language}. Second, in \cref{th:language}, we
establish sufficient conditions on the parametrizing maps which ensure
that the resulting language admits uncloneable advice. We comment, in
the next section, on certain specific instantiations of this language
admitting uncloneable advice based on existing
constructions for the copy-protection of pseudorandom functions.

\begin{definition}
\label{df:language}
	Let $g = (g_n : \{0,1\}^{d(n)} \to \{0,1\}^{c(n)})_{n \in \N}$ be a
	sequence of maps whose domains and codomains are parameterized by
	$d, c : \N \to \N$. We define the language $L^g = (L_0^g, L_1^g)$ as
	follows:
	\begin{itemize}
		\item
			If $w \in \{0,1\}^*$ is such that
			$\abs{w} \leq (d + c)(\abs{w})$, then
			\begin{equation}
				w \in L_b^g \iff w \cdot w = b.
			\end{equation}
		\item
			If $w \in \{0,1\}^*$ is such that
			$\abs{w} > (d + c)(\abs{w})$, we parse $w$ as $(x,y,z)$
			where $x$ is composed of the first $d(\abs{w})$ bits of
			$w$, $y$ of the next $c(\abs{w})$ bits, and $z$
			of the remaining bits. Then,
			\begin{equation}
				w = (x,y,z) \in L_b^g
				\iff
				(g_\abs{w}(x) \cdot y) \xor (z \cdot z) = b.
			\end{equation}
	\end{itemize}
\end{definition}

\begin{figure}[H]
	\input{fig-language.tex}
	\caption{\label{fg:language}%
		An illustration of an algorithm which determines if a given
		string $w \in \{0,1\}^*$ is in the set $L^g_0$ or the
		set $L^g_1$ where these are as given in
		\cref{df:language}. The single bit produced, at the right of the
		diagram, determines to which set $w$ belongs. In the case where
		the criterion $\abs{w} > (d + c)(\abs{w})$ is met,
		we parse $w$ as a triplet of strings $(x,y,z)$, each of the
		unique length ensuring that $(g_{\abs{w}}(x) \cdot y)
		\xor (z \cdot z)$ is well defined.}
\end{figure}
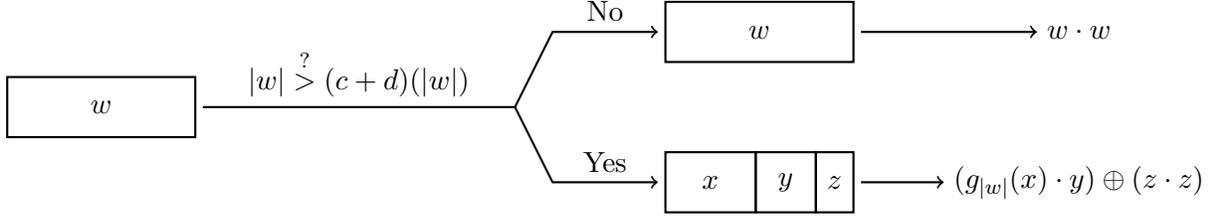

We note a nice property of $L^g$ which is independent of the choice of
$g$, namely that exactly half of all bit strings of a given non-zero
length are in the language.

\begin{lemma}
\label{th:balanced}
	Let $L^{g} = (L_0^{g}, L_1^{g})$ be as given
	in \cref{df:language}. Then, for all $n \geq 1$ we have that
	\begin{equation}
		\abs{L_0^{g} \cap \{0,1\}^n}
		=
		\abs{L_1^{g} \cap \{0,1\}^n}
		=
		2^{n - 1}.
	\end{equation}
	In other words, for any given non-zero length $n$, there are
	as many \emph{yes} instances as there are \emph{no} instances of
	length $n$ in $L^{g}$.
\end{lemma}
\begin{proof}
	This follows immediately from the fact that for all $n \geq 1$ and
	$b \in \{0,1\}$ we have that
	\begin{equation}
		\abs{\left\{s \in \{0,1\}^n \;:\; s \cdot s = b\right\}}
		=
		2^{n-1}.
	\end{equation}
\end{proof}

We now establish conditions under which the language $L^g$ defined above
admits uncloneable advice.

\begin{theorem}
\label{th:language}
	Let $
		f
		=
		(
			f_n :
			\{0,1\}^{\kappa(n)} \times \{0,1\}^{d(n)}
			\to \{0,1\}^{c(n)}
		)_{n \in \N}$ be
	a sequence of maps whose domains and codomains are parameterized by
	maps $\kappa,d,c : \N \to \N$ such that the following conditions are
	satisfied:
	\begin{itemize}
		\item
			There exists an $n' \in \N$ such that $n \geq n' \implies
			d(n) + c(n) < n$ and $c \in \omega(\log)$.
		\item
			There exists a copy-protection scheme $(G,E)$ for $f$ which
			is correct and secure.
	\end{itemize}
	Let $k = (k_n)_{n \in \N}$ be an exponential-time
	ingenerable~$\kappa$-sequence and let $g = (g_n)_{n \in \N}$ be the
	sequence of maps where $g_n = f_n(k_n,\cdot)$ for all~$n \in \N$.
	Then, the language $L^g$ is in
	$\textbf{neglQP}\text{/upoly}$.
\end{theorem}

At a high level, the proof proceeds as follows: The
advice states will be the program states generated by the
copy-protection scheme for the maps in the sequence $g$. Correctness
follows trivially from the correctness of the copy-protection scheme.
Uncloneability is a bit more involved. Consider only the cases
where $n \geq n'$ and so any string $w$ will be parsed as $(x,y,z)$. The
first step is to argue that the $z$ component does not really contribute
to the difficulty of determining if $w$ is in the language. The only
difficult part is for both guessers to be able to correctly and
simultaneously compute $g_n(x) \cdot y$. Since the $y$'s are
independently sampled for both guessers, the result of Kundu and Tan
(\cref{th:kundu-tan}) tells us that if both parties can only compute
$g_n(x)$ with a negligible probability of success, then they can
simultaneously correctly compute this inner product with at most a
negligible advantage above $\frac{1}{2}$. How can we show that this
probability of simultaneously guessing~$g_n(x)$ on independent values of
$x$ is negligible? It suffices to derandomize the copy-protection game
with the exponential-time ingenerable sequence $(k_n)_{n \in \N}$.
More precisely, since the advice states are produced from a secure
copy-protection scheme for maps where the trivial guessing probability
$2^{-c}$ is trivial (as $c \in \omega(\log)$),
we know that any efficient adversaries will fail to simultaneously and
correctly compute $f_n(k',x)$ with more than a negligible probability
if the key $k'$ is uniformly random (but identical for both adversaries) and
where $x$ is uniformly random and independent for both adversaries. By
\cref{th:random-to-ingenerable}, the probability will remain negligible
even when we do not sample $k'$ uniformly at random, but rather fix it
to the $n$-th element of an exponential-time ingenerable $\kappa$-sequence.

\begin{proof}
	Assume the copy-protection scheme $(G,E)$ has program length $q$.
	We will consider the sequence of states
	\begin{equation}
		\rho
		=
		\left(
			\rho_n
			=
			G_n(k_{n})
		\right)_{n \in \N},
	\end{equation}
	which are the program states produced by the copy-protection scheme
	$(G, E)$ for the maps $(g_n)_{n \in \N}$, as our advice for this
	language.

	\emph{Correctness.} We first show correctness. This follows directly
	from the assumed correctness of the copy-protection scheme $(G, E)$.
	In short, we implement the algorithm shown in \cref{fg:language}
	while using $\rho_n$ and $E_n$ to evaluate $g_n$. More explicitly,
	let $C = (C_n)_{n \in \N}$ be the efficient $(q+n,1)$-circuit family
	where each $C_n$ implements the following algorithm on input of
	$\rho \tensor w$:
	\begin{enumerate}
		\item
			Verify if $n \leq d(n) + c(n)$.
			\begin{itemize}
				\item
					If this inequality holds, output the bit $w \cdot w$
					and terminate.
				\item
					If this inequality does not hold, execute steps $2$
					to $4$ below.
			\end{itemize}
		\item
			Parse $w$ as a triplet $(x,y,z)$ where $x$ is the
			first $d(n)$ bits, $y$ the
			subsequence $c(n)$ bits, and $z$
			the remaining $n - d(n) + c(n)$ bits.
		\item
			Run $E_{n}(\rho_n \tensor x)$ and
			measure the output in the computational basis. Let $y'$
			denote this result.
		\item
			Output the bit $(y' \cdot y) \xor (z \cdot z)$ and
			terminate.
	\end{enumerate}
	If $n \leq d(n) + c(n)$, then $C_n(\rho_n \tensor w)$ outputs
	the correct bit, namely $w \cdot w$, with probability $1$.
	Else,~$C_n(\rho_n \tensor w)$ outputs the correct bit with a
	probability at least
	\begin{equation}
		\bra{g_{n}(x)}
			E_{n}(\rho_n \tensor x)
		\ket{g_{n}(x)}
		=
		\bra{f_{n}(k_{n}, x)}
			E_{n}\left(
				G_{n}(k_{n}) \tensor x
			\right)
		\ket{f_{n}(k_{n}, x)}
	\end{equation}
	which is the probability that step 3 correctly evaluates to
	$g_{n}(x)$. Since $(G, E)$ is correct, there exists a
	negligible function $\eta_c$ such that this probability is at least
	$1 - \eta_c(n)$. Hence, the correctness criterion is satisfied.

	\emph{Uncloneability.}
	First, let's set up an adversary and state what is sufficient to
	show to demonstrate the uncloneability criterion.

	Let $(A, B, C)$ be a triplet of polynomial-time
	$(q,q_B+q_C)$-, $(d+q_B,1)$-, and $(q_C+d,1)$-circuit
	families, respectively, for maps $q_B,q_C : \N \to \N$. Let
	$D_n$ be the random variable distributed on the set $\{0,1\}^n$ as
	defined in \cref{df:neglqp/upoly} for the language $L^{g}$.
	Note that, by \cref{th:balanced}, $D_n$ is uniformly
	distributed. Let
	\begin{equation}
		v(n)
		=
		\E_{w_B, w_C \gets D_n}
		\bra{L^{g}(w_B), L^{g}(w_C)}
			(B_n \tensor C_n)
			\left(
				w_B \tensor A_n(\rho_n) \tensor w_C
			\right)
		\ket{L^{g}(w_B), L^{g}(w_C)}
	\end{equation}
	be the probability that the adversaries $(A, B, C)$ succeed in
	sharing the advice state $\rho_n$ and use it to correctly
	determine the membership of two independent problem instances of
	length $n$ sampled according to $D_n$. It suffices to show the
	existence of a negligible function $\eta$ such that
	$v \leq \frac{1}{2} + \eta$. In fact, it suffices to show
	that this inequality holds for all $n \geq n'$ where, we recall,
	$n'$ is a threshold after which $n < d(n) + c(n)$ always holds.
	Assume from now on that $n \geq n'$ is always satisfied and parse
	every $w \in \{0,1\}^n$ as a triplet $(x,y,z)$ where $x$ is the
	first $d(n)$ bits, $y$ the subsequent $c(n)$ bits, and $z$ the
	remaining bits.

	Now, as a first step, we argue that the $z$ component of a problem
	instances essentially does not matter. In the process, we will
	express $v(n)$ in a form which will be closer to
	something to which we can apply \cref{th:kundu-tan}.

	Let $B'$ be a circuit family such that for all $n \geq n'$,
	the circuit $B'_n$ implements the following algorithm on input
	of $\rho \tensor x \tensor y$:
	\begin{itemize}
		\item
			Sample a uniformly random
			$z \gets \{0,1\}^{n-d(n)-c(n)}$.
		\item
			Run $B_n(\rho \tensor x \tensor y \tensor z)$ and
			measure the output in the computational basis. Let $b$
			denote the result of this measurement.
		\item
			Output the bit $b \xor (z \cdot z)$ and terminate.
	\end{itemize}
	(For completeness, we can assume that if $n < n'$ then $B'_n$
	simply discards its input and outputs $0$. These cases do not matter
	in our analysis.) Note that $B'$ is an efficient family of circuits.
	Define the circuit family $C' = (C'_n)_{n \in \N}$ analogously with
	respect to the family $C$. Note that $B'_n$
	outputs the bit $b'$ if and only if its execution of $B_n$
	outputs the bit $b' \xor (z \cdot z)$ and similarly for $C'_n$.
	Hence, since $D_n$ is uniformly distributed, we have that $v(n)$ can
	also be expressed as
	\begin{equation}
	\label{eq:language-1}
		\E_{\substack{
			x_B, x_C \gets \{0,1\}^{d(n)}\\
			y_C, y_C \gets \{0,1\}^{c(n)}
		}}
		\bra{
			g_{n}(x_B) \cdot y_B,
			g_{n}(x_C) \cdot y_C
		}
		(B'_n \tensor C'_n)(
			y_B \tensor \sigma_n^{x_B, x_C} \tensor y_C
		)
		\ket{
			g_{n}(x_B) \cdot y_B,
			g_{n}(x_C) \cdot y_C
		}
	\end{equation}
	where
	\begin{equation}
		\sigma_{n}^{x_B, x_C}
		=
			x_B
			\tensor
			A_n(\rho_n)
			\tensor
			x_C.
	\end{equation}
	To show that $v$ is at most negligibly more than $\frac{1}{2}$, it
	now suffices by \cref{th:kundu-tan} to show that for all pairs
	$(B'', C'')$ of efficient $(d + q_B, c)$- and
	$(q_C + d, c)$-circuit families, respectively, we have that
	\begin{equation}
		n \mapsto
		\E_{x_B, x_C \gets \{0,1\}^d(n)}
		\bra{g_n(x_B), g_n(x_C)}
			(B''_n \tensor C''_n)(\sigma_n^{x_B, x_C})
		\ket{g_n(x_B), g_n(x_C)}
	\end{equation}
	is a negligible map.

	Since $(G, E)$ is a secure copy-protection scheme and $2^{-c}$ is
	negligible, we have that
	\begin{equation}
		\E_{k \gets \{0,1\}^{\kappa(n)}}
		\E_{x_B, x_C \gets \{0,1\}^d(n)}
		\bra{f_n(k,x_B), f_n(k,x_C)}
			(B''_n \tensor C''_n)
			(x_B \tensor A_n(G_n(k)) \tensor x_C)
		\ket{f_n(k,x_B), f_n(k,x_C)}
	\end{equation}
	is negligible in $n$ and so, by \cref{th:random-to-ingenerable} and
	the fact that $\kappa \in \omega(\log)$,\footnote{
		If $\kappa \not\in \omega(\log)$, then $2^{-\kappa}$ is
		non-negligible. Thus, a triplet of adversaries
		can break the assumed security of the copy-protection scheme
		$(G, E)$ by sampling a key universally at random and using it,
		with the honest generation and evaluation algorithms, to
		correctly evaluate the function correctly with non-negligible
		probability. Indeed, they will have a non-negligible probability
		of having guessed the correct key.}
	we can fix the keys $k$ to be the elements of the exponential-time
	ingenerable sequence $(k_n)_{n \in \N}$ and obtain that
	\begin{equation}
		n
		\mapsto
		\E_{x_B, x_C \gets \{0,1\}^{d(n)}}
		\bra{f_n(k_n,x_B),f_n(k_n,x_C)}
			(B''_n \tensor C''_n)
			(x_B \tensor A_n(G_n(k_n)) \tensor x_C)
		\ket{f_n(k_n,x_B), f_n(k_n,x_C)}
	\end{equation}
	is also negligible. Recalling that $f_n(k_n, \cdot) = g_n$, this is
	precisely what is needed.

	Note that we neglect in this proof to give an explicit
	construction of a $(\kappa,1)$-circuit which models the complete
	set-up and attack of the copy-protection scheme on input of a given
	key $k$, as is technically required to apply
	\cref{th:random-to-ingenerable}. However, this construction would
	simply follow the same ideas as the ones presented in the proofs of
	\cref{th:money-derandomized} and \cref{th:tfkw-derandomized}: The
	attacking circuits from $B''$ and $C''$ are placed between an
	initial set-up circuit, which creates and distributes the program
	state and challenges, and a final referee circuit, which receives
	the guesses and outputs $1$ if and only if they are correct. This
	composition would yield the necessary circuits.

	Recalling that $\sigma^{x_B, x_C}_n = x_B \tensor A_n(G_n(k'))
	\tensor x_C$ is then sufficient to complete the proof.
\end{proof}

\subsection{%
	Instantiating our Construction of a Language with Uncloneable Advice}
\label{sc:language-2}

Our construction of a language with uncloneable advice in the previous
section did not consider an important question: Is it possible to
copy-protect a sequence of functions satisfying the necessary
conditions? The following theorem gives one answer to this question.
\begin{theorem}
\label{th:language-from-prf}
	If there exists a correct and secure copy-protection scheme for a
	pseudorandom function $f' = (f'_n)_{n \in \N}$ with outputs of
	super-logarithmic length, then there exists a sequence of maps $f =
	(f_n)_{n \in \N}$ satisfying the assumptions of \cref{th:language}.
\end{theorem}

Before proving this theorem, we recall the following result from
\cite{CLLZ21}. It establishes a set of sufficient conditions, which we
do not formalize in this work, for the copy-protection of a certain
construction of pseudorandom functions presented in that work. By the
previous theorem, it follows that these are also sufficient conditions
for the existence of a language with uncloneable advice.
\begin{theorem}[{\cite{CLLZ21}}]
	\label{th:CLLZ21}
	Assuming the existence of post-quantum indistinguishable
	obfuscation, one-way functions, and compute-and-compare obfuscation
	for the class of unpredictable distributions, there exists a secure
	and correct copy-protection scheme for a specific construction of
	pseudorandom functions where the length of the keys, inputs, and
	outputs are non-constant polynomials.
\end{theorem}

For completion, we recall the definition of a pseudorandom function. We
will not be using this definition in our work.
Up to specifying efficient circuit families as our computational model
and restricting the domains and codomains to be sets of bit strings,
this is the definition given by Zhandry \cite{Zha12}.
In short, a keyed function is pseudorandom if, when given oracle access,
it cannot be distinguished from a truly random function.
\begin{definition}
	A sequence of maps $
		f
		=
		(
			f_n :
			\{0,1\}^{\kappa(n)}
			\times
			\{0,1\}^{d(n)}
			\to
			\{0,1\}^{c(n)}
		)_{n \in \N}
	$ is a pseudorandom function if for all efficient $(0,1)$-circuit
	families $C^{O_g}$ having access to an oracle computing maps of
	the form $g : \{0,1\}^{d(n)} \to \{0,1\}^{c(n)}$ it
	holds that
	\begin{equation}
		\lambda
		\mapsto
		\abs{
			\E_{g}
				\bra{1} C_n^{O_g}(\varepsilon) \ket{1}
			-
			\E_{k \in \{0,1\}^{\kappa(n)}}
				\bra{1} C_n^{O_{f(k, \cdot)}} \ket{1}
		}
	\end{equation}
	is a negligible function and where we understand $\E_g$ to represent
	the expectation over the uniformly random choice of a map
	$g : \{0,1\}^{d(n)} \to \{0,1\}^{c(n)}$. Moreover, we
	require that there exists an efficient circuit family $F$ such that
	$F_n(k \tensor x) = f_n(k,x)$ for all $n \in \N$ and that $\kappa$
	be efficiently computable.
\end{definition}

Note that the PRF $f'$ assumed in \cref{th:language-from-prf} do
not necessarily satisfy the first condition of \cref{th:language}, which
is to say that there is no guarantee that $d(n) + c(n)$ will eventually
always be strictly smaller than $n$. In fact, this is never satisfied
for the PRF considered in \cite{CLLZ21}. Thus, we need to show that we
can ``slowdown'' the growth of this value, while keeping the
correctness and, more importantly, the security of the scheme.
The way we will do this is by repeating certain elements of the sequence
$f' = (f'_n)_{n \in \N}$ multiple times in succession to obtain a new
sequence of maps $f = (f_n)_{n \in \N}$. This will be parametrized by a
re-indexing map $\gamma : \N \to \N$ which will, in general, \emph{not}
be injective and where $f_n = f'_{\gamma(n)}$. In some sense, $\gamma$
``slows down'' the rate of growth of the security parameter for the
maps.
The content of the following lemma establishes conditions where this
transformation preserves correctness and security.

\begin{lemma}
\label{th:slowdown}
	Let $(G, E)$ be a copy-protection scheme for a sequence of maps $f$,
	as defined in \cref{eq:cp-maps}, which is correct and secure and
	where $2^{-c}$ is a negligible map.
	Let $\gamma : \N \to \N$ be a non-decreasing computable map in
	$\Omega(\lambda^{r_l})$ and $\mc{O}(\lambda^{r_u})$ for some
	$r_l, r_u \in \R^+$. Let
	$G^\gamma =(G^\gamma_\lambda = G_{\gamma(\lambda)})_{\lambda\in\N}$,
	$E^\gamma =(E^\gamma_\lambda = E_{\gamma(\lambda)})_{\lambda\in\N}$,
	and
	$f^\gamma =(f^\gamma_\lambda = f_{\gamma(\lambda)})_{\lambda\in\N}$.
	Then, $(G^\gamma, E^\gamma)$ is a copy-protection scheme for
	$f^\gamma$ which is correct and secure.

	Moreover, $2^{-c \circ \gamma}$ is also negligible, as it
	is simply $2^{-c} \circ \gamma$ and $\gamma \in
	\Omega(\lambda^{r_l})$.
\end{lemma}
\begin{proof}
	First, note that since $\gamma$ is efficiently computable and
	upper-bounded by a polynomial, as it is in $\mc{O}(\lambda^{r_u})$,
	then $G^\gamma$ and $E^\gamma$ are efficient circuit families.
	The correctness of $(G^\gamma, E^\gamma)$ follows trivially from the
	correctness of $(G, E)$ and the fact that the map
	$\eta \circ \gamma$ is negligible if $\eta$ is negligible as $\gamma
	\in \Omega(\lambda^{r_l})$. We move on to showing the security.

	Let $(A', B', C')$ be an attack against $(G^\gamma, E^\gamma)$. We
	construct an attack $(A, B, C)$ against the original scheme
	$(G, E)$ as follows. For completion, assume we have an existing
	attack $(\tilde{A}, \tilde{B}, \tilde{C})$ against $(G, E)$. The
	specifics of this attack do not matter for this proof; it
	will only be used to fill in some gaps and ensure that all of our
	objects are well defined. We use $n$ to index the security parameter
	of $(A', B', C')$ and $\lambda$ to index the one of $(A, B, C)$.

	For all $\lambda \in \N$, define $A_\lambda$ as follows:
	\begin{itemize}
		\item
			If $\gamma^{-1}(\lambda) = \varnothing$, then
			$A_\lambda = \tilde{A}_\lambda$.
		\item
			If $\gamma^{-1}(\lambda) \not= \varnothing$, then
			$A_\lambda$ is a circuit which samples uniformly at
			random an $n' \in \gamma^{-1}(\lambda)$, executes
			$A'_{n'}$ and then gives both $B_\lambda$ and
			$C_\lambda$ a copy of $n'$ in addition to their
			respective share of the output of $A'_{n'}$.
	\end{itemize}
	For all $\lambda \in \N$, we define $B_\lambda$ as follows:
	\begin{itemize}
		\item
			If $\gamma^{-1}(\lambda) = \varnothing$, then $B_\lambda =
			\tilde{B}_\lambda$.
		\item
			If $\gamma^{-1}(\lambda) \not= \varnothing$, then
			$B_\lambda$ is a circuit which reads the $n'$ value
			received from $A_\lambda$ and then executes the
			$B'_{n'}$ circuit on the remainder of the input
			received from $A_\lambda$.
	\end{itemize}
	The $C_\lambda$ circuits are defined analogously to the $B_\lambda$
	circuits. Before proceeding, we should ensure that the $(A, B, C)$
	defined here is a triplet of \emph{efficient} circuit families.

	First, we show that the set $\gamma^{-1}(\lambda)$ is efficiently
	computable from $\lambda$, in the sense that there exists a
	polynomial-time Turing machine which, on input of $1^{\lambda}$,
	outputs an encoding of $\gamma^{-1}(\lambda)$.

	By our assumptions on the map $\gamma$, there exists a
	$n_\gamma$ and $k_u, k_l \in \R^+$ such that
	$n \geq n_\gamma$ implies that $
		k_l n^{r_l} \leq \gamma(n) \leq k_u n^{r_u}
	$. Thus, for any $n \geq n_\gamma$,
	\begin{equation}
		n \in \gamma^{-1}(\lambda)
		\implies
		k_l \lambda^{r_l}
		\leq
		\gamma(n)
		=
		\lambda
		\leq
		k_u n^{r_u}
		\implies
		(\lambda/k_u)^{\frac{1}{r_u}}
		\leq n \leq
		(\lambda/k_l)^{\frac{1}{r_l}}.
	\end{equation}
	This implies that there are at most $
		\max\{n_\gamma, (\lambda/k_l)^{\frac{1}{\lambda_\ell}}\}
		\in
		\mc{O}(\gamma^{\frac{1}{r_\ell}})
	$ values in $\gamma^{-1}(\lambda)$. Specifically, the possible
	elements of $\gamma^{-1}(\lambda)$ are the non-negative integers no
	greater than $\max\{n_\gamma, (\lambda/k_l)^{\frac{1}{r_\ell}}\}$.
	Thus, one efficient way to compute $\gamma^{-1}(\lambda)$ is to
	simply compute $\gamma$ on each of these candidate integers and verify if the
	result is $n$. Since each computation can be done in polynomial-time
	and there are at most a polynomial number of candidate values, this
	is an efficient computation of $\gamma^{-1}(\lambda)$. Note that this
	reasoning also implies that $\abs{\gamma^{-1}(\lambda)}$ is
	upper bounded by a polynomial in $\lambda$.

	It then follows that the circuit families $(A, B, C)$ described
	above are efficient. Indeed, checking if $\gamma^{-1}(\lambda)$ is
	empty or not can be done efficiently. If it is empty, $A_\lambda$,
	$B_\lambda$ and $C_\lambda$ are simply circuits which are already
	assumed to be from an efficient family. If $\gamma^{-1}(\lambda)$
	is not empty, then $A_\lambda$, $B_\lambda$, and $C_\lambda$ is
	simply $\abs{\gamma^{-1}(\lambda)}$ circuits in parallel with a
	minimal overhead to route the inputs to the right circuit. Since
	$\abs{\gamma^{-1}(\lambda)}$ is polynomially bounded, this remains
	efficient.

	Thus, by the security of $(G, E)$ and the fact that
	$2^{-c}$ is negligible, there exists a negligible
	function $\eta$ such that
	\begin{equation}
		\E_{\substack{k \gets \{0,1\}^{\kappa(\lambda)}\\x_B, x_C \gets
		\{0,1\}^{d(\lambda)}}}
		\bra{f_\lambda(k,x_B), f_\lambda(k,x_C)}
		\left(B_\lambda \tensor C_\lambda\right)
		\left(x_B \tensor A_\lambda(G_\lambda(k)) \tensor x_C\right)
		\ket{f_\lambda(k,x_B), f_\lambda(k,x_C)}
		\leq
		\eta(\lambda)
	\end{equation}
	Let $v(\lambda)$ denote the left-hand side of the above. By our
	construction, the attack $(A, B, C)$ averages multiple instances of
	the attack $(A', B', C')$ when possible. More precisely, if we let
	$v'(n)$ denote the success probability of the attack $(A', B', C')$
	for the security parameter $n$, we have that
	\begin{equation}
		\E_{n \in \gamma^{-1}(\lambda)} v'(n) = v(\lambda) \leq
		\eta(\lambda)
	\end{equation}
	if $\gamma^{-1}(\lambda) \not= \varnothing$. In particular, this
	implies that $
		v'(n) \leq \abs{\gamma^{-1}(\gamma(n))} \cdot
		\eta\circ\gamma(n)
	$. By our assumption on the map $\gamma$, $\gamma(n)$ is upper bounded by a
	polynomial in $n$, which by our previous remarks also implies that
	$\abs{\gamma^{-1}(\gamma(n))}$ is upper bounded by a polynomial in
	$n$. Since $\gamma$ is in $\Omega(n^{r_l})$ and is non-decreasing,
	we have that
	$\eta \circ \gamma$ is negligible. It follows that $v'(n)$ is
	negligible, which is the desired result.
\end{proof}

While the previous lemma gives sufficient condition for when a
``$\gamma$ slowdown'' maintains the security and correctness guarantees
of copy-protected functions, it does not establish the existence of a
suitable $\gamma$. This is done, implicitly, by the following lemma.

\begin{lemma}
\label{th:slowdown-exists}
	Let $g : \N \to \N$ be a map such that $g\in \mc{O}(n^d)$ for
	a $d \in \N^+$. Then, there exists an
	efficiently computable non-decreasing map $f : \N \to \N$ such that
	$g \circ f$ is eventually strictly smaller than $n$ and
	$f \in \Omega(\lambda^r) \cap \mc{O}(\lambda^r)$ for a~$r \in \R^+$.
\end{lemma}
\begin{proof}
	Let $k, n_g \in \N$ be such that
	$n \geq n_g \implies g(\lambda) \leq k n^d$.
	Now, let $f$ be defined by
	\begin{equation}
		f(n)
		=
		\begin{cases}
			0
			& \qq{if} n < \sqrt[d]{k} n_g^{2d}
			\\
			\left\lfloor
				\frac{1}{\sqrt[d]{k}} n^{\frac{1}{2d}}
			\right\rfloor
			& \qq{else.}
		\end{cases}
	\end{equation}
	Trivially, $f$ is non-decreasing and efficiently computable.
	Then, for all $n > \max\{\sqrt[d]{k}n_g^{2d},1\}$, we
	have that
	\begin{equation}
		g \circ f(n)
		=
		g\left(
			\left\lfloor
				\frac{1}{\sqrt[d]{k_u}} n^{\frac{1}{2d}}
			\right\rfloor
		\right)
		\leq
		n^\frac{1}{2}
		<
		n.
	\end{equation}
	It now remains to show that there exists an $r \in \R^+$ such that
	$f \in \Omega(\lambda^r)$. Note that for all sufficiently large
	values of $\lambda$ we have that $f(\lambda) \geq
	\frac{1}{\sqrt[d]{k}} \lambda^\frac{1}{2d} - 1 \geq
	\frac{1}{2\sqrt[d]{k}} \lambda^\frac{1}{2d}$. Thus, $f \in
	\Omega(\lambda^\frac{1}{2d})$. We similarly find that $f \in
	\mc{O}(\lambda^\frac{1}{2d})$.
\end{proof}

Pulling these lemmas together, we can prove \cref{th:language-from-prf}.

\begin{proof}[Proof of \cref{th:language-from-prf}]
	Let $f' = (f'_n : \{0,1\}^{\kappa(n)} \times \{0,1\}^{d(n)} \to
	\{0,1\}^{c(n)})_{n \in \N}$ be the copy-protected PRF. We know that
	$d + c \in \mc{O}(n^r)$ for some $r \in \R^+$ else it would not be
	possible to efficiently compute this PRF.
	Then, by \cref{th:slowdown-exists}, there exists
	a non-decreasing efficiently computable $\gamma$ in both
	$\Omega(\lambda^\frac{1}{2d})$ and $\mc{O}(\lambda^\frac{1}{2d})$
	such that $(c + d)\circ \gamma(n)$ is eventually strictly smaller
	than $n$.

	Thus, by applying the transformation of \cref{th:slowdown} with this
	$\gamma$ to $f'$, we obtain a sequence of maps $f$ satisfying the
	assumptions of \cref{th:language}.
\end{proof}

%% file: fig-circuit-tfkw-derandomized.tex

\begin{tikzpicture}[thick,x=0.72em,y=0.72em]

	\node(a) at (0,0)             {$A_\lambda$};
	\draw ($(a) + (2,3)$) rectangle ($(a) - (2,3)$);

	\node(b) at ($(a) + ( 10,4)$) {$B_\lambda$};
	\draw ($(b) + (2,3)$) rectangle ($(b) - (2,3)$);

	\node(c) at ($(a) + ( 10,-4)$) {$C_\lambda$};
	\draw ($(c) + (2,3)$) rectangle ($(c) - (2,3)$);

	\node(r) at ($(a) + ( 20, 2)$) {$R_\lambda$};
	\draw ($(r) + (2,9)$) rectangle ($(r) - (2,9)$);

	\node(s) at ($(a) + (-10, 2)$) {$S_\lambda$};
	\draw ($(s) + (2,9)$) rectangle ($(s) - (2,9)$);

	\draw
		($(s) + ( 2,-2)$)
		to node[pos=1/2, above] {\footnotesize$\ketbra{x^\theta}$}
		($(a) + (-2,0)$);

	\draw
		($(s) + ( 2, 4)$)
		to node[pos=3/16, above] {\footnotesize$\theta$}
		($(b) + (-2, 2)$);

	\draw
		($(s) + ( 2,-8)$)
		to node[pos=3/16, above] {\footnotesize$\theta$}
		($(c) + (-2,-2)$);

	\draw
		($(s) + ( 2, 8)$)
		to node[pos=3/26, above] {\footnotesize$x$}
		($(r) + (-2, 8)$);

	\draw
		($(a) + ( 2, 2)$)
		to
		($(b) + (-2,-2)$);

	\draw
		($(a) + ( 2,-2)$)
		to
		($(c) + (-2, 2)$);

	\draw
		($(b) + ( 2, 0)$)
		to
		($(r) + (-2, 2)$);

	\draw
		($(c) + ( 2, 0)$)
		to
		($(r) + (-2,-6)$);

	\draw
		($(r) + ( 2, 0)$)
		to node[pos=5/7, above] {\footnotesize$0$ or $1$}
		($(r) + ( 8, 0)$);

	\node (x) at ($(s) + (-8, 2)$) {};
	\node (t) at ($(s) + (-8,-2)$) {};
	\node (split) at ($(s) + (-10,0)$) {};

	\draw
		($(x)$)
		to node[pos=1/2, above] {\footnotesize$x$}
		($(s) + (-2, 2)$);

	\draw
		($(t)$)
		to node[pos=1/2, above] {\footnotesize$\theta$}
		($(s) + (-2,-2)$);

	\draw
		($(split) - (6,0)$)
		to node[pos=1/6, above] {\footnotesize$s$}
		($(split)$);
	\draw
		($(split)$)
		to
		($(x)$);
	\draw
		($(split)$)
		to
		($(t)$);

	\draw ($(s) - (2+6+2+2,9+2)$) rectangle ($(r) + (2+2,9+2)$);

	\node (f) at ($(s) - (2+6+2,-9)$) {$\tilde{C}_\lambda$};

\end{tikzpicture}

%% file: fig-circuit-promise-odd.tex
\begin{tikzpicture}[thick]

	\tikzset{
		cross/.style={
			path picture={
				\draw[thick,black]
				(path picture bounding box.north)
				--
				(path picture bounding box.south)
				(path picture bounding box.west)
				--
				(path picture bounding box.east);
			}
		}
	}

	\tikzstyle{operator} = [draw,fill=white,minimum size=1.5em]
	\tikzstyle{phase} = [draw,fill,shape=circle,minimum size=0.5em,inner sep=0pt]
	\tikzstyle{xor} = [draw,cross,shape=circle,minimum size=0.75em,inner
	sep=0pt]
	\matrix[row sep=0.5em, column sep=1em] (circuit) {

    \node           (i1)  {};          &
	\node[operator] (T11) {$\Tr$};     \\

	\node           (i2)  {};    &
	\node[operator] (T21) {$\Tr$}; & \\

	\node           (i3)  {}; &
	\node[operator] (T31) {$\Tr$};     \\

	\\

	\node[operator] (s4)  {$\ket{0}$}; &
	\coordinate (o4);       \\
};
    \begin{pgfonlayer}{background}
        \draw[thick]
			(i1) -- (T11)
			(i2) -- (T21)
			(i3) -- (T31)
			(s4) -- (o4);
			;
    \end{pgfonlayer}
\end{tikzpicture}

%% file: fig-circuit-promise-even.tex
\begin{tikzpicture}[thick]

	\tikzset{
		cross/.style={
			path picture={
				\draw[thick,black]
				(path picture bounding box.north)
				--
				(path picture bounding box.south)
				(path picture bounding box.west)
				--
				(path picture bounding box.east);
			}
		}
	}

	\tikzstyle{operator} = [draw,fill=white,minimum size=1.5em]
	\tikzstyle{phase} = [draw,fill,shape=circle,minimum size=0.5em,inner sep=0pt]
	\tikzstyle{xor} = [draw,cross,shape=circle,minimum size=0.75em,inner
	sep=0pt]
	\matrix[row sep=0.5em, column sep=1em] (circuit) {

    \node           (i1)  {};          &
    \node[operator] (H11) {$H$};       &
                                       &
    \node[phase]    (P13) {};          &
                                       &
	\node[operator] (T15) {$\Tr$};     \\

	\node           (i2)  {};    &
	                             &
	\node[operator] (H22) {$H$}; &
	                             &
	\node[phase] (P24) {};       &
	\node[operator] (T25) {$\Tr$};     \\

	\\

	\node           (i3)  {}; &
	\node[phase]    (P31) {}; &
	                          &
	                          &
	                          &
	\node[operator] (T35) {$\Tr$};     \\

	\node           (i4)  {}; &
	                          &
	\node[phase]    (P42) {}; &
	                          &
	                          &
	\node[operator] (T45) {$\Tr$};\\

	\\

	\node           (i5)  {}; &
	                          &
	                          &
	\node[phase]    (P53) {}; &
	                          &
	\node[operator] (T55) {$\Tr$};\\

	\node           (i6)  {}; &
	                          &
	                          &
	                          &
	\node[phase]    (P64) {}; &
	\node[operator] (T65) {$\Tr$};\\

	\\

	\node[operator] (s7)  {$\ket{0}$}; &
	                          &
	                          &
	\node[xor]    (P73) {}; &
	\node[xor]    (P74) {}; &
	\coordinate (o7);       \\
    };
    \begin{pgfonlayer}{background}
        \draw[thick]
			(i1) -- (T15)
			(i2) -- (T25)
			(i3) -- (T35)
			(i4) -- (T45)
			(i5) -- (T55)
			(i6) -- (T65)
			(s7) -- (o7)
			(H11) -- (P31)
			(H22) -- (P42)
			(P13) -- (P53) -- (P73)
			(P24) -- (P64) -- (P74)
			;

    \draw[decorate,decoration={brace},thick]
		($(i2) + (-0.25em,-0.25em)$) to
		node[midway,left] {$\tsf{A}$}
		($(i1) + (-0.25em,0.25em)$);

	\draw[decorate,decoration={brace},thick]
		($(i4) + (-0.25em,-0.25em)$) to
		node[midway,left] {$\tsf{B}$}
		($(i3) + (-0.25em,0.25em)$);

	\draw[decorate,decoration={brace},thick]
		($(i6) + (-0.25em,-0.25em)$) to
		node[midway,left] {$\tsf{S}$}
		($(i5) + (-0.25em,0.25em)$);
    \end{pgfonlayer}
\end{tikzpicture}

%% file: fig-language.tex
	\begin{tikzpicture}
		\node (w) at (0,0) {$w$};
		\draw[thick]
		($(w) + (1.25,0.4)$) rectangle ($(w) - (1.25,0.4)$);

		\node (x) at ($(w) + (8.1,-1)$) {$x$};
		\draw[thick]
		($(x) + (0.6,0.4)$) rectangle ($(x) - (0.6,0.4)$);
		\node (y) at ($(x) + (1,0)$) {$y$};
		\draw[thick]
		($(y) + (0.4,0.4)$) rectangle ($(y) - (0.4,0.4)$);
		\node (z) at ($(y) + (0.65,0)$) {$z$};
		\draw[thick]
		($(z) + (0.25,0.4)$) rectangle ($(z) - (0.25,0.4)$);

		\node (resyes) at ($(w)+(13,-1)$)
			{$(g_\abs{w}(x)\cdot y)\xor(z\cdot z)$};

		\node (w2) at ($(w) + (8.75,1)$) {$w$};
		\draw[thick]
		($(w2) + (1.25,0.4)$) rectangle ($(w2) - (1.25,0.4)$);
		
		\node (resno) at ($(w) +(13,1)$)
			{$w \cdot w$};

		\draw[thick,->]
		($(w) + (1.35,0)$)
		to node[above]
		{$\abs{w} \stackrel{?}{>} (c + d)(\abs{w})$}
		($(w) + (5.5,0)$)
		to
		($(w) + (6,-1)$)
		to node[above] {Yes}
		($(w) + (7.4,-1)$);
		\draw[thick,->] ($(z) + (0.35,0)$) to (resyes);

		\draw[thick,->]
		($(w) + (5.5,0)$)
		to
		($(w) + (6,1)$)
		to node[above] {No}
		($(w) + (7.4,1)$);
		\draw[thick,->] ($(w2) + (1.35,0)$) to (resno);
	\end{tikzpicture}

%% file: 7-mlr.tex
\section{Comparing Ingenerable and Martin-L\"of Random Sequences}
\label{sc:mlr}

The goal of this section is to show that the concepts of ingenerable
sequences and of Martin-L\"of random sequences are distinct but related.
For the rest of the appendix, ``ingenerable'' is understood to mean
``computably ingenerable''.

\begin{itemize}
	\item
		For reasonable choices of $\ell$, there exists
		weakly ingenerable $\ell$-sequences which are not
		ingenerable (\cref{th:ingen-not-weak-ingen}).
	\item
		For any non-constant polynomial $\ell$, every Martin-L\"of
		random sequence yields a weakly ingenerable $\ell$-sequence via
		a natural bijection (\cref{th:martin-lof-is-ingenerable}).
	\item
		There exists ingenerable $\ell$-sequences which do not
		correspond to Martin-L\"of random sequences under the same
		natural bijection mentioned above
		(\cref{th:ingenerable-is-not-martin-lof}).
\end{itemize}

For any non-constant polynomial $\ell$, \cref{fg:ingenerable-martin-lof}
shows the relation obtained between the sets of weakly
ingenerable $\ell$-sequences, ingenerable $\ell$-sequences, and
Martin-L\"of random sequences obtained from the results above.

\begin{figure}[H]
	\begin{center}
	\begin{tikzpicture}
		\node[align=center] (wig) at (0,0) {Weakly Computably Ingenerable
		$\ell$-sequences};
		\node[align=center] (sig) at (-3.2,-1.5) {Computably Ingen-\\erable
		$\ell$-sequences};
		\node[align=center] (mlr) at (3,-1.5) {Martin-L\"of\\Random};

		\draw[thick, rounded corners=5]
			($(wig) + (6,0.5)$)
			rectangle
			($(wig) - (6,2.75)$);

		\draw[thick] ($(sig)+(1,0)$) ellipse (3.5 and 1);
		\draw[thick] ($(mlr)-(1,0)$) ellipse (3.5 and 1);
		

	\end{tikzpicture}
	\end{center}
	\caption{\label{fg:ingenerable-martin-lof}%
		The relation between the sets of
		weakly computably ingenerable $\ell$-sequences, ingenerable
		$\ell$-sequences, and Martin-L\"of random sequences, via the
		$\text{cut}_\ell$ bijection, for any non-constant polynomial
		$\ell : \N \to \N$. All three sets are distinct, but the
		precise relation between ingenerable $\ell$-sequences and
		Martin-L\"of random sequences remains unknown.}
\end{figure}
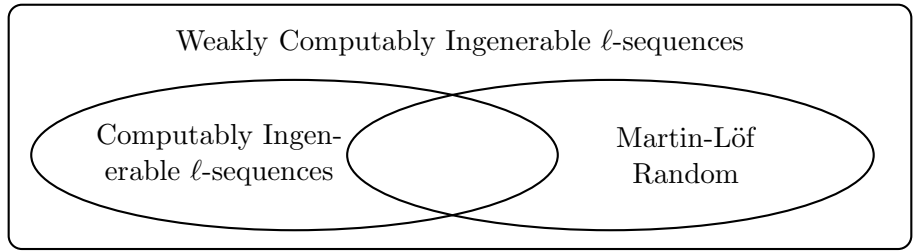

Before moving on to considering Martin-L\"of randomness, we separate the
notions of ingenerability and weak ingenerability, except in the trivial
cases where all $\ell$-sequences are ingenerable.

\begin{theorem}\label{th:ingen-not-weak-ingen}
	Let $\ell : \N \to \N$ be a computable map not in $\mc{O}(\log)$.
	Then, there exists a weakly ingenerable $\ell$-sequence which is not
	ingenerable.
\end{theorem}
\begin{proof}
	Let $(s_n)_{n\in\N}$ be an ingenerable $\ell$-sequence. For all
	$n \in \N$, we define
	\begin{equation}
		s^0_n
		=
		\begin{cases}
			0^{\ell(n)}
			& \qq{if} n \text{ is even}
			\\
			s_n & \qq{if} n \text{ is odd}
		\end{cases}
		\qq{and}
		s^1_n
		=
		\begin{cases}
			s_n
			& \qq{if} n \text{ is even}
			\\
			0^{\ell(n)}
			& \qq{if} n \text{ is odd}
		\end{cases}
	\end{equation}

	We first show that at least one of $(s^0_n)_{n \in \N}$ or
	$(s^1_n)_{n \in \N}$ is not ingenerable. Aiming for a contradiction,
	assume both are ingenerable. Let $(C_n)_{n \in \N}$ be a uniform
	$(0,\ell)$-circuit family such that $C_n$ simply outputs
	$\ketbra{0^{\ell(n)}}$ for all $n \in \N$. As we have assumed that
	both sequences are ingenerable, there exists polynomials $p_b$ and
	integers $\tilde{n}_b$ such that
	\begin{equation}
	\label{eq:ingen=>weak-ingen}
		n \geq \tilde{n}_b
		\implies
		\bra*{s^b_n} C_n(\varepsilon) \ket*{s^b_n}
		<
		p_b(n) \cdot 2^{-\ell(n)}.
	\end{equation}
	for both values of $b \in \{0,1\}$. Taking $\tilde{n} = \tilde{n}_0
	+ \tilde{n}_1$ and $p = p_0 + p_1$, this implies that
	\begin{equation}
		n \geq \tilde{n}
		\implies
		\bra*{s^b_n} C_n(\varepsilon) \ket*{s^b_n}
		<
		p(n) \cdot 2^{-\ell(n)}
	\end{equation}
	for both $b \in \{0,1\}$. Since at least one of $s_n^0$ or $s_n^1$
	is always the all-zero string, which precisely what is produced by
	$C_n$, this implies that

	\begin{equation}
		n \geq \tilde{n}
		\implies
		1
		<
		p(n) \cdot 2^{-\ell(n)}
	\end{equation}
	from which we can conclude that $\ell \in \mc{O}(\log)$, a
	contradiction. Hence, at least one of the two sequences is not
	ingenerable.

	Next, we show that $(s^b_n)_{n \in \N}$ is weakly ingenerable for
	both values of $b \in \{0,1\}$. Since $(s_n)_{n \in \N}$ is
	ingenerable, there exist a polynomial $p$ such that for any uniform
	$(0,\ell)$-circuit family $(C_n)_{n \in \N}$ there is a 
	$\tilde{n} \in \N$ such that
	\begin{equation}
		n \geq \tilde{n}
		\implies
		\bra{s_n} C_n(\varepsilon) \ket{s_n}
		<
		p(n)\cdot2^{-\ell(n)}.
	\end{equation}
	Thus, for both $b \in \{0,1\}$,
	\begin{equation}
		\bra*{s^b_n} C_n(\varepsilon) \ket*{s^b_n}
		<
		p(n)\cdot2^{-\ell(n)}
	\end{equation}
	infinitely often. Hence, $(s^b_n)_{n \in \N}$ is weakly ingenerable.
\end{proof}

\subsection{Preliminaries}

We review here the basics of Kolmogorov complexity and Martin-L\"of
randomness as well as establish some extra notation pertaining to bit
strings. We generally follow the textbook of Li and
Vit\'anyi~\cite{LV19}.

\subsubsection{More on Bit Strings and the $\text{cut}_\ell$ Map}
We denote by $\{0,1\}^\infty$ be the set of all infinite sequences of
bits. When we wish to emphasize that a sequence of $0$'s and $1$'s be
interpreted as a bit string and not a number, we will write them as
sequences of $\ttt{0}$'s and $\ttt{1}$'s.

For a Turing machine $\ttt{T}$, we let $\langle\ttt{T}\rangle \in
\{0,1\}^*$ denote a reasonable binary encoding of $\ttt{T}$.
The specifics of this encoding are not needed for this work. Similarly,
for any $n \in \N$ we let $\langle n \rangle \in \{0,1\}^*$ denote the
shortest representation of $n$ in binary. Here, we adopt the same
convention as in \cite{LV19} and represent the number $0$ with the empty
string and continue from there. Thus, $\langle 0 \rangle =
\varepsilon$, $\langle 1 \rangle = \ttt{0}$,
and $\langle3\rangle=\ttt{10}$. This yields a bijection between
$\{0,1\}^*$ and $\N$ as well as the relation
\begin{equation}
	\abs{\langle n \rangle}
	=
	\floor{\log(n+1)}.
\end{equation}
We also define $\langle n \rangle_m$ to be the binary representation of
$n$ that is left-padded by $0$'s to be of length at least $m$,
\emph{e.g.} $\langle 5 \rangle_6 = \ttt{000100}$.

For any string $x \in \{0,1\}^*$, we let $_\text{d}x$ be the string
which repeats every bit of $x$ twice in succession. For example,
$_\text{d}\langle 5 \rangle = _\text{d}\ttt{101} = \ttt{110011}$.

Let $s \in \{0,1\}^* \cup \{0,1\}^\infty$ be a finite or infinite string
and let $n, m \in \N^+$ be two strictly positive integers. If $n \leq m
\leq \abs{s}$, we let $s_{[n,m]} \in \{0,1\}^{m-n+1}$ be the
string composed precisely of the $n$-th to $m$-th bits, inclusively, of
$s$. Otherwise, we set $s_{[n,m]} = \varepsilon$.

For any map $\ell : \N \to \N$, we denote the set of $\ell$-sequences by
 $\{0,1\}^*_\ell$.  
We now define a map $\text{cut}_\ell : \{0,1\}^\infty \to \{0,1\}^*_\ell$.
Intuitively, $\text{cut}_\ell$ will ``cut'' an infinitely sequence
$s \in \{0,1\}^\infty$ into finite strings $s_\lambda$ of length
$\ell(\lambda)$. The string $s_0$ will be first $\ell(0)$ bits of $s$,
the string $s_1$ will be the next~$\ell(1)$ bits, $s_2$ will be the
next~$\ell(2)$ bits, and so on and so forth. If the map $\ell$ is
non-zero infinitely often, then $\text{cut}_\ell$ is a bijection and the
inverse map $\text{cut}^{-1}_\ell : \{0,1\}^*_\ell \to \{0,1\}^\infty$
is simply the infinite concatenation, in order, of all strings in a
sequence $(s_\lambda)_{\lambda \in \N}$. More formally, let
$L : \N \to \N$ be defined by $n \mapsto \sum_{i = 0}^n \ell(i)$.
Then,~$\text{cut}_\ell(s) = (s_n)_{n \in \N}$ where
\begin{equation}
	s_n
	=
	\begin{cases}
		s_{[1,L(0)]} & \qq{if} n = 0
		\\
		s_{[1+L(n-1),L(n)]} & \qq{else.}
	\end{cases}
\end{equation}

\subsubsection{Kolmogorov Complexity}

We now move on to defining the Kolmogorov complexity of a string. This
is also known as the \emph{plain} Kolmogorov complexity when we wish to
explicitly distinguish it from other variations of this definition, such
as the prefix Kolmogorov complexity.\footnote{%
	Note that all our references on this topic use $C$ for the plain
	Kolmogorov complexity and reserve $K$ for the prefix Kolmogorov
	complexity. We use $K$ for the Kolmogorov complexity as we typically
	use $C$ to denote circuits.}

\begin{definition}
\label{df:kolmogorov}
	Let $\ttt{T}$ be a Turing machine. The Kolmogorov complexity of a
	string $x \in \{0,1\}^*$ with respect to the machine $\texttt{T}$,
	denoted $K_\ttt{T}(x)$, is defined by
	\begin{equation}
		K_\texttt{T}(x)
		=
		\begin{cases}
			\infty
			& \qq{if}
			\ttt{T}(y) \not= x \text{ for all } y \in \{0,1\}^*
			\\
			\min\left\{ \abs{y} \;:\; \ttt{T}(y) = x\right\}
			& \qq{else.}
	\end{cases}
\end{equation}
\end{definition}

If $\ttt{U}$ is a universal Turing machine, in the sense that
$\ttt{U}\left(\langle\ttt{T}\rangle , x\right) = \ttt{T}(x)$ for all\
strings $x \in \{0,1\}^*$ where $\ttt{T}(x)$ halts, then there exists a
constant $c \in \N$ such
that $K_{\ttt{U}}(x) \leq K_{\ttt{T}}(x) + c$. Indeed, it suffices to
take $c = \abs{\langle\ttt{T}\rangle}$. It follows that for any two
universal Turing machines $\ttt{U}$ and $\ttt{U}'$ there exists a
constant $c$ such that~$K_\ttt{U}(x) \leq K_{\ttt{U}'}(x) + c$ for
all strings $x \in \{0,1\}^*$.

With the above in mind, we fix for the remainder of this section a
universal Turing machine $\ttt{U}$ and set $K = K_\ttt{U}$. Up a to an
additive constant, the particular choice of $\ttt{U}$ does not matter.

As an immediate result, we note that there exists a constant $c$ such
that $K(x) \leq \abs{x} + c$ for all strings $x \in \{0,1\}^*$. Indeed,
let $\ttt{Id}$ be the Turing machine which immediately halts. We then
have that $\ttt{U}(\langle{Id}\rangle , x) = x$ for all
$x \in \{0,1\}^*$ and so $K(x)\leq\abs{\langle\ttt{Id}\rangle}+\abs{x}$.

Finally, we also define the \emph{conditional} Kolmogorov
complexity. In short, we assume here that the Turing machine also has
access to another string $z$ when trying to compute $x$ and the length
of $z$ is not counted in the resulting measure of complexity.

\begin{definition}
	Let $\ttt{T}$ be a Turing machine. The Kolmogorov complexity of a
	string  $x \in \{0,1\}^*$ conditioned on $z \in \{0,1\}^*$ with respect
	to $\tt{T}$ is defined by
	\begin{equation}
		K_\ttt{T}(x | z)
		=
		\begin{cases}
			\infty & \qq{if} \ttt{T}(_\text{d}z , 01 , y)\not= x
			\text{ for all } y \in \{0,1\}^*
			\\
			\min\{\abs{y} \;:\; \ttt{T}(_\text{d}z , \ttt{01} , y)=x\}
			&\qq{else.}
		\end{cases}
	\end{equation}
\end{definition}
Note that the double encoding of $z$ followed by $\ttt{01}$ allows the
Turing machine to unambiguously distinguish between $z$ and $y$.

We again fix a single universal Turing machine $\ttt{U}$ and always
measure the conditional Kolmogorov complexity with respect to this
machine.

\subsubsection{Martin-L\"of Random Sequences}

A very succinct characterization of Martin-L\"of randomness can be given
in terms of prefix Kolmogorov complexity. We take this characterization
as our definition. We do not formally define prefix Kolmogorov
complexity as we will not directly be using this characterization, but
it roughly corresponds to restricting the plain Kolmogorov complexity to
only consider Turing machines whose behaviours are completely
characterized by their actions on a prefix set.

\begin{definition}[{\cite[Theorem 3.5.1]{LV19}}]
	A string $w \in \{0,1\}^\infty$ is Martin-L\"of random if there
	exists a constant $c$ such that the prefix Kolmogorov complexity of
	$w_{[1,n]}$ is at least $n - c$ for all $n \in \N$.
\end{definition}

At a high level, this characterization captures the idea that there
should be a limit on how much the prefixes of a Martin-L\"of random
sequence can be compressed. For our needs, the following two theorems
will be sufficient. The first, due to Miller and
Yu \cite{MY08} (although we use a restricted version of the formulation
given in \cite[Theorem 2.5.4]{LV19}) is expressed in terms of the
conditional plain Kolmogorov complexity. The second is due to Schnorr
\cite{Sch71} and is based on the concept of computable
martingales.\footnote{%
	Slightly relaxing the conditions on $f$ in this theorem, namely only
	requiring it to be \emph{weakly computable} and letting its
	codomain be $\R^+_0$, actually yields a characterization of
	non-Martin-L\"of randomness.}

\begin{theorem}
\label{th:miller-yu}
	For all Martin-L\"of random $w \in \{0,1\}^\infty$, there exists a
	$c \in \N$ such that for all $n \in \N$
	\begin{equation}
		K(w_{[1:n]} | \langle n \rangle) \geq n - 2\log(n) - c.
	\end{equation}
\end{theorem}

\begin{theorem}
\label{th:schnorr}
	Let $f : \{0,1\}^* \to \N$ be a computable map such that for all $x
	\in \{0,1\}^*$ we have that
	\begin{equation}
		f(x) = \frac{f(x , 0) + f(x , 1)}{2}.
	\end{equation}
	Any $w \in \{0,1\}^\infty$ such that $\limsup_{n \to \infty}
	f(w_{[1:n]}) = \infty$ is not Martin-L\"of random.
\end{theorem}

\subsection{Ingenerable Sequences Need Not be Martin-L\"of Random}

We prove in this section \cref{th:ingenerable-is-not-martin-lof} which
states that, under the action of the $\text{cut}^{-1}_\ell$ bijection
for a computable map $\ell$, ingenerable sequences may not yield
Martin-L\"of random sequences. The core of our argument is that we can
find ingenerable sequences such that their concatenation yields an
infinite sequence of bits with $0$'s at infinitely many locations which
can be computed. This is sufficient to ensure non-Martin-L\"of
randomness as it will allow us to construct a computable martingale, in
the sense of \cref{th:schnorr}, who's value will grow to infinity.

In practice, we will use the following corollary of \cref{th:schnorr}.

\begin{corollary}
\label{th:schnorr-2}
	Let $L : \N \to \N^+$ be a computable strictly monotone map.
	If a string $s \in \{0,1\}^\infty$ is such that $s_{[L(i):L(i)]} = 0$
	for all $i \in \N$, then $s$ is not Martin-L\"of random.
\end{corollary}
\begin{proof}
	Consider a Turing $\ttt{T}$ machine which, on input of a string $x
	\in \{0,1\}^*$,  implements the following algorithm:
	\begin{enumerate}
		\item
			Initialize a counter $v$ with value $1$ and a counter $i$
			with value $0$.
		\item
			Compute $L(i)$. If $L(i) > \abs{x}$, then output the value
			of $v$ and halt. Else:
			\begin{enumerate}
				\item
					Check if $x_{[L(i),L(i)]}$ is $0$. If it is,
					double the value of $v$. If it is not, set $v$ to be
					$0$.
				\item
					Increment $i$ by setting it to $i + 1$ and return to
					step 2.
			\end{enumerate}
	\end{enumerate}
	Note that $\ttt{T}$ halts on every input since $L$ is strictly
	monotone. Let $f : \{0,1\}^* \to \N$ be the function computed by
	$\ttt{T}$ and note that $f(x) = \frac{f(x,0)+f(x,1)}{2}$ for
	all strings~$x \in \{0,1\}^*$. By \cref{th:schnorr}, it suffices to
	show that $\lim_{n \to \infty}f(s_{0:n}) = \infty$ to obtain
	that~$s$ is not Martin-L\"of random. To do this, we simply note
	that
	\begin{equation}
		f(s_{[1,n]}) = 2^{\abs{L(\N) \cap [n]}}
	\end{equation}
	and that since $\abs{L(\N)}$ is infinite, this value tends to
	infinity as $n$ grows. Indeed, consider the value of the register
	$v$ maintained by the Turing machine described above when it is run
	on $s_{[1:n]}$. Step~2(a) will be executed exactly
	$\abs{L(\N)\cap[n]}$ times. By hypothesis, $s_{[L(i),L(i)]}$ is
	always $0$ and thus the value of~$v$ is doubled every time. This
	yields the desired result.
\end{proof}

At this point, it essentially suffices to argue that there are
ingenerable sequences with infinitely many zeros at precisely computable
locations.

\begin{lemma}
\label{th:zero-leading-ingenerable}
	Let $(s_\lambda)_{\lambda \in \N}$ be an ingenerable
	$\ell$-sequence. For all $\lambda \in \N$, let $s'_\lambda$ be $s_\lambda$,
	except with the first bit replaced by $0$, unless
	$s_\lambda = \varepsilon$ in which case no change is made. Then the
	sequence $(s'_\lambda)_{\lambda \in \N}$ is also ingenerable.
\end{lemma}
\begin{proof}
	Let $(C'_\lambda)_{\lambda \in \N}$ be a uniform $(0,\ell)$-circuit
	family. (If no such family exists, such as if $\ell$ is
	uncomputable, then the proof is complete.) For all $\lambda \in \N$,
	let $C_\lambda$ be the circuit which first runs the circuit
	$C'_\lambda$ and then
	replaces the first qubit with the maximally mixed state, unless
	$C'_\lambda$ is a~$(0,0)$-circuit in which case
	$C_\lambda = C'_\lambda$. Clearly, $(C_\lambda)_{\lambda \in \N}$ is
	also a uniform $(0, \ell)$-circuit family. Note that
	\begin{equation}
		\frac{1}{2}
		\bra*{s'_\lambda} C'_\lambda(\varepsilon) \ket*{s_\lambda}
		\leq
		\bra{s_\lambda} C_\lambda(\varepsilon) \ket{s_\lambda}
	\end{equation}
	since, conditioned on $C_\lambda$ correctly outputting the last
	$\ell(\lambda)-1$ bits of $s_\lambda$, $C'_\lambda$ has a
	probability $\frac{1}{2}$ of outputting $s'_\lambda$. Note that if
	$\ell(\lambda) = 0$, the inequality still holds as the left-hand
	side is $\frac{1}{2}$ and the right hand side is $1$. Now, since
	$(s_\lambda)_{\lambda \in \N}$ is ingenerable, there exists a
	polynomial $p$ such that
	\begin{equation}
		\bra{s_\lambda} C_\lambda(\varepsilon) \ket{s_\lambda}
		\leq
		p(\lambda) \cdot 2^{-\ell(\lambda)}
	\end{equation}
	for all sufficiently large values of $\lambda$. Combining this with
	the previous inequality yields the desired result as $2p$ is a
	polynomial.
\end{proof}

We now pull everything together to state and prove the theorem.

\begin{theorem}
\label{th:ingenerable-is-not-martin-lof}
	Let $\ell : \N \to \N$ be a computable function which is nonzero
	infinitely often. Then, there exists an ingenerable $\ell$-sequence
	$(s_\lambda)_{\lambda \in \N}$ such that
	$\text{cut}^{-1}_\ell\left((s_\lambda)_{\lambda \in \N}\right)$ is
	not Martin-L\"of random.
\end{theorem}
\begin{proof}
	By \cref{th:ingenerable}, there exists an ingenerable
	$\ell$-sequence $(s_\lambda)_{\lambda \in \N}$ and by
	\cref{th:zero-leading-ingenerable} we can assume that the first bit
	of every non-empty $s_\lambda$ in this sequence is $0$. Let
	$N = \{n_0, n_1, n_2, \ldots\} \subseteq \N$ be the set of values on
	which $\ell$ is nonzero. We index the values of $N$ such that
	$n_j < n_{j + 1}$ for all~$j \in \N$. Now,
	define the map~$L : \N \to \N^+$ by
	\begin{equation}
		i \mapsto 1 + \sum_{j = 0}^{i-1} \ell(n_j).
	\end{equation}
	and note that it is computable and strictly increasing. By
	construction, the $L(i)$-th bit of
	$\text{cut}^{-1}_\ell\left(s_\lambda\right)$ is $0$ for all
	$i \in \N$. Indeed, $L(i)$ identifies the location where the first
	bit of $s_{n_i}$ will be situated in the infinite string
	$\text{cut}^{-1}_\ell\left((s_\lambda)_{\lambda \in \N}\right)$.
	Thus, by \cref{th:schnorr-2}, this string is not Martin-L\"of
	random.
\end{proof}

\subsection{Martin-L\"of Random Sequences give Weakly Ingenerable Sequences}

The key observation in this section is that if we already given the
description of a circuit $C$ which generates a string $s$ with at least
some probability $p$, \emph{i.e.:} $\bra{s}C(\varepsilon)\ket{s}\geq p$,
then we can describe $s$ with at most $\approx\log(p^{-1})$ extra bits.
Indeed, there are at most $p^{-1}$ strings satisfying the previous
inequality and and we can identify $s$ by by specifying it's location,
in lexicographic order, in the set of such strings. If $p$ is large
enough, this will be shorter than $\abs{s}$.

Using this idea, we can show that if
$\text{cut}_\ell(s) = (s_\lambda)_{\lambda \in \N}$ is not weakly
ingenerable if $\ell$ is a polynomial, as in that case all but some initial
set of segments of $s$ can be compressed by the above idea. Moreover,
the $C$ circuits can in fact be taken from a uniform family, which is to
say that they can all be generate by a single Turing machine and thus
they only have a constant contribution to the Kolmogorov complexity.
This, and a few extra technical considerations, is enough to show that
$s$ is not Martin-L\"of random. Taking the contrapositive yields the
following theorem.

\begin{theorem}
\label{th:martin-lof-is-ingenerable}
	Let $s \in \{0,1\}^\infty$ be a Martin-L\"of random sequence and let
	$\ell : \N \to \N$ be a polynomial. Then, $\text{cut}_\ell(s)$ is
	weakly ingenerable.
\end{theorem}

\begin{proof}
	Assume that $\ell(\lambda)$ is not constant. If it was, then
	$\text{cut}_\ell(s)$ would be trivially ingenerable, and
	hence weakly ingenerable, by virtue of
	$\ell \not\in \omega(\log)$ and the proof would be done. Note that
	since $\ell$ is never negative, this implies that
	$\lim_{\lambda \to \infty} \ell(\lambda) = \infty$.

	Aiming for a contradiction, assume that the $\ell$-sequence
	$\text{cut}_\ell(s) = \left(s_\lambda\right)_{\lambda \in \N}$ is
	not weakly ingenerable. Thus, there exists a uniform family
	of~$(0,\ell)$-circuits $(C_\lambda)_{\lambda \in \N}$ such that
	\begin{equation}
		\bra{s_\lambda}C_\lambda(\varepsilon)\ket{s_\lambda}
		\geq
		\left(p(\lambda) + 2\right) \cdot 2^{-\ell(\lambda)}
		\qq{where}
		p(\lambda) = 2\ell(\lambda+1)^3
	\end{equation}
	for all but finitely many values of~$\lambda$. Let $\ttt{C}$ be a
	Turing machine which generates this circuit family.

	Our goal is now to describe a Turing machine $\ttt{GEN}$ and strings
	$a_n\in\{0,1\}^*$ such that for all constants $c$, there exists an
	$n_c$ such that
	\begin{equation}
		\ttt{GEN}(_\text{d}\langle n_c \rangle,01, a_{n_c})
		=
		s_{[1,{n_c}]}
		\qq{and}
		\abs{a_{n_c}} < n_c - 2\log(n_c) - c
	\end{equation}
	Taking the contrapositive of \cref{th:miller-yu}, this is sufficient
	to conclude that $s$ is not Martin-L\"of random, yielding our
	contradiction.

	We construct $\ttt{GEN}$ from a few subroutines, which we will also
	describe as Turing machines. First, let $\ttt{SIM}_\ttt{C}$ be a
	Turing machine which, on input of
	$\langle\lambda\rangle$ and a string $y\in \{0,1\}^{\ell(\lambda)}$,
	computes the value of $\bra{y}C_\lambda(\varepsilon)\ket{y}$ within
	an additive error of $2^{-\ell(\lambda)}$. Next, let
	$\ttt{SET}_\ttt{C}$ be a Turing machine which on input of
	$\langle\lambda\rangle$ implements the following algorithm:
	\begin{enumerate}
		\item
			Initialize an empty set $\mc{S}$.
		\item
			Iterating over all strings $y \in \{0,1\}^{\ell(\lambda)}$,
			do the following:
			\begin{enumerate}
				\item
					Run $
						\ttt{SIM}_\ttt{C}(
							_\text{d}\langle\lambda\rangle
							, \texttt{01} , y)
					$.
				\item
					If this determines that $
						\bra{y}C_\lambda(\varepsilon)\ket{y}
						\geq
						(p(\lambda) + 1) \cdot 2^{-\ell(\lambda)}
					$, add $y$ to $\mc{S}$.
			\end{enumerate}
		\item
			Output a description of the set $\mc{S}$ and halt.
	\end{enumerate}

	Let's establish a bit of notation for the sets $\mc{S}$ computed by
	$\ttt{SET}_\ttt{C}$ and examine some of their properties. This will
	also allow us to define the strings $a_n$ before
	moving on to describing $\ttt{GEN}$.

	For all $\lambda \in \N$, let $\mc{S}_\lambda$ be the set ultimately
	computed by $\ttt{SET}_\ttt{C}(\langle\lambda\rangle)$. Note that,
	by the additive error bound we imposed on $\ttt{SIM}_\ttt{C}$, this
	set contains every $y$ satisfying $
		\bra{y} C_\lambda(\varepsilon)\ket{y}
		\geq
		(p(\lambda) + 2)\cdot2^{-\ell(\lambda)}
	$ and none satisfying $
		\bra{y} C_\lambda(\varepsilon) \ket{y}
		<
		p(\lambda) \cdot 2^{-\ell(\lambda)}
	$. This implies that for every $\lambda$ such that
	$p(\lambda) \not= 0$, the set $\mc{S}_\lambda$ contains at most
	$p(\lambda)^{-1} \cdot 2^{\ell(\lambda)}$ elements and that
	$s_\lambda \in \mc{S}_\lambda$ for all but finitely many values of
	$\lambda$. If $s_\lambda \in \mc{S}_\lambda$, we say that
	$s_\lambda$ is compressible. If $s_\lambda$ is compressible, we let
	$r_\lambda \in \N$ denote it's position, in the lexicographic order,
	among the elements of $\mc{S}_\lambda$. We let $n_\lambda$ be the
	maximal number of bits required to express $r_\lambda$. If
	$p(\lambda) \not= 0$, then
	\begin{equation}
		n_\lambda
		=
		\abs{\langle \abs{\mc{S}_\lambda} \rangle}
		\leq
		\log(\abs{\mc{S}_\lambda} + 1)
		\leq
		\log(2\abs{\mc{S}_\lambda})
		\leq
		\log\left(2 \cdot \frac{2^{\ell(\lambda)}}{p(\lambda)}\right)
		=
		\ell(\lambda) - 3\log(\ell(\lambda+1)).
	\end{equation}
	We let
	\begin{equation}
		\tilde{\lambda}
		=
		\max\left\{
			\left\{\lambda \;:\; s_\lambda\not\in\mc{S}_\lambda\right\}
			\cup
			\left\{\lambda \;:\; p(\lambda) = 0\right\}
		\right\}
	\end{equation}
	be the largest integer such that $s_{\tilde{\lambda}}$ is not
	compressible or that $\ell(\tilde{\lambda}) = 0$. Recall that $\ell$
	tends to infinity and so such a $\tilde{\lambda}$ exists. If
	$\lambda > \tilde{\lambda}$, then $s_\lambda \in \mc{S}_\lambda$ and
	it takes at most $\ell(\lambda) - 2\log(\ell(\lambda + 1))$ bits to
	give its location in the set.

	We now define $a_n$ for all $n \in \N$. Let $L : \N \to \N$ be
	defined by $\lambda \mapsto \sum_{i = 0}^\lambda \ell(i)$ and let
	$\lambda_n \in \N$ be the largest integer such that
	$L(\lambda_n) \leq n$, if such an integer exists. In other words,
	$\lambda_n$ is the number of complete strings from the sequence
	$(s_\lambda)_{\lambda \in \N}$ which appear in $s_{[1:n]}$. Then,
	\begin{equation}
		a_n
		=
		\begin{cases}
			s_{[1,n]}
			& \qq{if} n \leq L(\tilde{\lambda})
			\\
			\left(
			s_{[1,L(\tilde{\lambda})]}
			,
			\langle
				r_{\tilde{\lambda} + 1}
			\rangle_{n_{\tilde{\lambda} +1}}
			,
			\langle
				r_{\tilde{\lambda} + 2}
			\rangle_{n_{\tilde{\lambda} +2}}
			,
			\ldots
			,
			\langle
				r_{\lambda_n}
			\rangle_{n_{\lambda_n}}
			,
			s_{[L(\lambda_n)+1,n]}
			\right)
			& \qq{else.}
		\end{cases}
	\end{equation}
	In general (\emph{i.e.}: the second case in the above equation),
	$a_n$ includes the following information:
	\begin{itemize}
		\item
			A prefix $s_{[1,L(\tilde{\lambda})]}$ of $s_{[1,n]}$ which
			may not be compressible by our method.
		\item
			For every $s_\lambda$ which appears in whole in
			$s_{[1,n]}$ and which can be compressed by our method,~$a_n$
			includes the value of $r_\lambda$ in binary and padded to a
			predictable length.
		\item
			The suffix $s_{[L(\lambda_n)+1,n]}$, representing the
			incomplete part of $s_{\lambda_n+1}$ which is present in
			$s_{[1,n]}$.
	\end{itemize}
	For all $n$ such that $\lambda_n > \tilde{\lambda}$ we have that
	\begin{equation}
	\begin{aligned}
		\abs{a_n}
		&=
		L(\tilde{\lambda})
		+
		\left(
			\sum_{\lambda = \tilde{\lambda} + 1}^{\lambda_n}
			n_\lambda
		\right)
		+
		(n-L(\lambda_n))
		\\
		&\leq
		L(\tilde{\lambda})
		+
		\left(
			\sum_{\lambda = \tilde{\lambda} + 1}^{\lambda_n}
			\ell(\lambda) - 3\log(\ell(\lambda+1))
		\right)
		+
		(n-L(\lambda_n))
		\\
		&\leq
		n
		-
		\sum_{\lambda = \tilde{\lambda} + 1}^{\lambda_n}
			3\log(\ell(\lambda+1))
		\\
		&=
		n
		-
		\sum_{\lambda = \tilde{\lambda} + 2}^{\lambda_n+1}
			3\log(\ell(\lambda))
		\\
		&\leq
		n
		-
		3\log\left(
		\sum_{\lambda = \tilde{\lambda} + 2}^{\lambda_n+1}
			\ell(\lambda)
		\right)
		\\&=
		n
		-
		3\log\left(
			L(\lambda_n + 1)
			-
			L(\tilde{\lambda} + 1)
		\right)
	\end{aligned}
	\end{equation}
	Using the facts that $n < L(\lambda_n +1)$ and that if $x \in \N$
	and $y \geq x + 1$ then $\log(y - x) \geq \log(y) - x$, we obtain
	\begin{equation}
	\begin{aligned}
		\abs{a_n}
		&\leq
		n
		-
		3\log(L(\lambda_n + 1))
		+
		3
		L(\tilde{\lambda} + 1)
		\\&\leq
		n - 3\log(n) + 3L(\tilde{\lambda} + 1)
		\\&=
		n - 2\log(n) + 3L(\tilde{\lambda} + 1) - \log(n)
	\end{aligned}
	\end{equation}
	Since $3L(\tilde{\lambda} + 1)$ is a constant, we have that
	\begin{equation}
		\lim_{n \to \infty}
			3L(\tilde{\lambda} + 1) - \log(n)
		=
		-\infty
	\end{equation}
	and so for every constant $c$ there exists an $n_c \in \N$ such that
	$\abs{a_{n_c}} < n_c - 2 \log(n_c) - c$.
\end{proof}

%% file: 8-proofs.tex
\section{Supplementary Proofs}

\subsection{Proof of \cref{th:decision-to-search}}
\label{sc:decision-to-search}

Before proving \cref{th:decision-to-search}, we highlight a small
technical fact. If a subset $\mc{S}$ of a finite set $\mc{X}$ is
constructed by iterating over all elements $x \in \mc{X}$ and including
this element in $\mc{S}$ with probability $p_x$, independently of all
other choices, then the expected size of $\mc{S}$ is
$\sum_{x \in \mc{X}} p_x$.

\begin{fact}
\label{th:expected-set-cardinality}
	Let $\mc{X}$ be a finite set and, for every $x \in \mc{X}$, let
	$p_x \in [0,1]$ be a real number. Let $\mc{S}$ be a random variable
	distributed on $\mc{P}(\mc{X})$ such that, for all
	$\mc{Y} \in \mc{P}(\mc{X})$
	\begin{equation}
		\Pr\left[\mc{S} = \mc{Y}\right]
		=
		\left(\prod_{x \in \mc{Y}} p_x\right)
		\cdot
		\left(\prod_{x \in \mc{Y} \setminus \mc{S}} 1 - p_x\right).
	\end{equation}
	Then, $
		\E \abs{\mc{S}}
		=
		\sum_{x \in \mc{X}} p_x
	$.
\end{fact}

\begin{proof}
	For any $\mc{Y} \in \frak{\mc{X}}$, let $\delta_\mc{Y} : \mc{X} \to
	\{0,1\}$ be the characteristic function for the set $\mc{Y}$. We can
	then compute $\E\abs{\mc{S}}$ as
	\begin{equation}
		\sum_{\mc{Y} \in \mc{P}\left(\mc{X}\right)}
		\abs{\mc{Y}} \cdot \Pr\left[\mc{S} = \mc{Y}\right]
		=
		\sum_{\mc{Y} \in \mc{P}\left(\mc{X}\right)}
		\sum_{x \in \mc{X}}
		\delta_\mc{Y}(x)
		\cdot
		\Pr\left[\mc{S} = \mc{Y}\right]
		=
		\sum_{x \in \mc{X}}
		\sum_{\mc{Y} \in \mc{P}\left(\mc{X}\right)}
		\delta_\mc{Y}(x)
		\cdot
		\Pr\left[\mc{S} = \mc{Y}\right]
	\end{equation}
	which is precisely $\sum_{x \in \mc{X}} p_x$, the desired result, as
	$\sum_{\mc{Y} \in \mc{P}\left(\mc{X}\right)}
		\delta_\mc{Y}(x)
		\cdot
		\Pr\left[\mc{S} = \mc{Y}\right]
		=
		p_x$.
\end{proof}

We can now prove the lemma.

\begin{proof}[Proof of \cref{th:decision-to-search}.]
	Let $\mc{S}$ be the random variable distributed on
	$\mc{P}\left(\{0,1\}^n\right)$ which models the contents of the set
	maintained by the circuit $\tilde{C}$ once it terminates step $2$.
	Note that
	\begin{equation}
		\Pr\left[\mc{S} = \mc{X}\right]
		=
		\left(\prod_{x \in \mc{X}} \bra{1} C(x) \ket{1}\right)
		\left(
			\prod_{x \in \{0,1\}^n \setminus \mc{X}}
			1 - \bra{1} C_x(x) \ket{1}
		\right).
	\end{equation}
	We see that
	\begin{equation}
	\label{pr:s2d-1}
		\bra{s} \tilde{C}(s) \ket{s}
		\geq
		\sum_{\substack{
			\mc{X}
			\in
			\mc{P}\left(\{0,1\}^n\right)
			\\
			s \in \mc{X}
			}
		}
			\frac{1}{\abs{\mc{X}}}
			\cdot
			\Pr\left[\mc{S} = \mc{X}\right]
	\end{equation}
	where the inequality is obtained by neglecting the case where
	$\mc{S}$ is empty.

	Next, let $\mc{S}'$ be the random variable distributed on
	$\mc{P}\left(\{0,1\}^n \setminus \{s\}\right)$ such that
	\begin{equation}
		\Pr\left[\mc{S}' = \mc{X}\right]
		=
		\left(
			\prod_{x \in \mc{X}} \bra{1} C(x) \ket{1}
		\right)
		\left(
			\prod_{
				x\in\left(\{0,1\}^n\setminus\{s\}\right)\setminus\mc{X}
			}
			1 - \bra{1} C_x(x) \ket{1}
		\right)
\end{equation}
	and note that if $s \in \mc{X}$ then $\Pr[\mc{S} =
	\mc{X}] = \bra{1} C(s) \ket{1} \cdot \Pr[\mc{S}' = \mc{X} \setminus
	\{s\}]$ for any $\mc{X} \in \mc{P}(\{0,1\}^n)$. We then have that
	\begin{equation}
	\begin{aligned}
		\sum_{\substack{
			\mc{X}
			\in
			\mc{P}\left(\{0,1\}^n\right)
			\\
			s \in \mc{X}
			}
		}
			\frac{1}{\abs{\mc{X}}}
			\cdot
			\Pr\left[\mc{S} = \mc{X}\right]
		&=
		\sum_{\substack{
			\mc{X}
			\in
			\mc{P}\left(\{0,1\}^n\right)
			\\
			s \in \mc{X}
			}
		}
			\frac{
				1
			}{
				1 + \abs{\mc{X}\setminus \{s\}}
			}
			\cdot
			\Pr\left[\mc{S} = \mc{X}\right]
		\\&=
		\sum_{\substack{
			\mc{X}
			\in
			\mc{P}\left(\{0,1\}^n\right)
			\\
			s \in \mc{X}
		}
		}
			\frac{
				\bra{s} C(s) \ket{s}
			}{
				1 + \abs{\mc{X}\setminus \{s\}}
			}
			\cdot
			\Pr\left[\mc{S}' = \mc{X} \setminus \{s\}\right]
		\\&=
		\sum_{
			\mc{X}'
			\in
			\frak{P}\left(\{0,1\}^n \setminus \{s\}\right)
		}
			\frac{
				\bra{s}C(s)\ket{s}
			}{
				1 + \abs{\mc{X}'}
			}
			\cdot
			\Pr\left[\mc{S}' = \mc{X}'\right]
		\\&=
		\bra{s}C(s)\ket{s}
		\cdot
		\E \frac{1}{1 + \abs{\mc{S}'}}
		\\&\geq
		\bra{s}C(s)\ket{s}
		\cdot
		\frac{1}{1 + \E\abs{\mc{S}'}}
		\\&=
		\bra{s}C(s)\ket{s}
		\cdot
		\frac{1}{
			1
			+
			\sum_{x \in \{0,1\}^n \setminus \{s\}}
			\bra{1}C(x)\ket{1}
		}
	\end{aligned}
	\end{equation}
	where the inequality is obtained by Jensen's inequality and the
	expectation is evaluated by \cref{th:expected-set-cardinality}.
	Combining this with \cref{pr:s2d-1} yields the desired result.
\end{proof}

%% file: 9-bibliography.bbl
\newcommand{\etalchar}[1]{$^{#1}$}
\makeatletter\@ifundefined{url}{\newcommand{\url}[1]{\texttt{#1}}}{}\@ifundefined{href}{\newcommand{\href}[2]{\texttt{#2}}}{}\@ifundefined{mathbb}{\newcommand{\mathbb}[1]{#1}}{}\makeatother